\documentclass[11pt]{article}
\usepackage{amsmath,amsthm,latexsym,amssymb,amsfonts,epsfig,psfrag,color}
\usepackage{graphicx}
\usepackage{url}
\usepackage{caption}
\usepackage{subcaption}
\usepackage{sidecap}
\usepackage{wrapfig}
\usepackage[small,nohug,heads=vee]{diagrams}
\diagramstyle[labelstyle=\scriptstyle]

\makeatletter \@addtoreset{equation}{section} \makeatother

\newtheorem{Theorem}{Theorem}[section]
\newtheorem{Definition}{Definition}[section]

\newtheorem{Corollary}{Corollary}[section]

\newtheorem{Example}{Example}[section]

\def\be{\begin{equation}}
\def\ee{\end{equation}}
\def\ba{\begin{eqnarray}}
\def\ea{\end{eqnarray}}

\newcommand\nn{\nonumber}
\newcommand\q{\quad}
\def\Nl{{\mathchoice
{\setbox0=\hbox{$\displaystyle\rm N$}\hbox{\hbox to0pt
{\kern0.4\wd0\vrule height0.9\ht0\hss}\box0}}
{\setbox0=\hbox{$\textstyle\rm N$}\hbox{\hbox to0pt
{\kern0.4\wd0\vrule height0.9\ht0\hss}\box0}}
{\setbox0=\hbox{$\scriptstyle\rm N$}\hbox{\hbox to0pt
{\kern0.4\wd0\vrule height0.9\ht0\hss}\box0}}
{\setbox0=\hbox{$\scriptscriptstyle\rm N$}\hbox{\hbox to0pt
{\kern0.4\wd0\vrule height0.9\ht0\hss}\box0}}}}
\def\Zl{{\mathchoice
{\setbox0=\hbox{$\displaystyle\rm Z$}\hbox{\hbox to0pt
{\kern0.4\wd0\vrule height0.9\ht0\hss}\box0}}
{\setbox0=\hbox{$\textstyle\rm Z$}\hbox{\hbox to0pt
{\kern0.4\wd0\vrule height0.9\ht0\hss}\box0}}
{\setbox0=\hbox{$\scriptstyle\rm Z$}\hbox{\hbox to0pt
{\kern0.4\wd0\vrule height0.9\ht0\hss}\box0}}
{\setbox0=\hbox{$\scriptscriptstyle\rm Z$}\hbox{\hbox to0pt
{\kern0.4\wd0\vrule height0.9\ht0\hss}\box0}}}}
\def\Ql{{\mathchoice
{\setbox0=\hbox{$\displaystyle\rm Q$}\hbox{\hbox to0pt
{\kern0.4\wd0\vrule height0.9\ht0\hss}\box0}}
{\setbox0=\hbox{$\textstyle\rm Q$}\hbox{\hbox to0pt
{\kern0.4\wd0\vrule height0.9\ht0\hss}\box0}}
{\setbox0=\hbox{$\scriptstyle\rm Q$}\hbox{\hbox to0pt
{\kern0.4\wd0\vrule height0.9\ht0\hss}\box0}}
{\setbox0=\hbox{$\scriptscriptstyle\rm Q$}\hbox{\hbox to0pt
{\kern0.4\wd0\vrule height0.9\ht0\hss}\box0}}}}
\def\Rl{{\mathchoice
{\setbox0=\hbox{$\displaystyle\rm R$}\hbox{\hbox to0pt
{\kern0.4\wd0\vrule height0.9\ht0\hss}\box0}}
{\setbox0=\hbox{$\textstyle\rm R$}\hbox{\hbox to0pt
{\kern0.4\wd0\vrule height0.9\ht0\hss}\box0}}
{\setbox0=\hbox{$\scriptstyle\rm R$}\hbox{\hbox to0pt
{\kern0.4\wd0\vrule height0.9\ht0\hss}\box0}}
{\setbox0=\hbox{$\scriptscriptstyle\rm R$}\hbox{\hbox to0pt
{\kern0.4\wd0\vrule height0.9\ht0\hss}\box0}}}}
\def\Cl{{\mathchoice
{\setbox0=\hbox{$\displaystyle\rm C$}\hbox{\hbox to0pt
{\kern0.4\wd0\vrule height0.9\ht0\hss}\box0}}
{\setbox0=\hbox{$\textstyle\rm C$}\hbox{\hbox to0pt
{\kern0.4\wd0\vrule height0.9\ht0\hss}\box0}}
{\setbox0=\hbox{$\scriptstyle\rm C$}\hbox{\hbox to0pt
{\kern0.4\wd0\vrule height0.9\ht0\hss}\box0}}
{\setbox0=\hbox{$\scriptscriptstyle\rm C$}\hbox{\hbox to0pt
{\kern0.4\wd0\vrule height0.9\ht0\hss}\box0}}}}
\def\Hl{{\mathchoice
{\setbox0=\hbox{$\displaystyle\rm H$}\hbox{\hbox to0pt
{\kern0.4\wd0\vrule height0.9\ht0\hss}\box0}}
{\setbox0=\hbox{$\textstyle\rm H$}\hbox{\hbox to0pt
{\kern0.4\wd0\vrule height0.9\ht0\hss}\box0}}
{\setbox0=\hbox{$\scriptstyle\rm H$}\hbox{\hbox to0pt
{\kern0.4\wd0\vrule height0.9\ht0\hss}\box0}}
{\setbox0=\hbox{$\scriptscriptstyle\rm H$}\hbox{\hbox to0pt
{\kern0.4\wd0\vrule height0.9\ht0\hss}\box0}}}}
\def\Ol{{\mathchoice
{\setbox0=\hbox{$\displaystyle\rm O$}\hbox{\hbox to0pt
{\kern0.4\wd0\vrule height0.9\ht0\hss}\box0}}
{\setbox0=\hbox{$\textstyle\rm O$}\hbox{\hbox to0pt
{\kern0.4\wd0\vrule height0.9\ht0\hss}\box0}}
{\setbox0=\hbox{$\scriptstyle\rm O$}\hbox{\hbox to0pt
{\kern0.4\wd0\vrule height0.9\ht0\hss}\box0}}
{\setbox0=\hbox{$\scriptscriptstyle\rm O$}\hbox{\hbox to0pt
{\kern0.4\wd0\vrule height0.9\ht0\hss}\box0}}}}
\newcommand{\cc}{\mathcal C}

\newcommand{\cf}{\mathcal F}

\newcommand{\ch}{\mathcal H}

\newcommand{\ck}{\mathcal K}
\newcommand{\cl}{\mathcal L}

\newcommand{\cp}{\mathcal P}
\newcommand{\cq}{\mathcal Q}

\newcommand{\ct}{\mathcal T}

\newcommand{\fh}{\mathfrak{h}}

\def\nn{\nonumber}

\newcommand{\eqa}{\begin{eqnarray}}
\newcommand{\neqa}{\end{eqnarray}}

\def\w{\wedge}

\newcommand{\1}{$\{10j\}$}

\newcommand{\p}{\partial}

\def\f{\frac}

\usepackage{bbm}

\def\q{{\quad}}

\definecolor{bianca}{rgb}{0,0.,0.8}

\definecolor{bianca2}{rgb}{0,0.5,0.8}

\definecolor{phil}{rgb}{0.9,0.,0}

\begin{document}

{\renewcommand{\thefootnote}{\fnsymbol{footnote}}

\title{Constraint analysis for variational discrete systems}
\author{Bianca Dittrich\footnote{e-mail address: {\tt bdittrich@perimeterinstitute.ca}}  $^{1}$ and Philipp A H\"ohn\footnote{e-mail address: {\tt phoehn@perimeterinstitute.ca}}  $^{1,2}$\\
 \small $^1$ Perimeter Institute for Theoretical Physics,\\
 \small 31 Caroline Street North, Waterloo, Ontario, Canada N2L 2Y5\\
 \small   $^2$ Institute for Theoretical Physics, Universiteit Utrecht,\\
 \small  Leuvenlaan 4, NL-3584 CE Utrecht, The Netherlands 
}

\date{\small PI--QG--323, ITP--UU--13/05, SPIN--13/03}

}

\setcounter{footnote}{0}
\maketitle

\vspace{-.9cm}

\begin{abstract}
A canonical formalism and constraint analysis for discrete systems subject to a variational action principle are devised. The formalism is equivalent to the covariant formulation, encompasses global and local discrete time evolution moves and naturally incorporates both constant and evolving phase spaces, the latter of which is necessary for a time varying discretization. The different roles of constraints in the discrete and the conditions under which they are first or second class and/or symmetry generators are clarified. 
The (non--) preservation of constraints and the symplectic structure is discussed; on evolving phase spaces the number of constraints at a fixed time step depends on the initial and final time step of evolution. 
Moreover, the definition of observables and a reduced phase space is provided; again, on evolving phase spaces the notion of an observable as a propagating degree of freedom requires specification of an initial and final step and crucially depends on this choice, in contrast to the continuum. 
However, upon restriction to translation invariant systems, one regains the usual time step independence of canonical concepts. This analysis applies, e.g., to discrete mechanics, lattice field theory, quantum gravity models and numerical analysis.
\end{abstract}

\section{Introduction}

Discretizations of dynamical systems assume a pivotal role in both numerical simulations of physical models  and the non--perturbative quantization of field theories. Quantum gravity is a prominent example of the latter in which discrete structures appear in many different ways. 

Usually, it is space--time itself which is discretized, for instance, by a triangulation. In gravity space--time is a dynamical entity, such that one may expect some aspects of this discretization also to be dynamical, i.e.\ to change in time. For instance, different hypersurfaces in a generic space--time triangulation in Regge Calculus \cite{Regge:1961px}---the most well--known discretization of General Relativity---are comprised of different numbers of simplices \cite{Dittrich:2011ke}. In particular, one may wish to model an expanding/contracting universe with a time dependent discretization, in which the number of degrees of freedom grows/shrinks with the size of the spatial geometry \cite{Jacobson:1999zk,Foster:2004yc}. In the quantum theory, this leads to the notion of `evolving Hilbert spaces' \cite{Doldan:1994yg}. In the numerical analysis of physical systems, on the other hand, it is often convenient or (for computational savings) even necessary to adapt either to a more refined or coarser discretization during the evolution of the system. This, in particular, has culminated in the standard tool of `adaptive mesh refinement' \cite{adapmeshref}.

In such situations one has to cope with discrete systems whose numbers of degrees of freedom change in time. In order to systematically treat such scenarios, first of all at the classical level, the present work shall develop a general canonical formalism applicable to arbitrary discrete systems governed by a variational action principle. In particular, this formalism applies to discrete mechanics, lattice field theories and certain discrete gravity models. Its general features are the following:
\begin{itemize}
\item It can handle both constant and evolving phase spaces, the latter by natural phase space extensions. 
\item It is equivalent to the covariant formalism. By using the action (or Hamilton's principal function) as a generating function, the canonical formalism is directly derived from the covariant formulation. Furthermore, ensuring that all constraints (on the extended phase spaces) are always satisfied is tantamount to solving the (covariant) equations of motion.
\item It is insensitive to the particular discretization and form of the (effective) action. Hence, it is amenable to coarse graining methods (which are relevant for the continuum limit).
\end{itemize}

We shall see that systems with changing phase space dimension necessarily lead to constraints. It is essential to distinguish the latter type of constraints from symmetry generators, in particular, in theories of gravity. 
In continuum (Einstein) gravity, the diffeomorphism symmetry leads to the Hamiltonian and spatial diffeomorphism constraints, which also generate (infinitesimal) time evolution. The situation in the discrete is, however, very different because diffeomorphism symmetry is generically broken \cite{Dittrich:2008pw,Bahr:2009ku,Dittrich:2012qb,Bahr:2011xs} such that Hamiltonian and diffeomorphism constraints (generically) do {\it not} arise. Rather, equivalence of the canonical and covariant picture requires canonical time evolution to proceed in discrete steps \cite{Dittrich:2011ke}. Nonetheless, a changing phase space dimension will always lead to constraints, which only in special situations include or even coincide with the Hamiltonian and diffeomorphism constraints.

A proper understanding of the role of the constraints is thus important for the interpretation of both the classical and quantum theories based on such discretizations \cite{Jacobson:1999zk,Foster:2004yc,Unruh:1993js}. In this article, we shall provide the necessary constraint analysis for variational discrete systems which in some aspects is analogous to and in others quite dissimilar from the continuum Dirac procedure \cite{Dirac,Henneaux:1992ig}. In the continuum, the roles of the constraints pertain, in particular, to the identification of gauge symmetries and propagating degrees of freedom, as well as observables. On evolving phase spaces, on the other hand, one has to reconsider the very concepts of propagating degree of freedom and observable; as we shall see, their definition, in contrast to the continuum, cannot be associated to one particular instant of time only. 

The constraint analysis for variational discrete systems has the following main results:
\begin{itemize}
\item The (non--) preservation of constraints and the symplectic structure on evolving (and constant) phase spaces is clarified. In general, the number of constraints at a fixed time step critically depends on the initial and final steps of evolution because equations of motion can lead to a `propagation' of constraints.
\item The conditions under which the canonical constraints are symmetry generators and first or second class are investigated. 
\item The concept of observables as gauge invariant and propagating degrees of freedom requires {\it two} time steps on an evolving phase space and crucially depends on these initial and final steps under consideration. For two different pairs of steps one, in general, finds a different number of propagating degrees of freedom.
\item The reduced phase space at some step $n$ in an evolution $i\rightarrow n\rightarrow f$ depends on initial and final steps $i$ and $f$. 
The classification of degrees of freedom is, in general, step dependent.
\end{itemize}
The step dependence of many canonical concepts notwithstanding,\footnote{Upon restriction to translation invariant systems this step dependence disappears.} the new formalism is fully consistent, solves the covariant equations of motion and unambiguously describes the propagation of data between different time steps. It is therefore applicable to evolving lattices and offers a comprehensive picture of the discrete dynamics. 

For example, as shall be explained in the course of this article, systems which have been evolved from a phase space with zero degrees of freedom (i.e.\ from `nothing' as in the `no boundary proposal' \cite{Hartle:1983ai}) are always totally constrained, in the sense that the dimension of the reduced phase space on the evolving `time slice' is zero. On account of the step dependence of observables, however, this does not mean that these systems are necessarily devoid of propagating degrees of freedom. 

These results may prove useful for numerical implementations in discrete mechanics and lattice field theory and shed light on the interpretation of quantum transition amplitudes, e.g., in loop quantum gravity \cite{Perez:2001gja,Perez:2012wv,Speziale:2008uw}. Moreover, they further the interpretation of the newly developed canonical formulation of Regge Calculus \cite{Dittrich:2011ke,Dittrich:2009fb} (for a brief summary of this formalism, see also \cite{Hoehn:2011cm}). In particular, both loop quantum gravity and Regge Calculus also involve discretizations (graphs), encoding different numbers of degrees of freedom for the initial and final state. In fact, the present formalism shares certain features with the `general boundary formulation' \cite{Oeckl:2003vu,Oeckl:2005bv}, where transition amplitudes are defined between very different types of surfaces (including `zero--size' surfaces) and the traditional splitting of the space--time manifold into ${\cal M} =\Sigma \times \Rl$ is not assumed.
\\

The rest of this manuscript is organised as follows: in section \ref{sec_Hfct}, we shall briefly introduce variational discrete systems and, subsequently, discuss the global dynamics of such systems in the regular and irregular case in section \ref{sec_global} (for a detailed exposition of the regular case see \cite{marsdenwest}). The irregular setting encompasses systems with evolving numbers of degrees of freedom. In particular, section \ref{sec_global} introduces the Legendre transforms and the canonical formalism for discrete systems. We shall provide sufficiently many details to make the present article relatively self--contained.


Section \ref{sec_local} concerns local evolution moves in which only a (small) subset of an equal time hypersurface is evolved. A typical example are the Pachner moves which constitute the most elementary evolution steps for a triangulated equal time hypersurface. For Regge Calculus, these have earlier been introduced as canonical evolution moves in \cite{Dittrich:2011ke}, below we shall generalize this concept to general variational discrete systems. Section \ref{sec_conrole} analyzes the different roles constraints can assume in the discrete and specifies in which cases they constitute gauge symmetry generators. Moreover, we generalize the concept of observables and propagating degrees of freedom to discrete systems with varying numbers of degrees of freedom. We shall identify such observables and provide a counting of the number of propagating degrees of freedom. Finally, this article will close with a summary and discussion in section \ref{conclusions}.

Along the way, we shall illustrate general features by means of a scalar field living on a 2D triangulation. Some technical details have been moved to the appendices.

\section{Preliminaries: variational discrete systems}\label{sec_Hfct}

We shall consider general variational discrete 
systems whose evolution proceeds in discrete time steps which we label by $n\in\mathbb{Z}$. 

Let $\cq_n$ be the configuration space of the system at time step $n$. ${\cal Q}_n$ shall be coordinatized by $x_n^i$, where $i$ takes value in some index set determined by the dimension of ${\cal Q}_n$. (We shall often omit the index $i$ for notational convenience.) We expressly emphasize that the $\cq_n$ at the various steps $n$ need not be of the same dimension. The dynamics of such systems shall be described by discrete actions of the form
\ba\label{genact}
S_K=\sum^K_{n=1}S_n(x_{n-1},x_{n})\,,
\ea
where the sum ranges over the individual time steps $n$. The individual action contribution $S_n$ governs the {\it discrete time evolution move} $(n-1)\rightarrow n$, i.e.\ the discrete time evolution from time step $n-1$ to time step $n$. These systems are 
\begin{itemize}
\item[(i)] variational because the configuration spaces $\cq_n$ are to be continuous manifolds such that the dynamics is gained from a variational action principle,
\item[(ii)] discrete because the time evolution proceeds in discrete steps labeled by $n$, and, finally,
\item[(iii)] mechanical because $\cq_n$ shall be finite dimensional. 
\end{itemize}

The main assumption we shall henceforth make is that of {\it additivity} of the discrete action, in order for the sum (\ref{genact}) to make sense. That is, the action contribution of the union of any two evolution moves should simply be the sum of the individual action contributions of the two moves. In particular, if the two moves $(n-1)\rightarrow n$ and $n\rightarrow (n+1)$ share a `common boundary' described by the variables $x_{n}$, we can solve the equations of motion
\ba
\f{\p S_n(x_{n-1},x_n)}{\p x_n}+\f{\p S_{n+1}(x_n,x_{n+1})}{\p x_n}=0
\ea 
for $x_n=\chi_n(x_{n-1},x_{n+1})$ in order to obtain a new `effective' action contribution
\ba
\tilde{S}_{(n-1)\rightarrow(n+1)}(x_{n-1},x_{n+1}):=S_n\left(x_{n-1},\chi_n(x_{n-1},x_{n+1})\right)+S_{n+1}\left(\chi_n(x_{n-1},x_{n+1}),x_{n+1}\right)\,,
\ea
governing the `effective evolution move' $(n-1)\rightarrow(n+1)$. The action evaluated on solutions, given suitable boundary data, is called Hamilton's principal function \cite{marsdenwest,marsdenratiu,Dittrich:2011ke,Gambini:2002wn,DiBartolo:2004cg,Rovelli:2011mf} and is a function of this boundary data. In our case Hamilton's principal function $\tilde{S}_{(n-1)\rightarrow(n+1)}$ is a function of $x_{n-1}$ and $x_{n+1}$ coordinatizing the initial and final boundary, respectively, and defines a new contribution in the action sum (\ref{genact}) that summarizes two elementary ones. That is, for full generality we allow for the two possibilities that a given individual contribution $S_n$ in (\ref{genact}) may either be an elementary (`bare') action or Hamilton's principal function that summarized several elementary evolution moves into a single effective one.


Such variational discrete systems encompass a large variety of discretized physical systems, ranging from discrete mechanics \cite{marsdenwest}, over lattice field theory \cite{smitbook} to discrete gravity. In particular, Regge Calculus falls into this class of discrete mechanical systems \cite{Dittrich:2011ke,Dittrich:2009fb}: the lengths of the edges of a Regge triangulation can take continuous values, in a finite triangulation there are finitely many edges and the discrete time evolution proceeds by discrete evolution moves (e.g.\ Pachner \cite{Dittrich:2011ke} or tent moves \cite{Dittrich:2009fb}). Furthermore, due to boundary terms, the Regge action {is} additive in the above sense \cite{Hartle:1981cf,Dittrich:2011ke}.\footnote{The Regge action is given by $S_{\rm Regge}=\sum_{h\subset\ct\setminus\p\ct}V_h\varepsilon_h+\sum_{h\subset\partial\ct}V_h\psi_h$, where $h$ denotes the hinges, i.e.\ $(D-2)$--dimensional subsimplices in the $D$--dimensional triangulation $\ct$, $V_h$ denotes the volume of $h$, $\varepsilon_h$ denotes the deficit angle around hinge $h$ in the bulk $\ct\setminus\p\ct$ of the triangulation and $\psi_h$ denotes the extrinsic curvature angle around hinge $h$ in the boundary $\p\ct$ of the triangulation. When two pieces of triangulation $\ct_1$ and $\ct_2$ are glued together, (parts of) their boundaries and, in particular, the hinges in the boundaries are identified. The two extrinsic curvature angles around a fixed hinge $h$ in both boundaries of $\ct_1$ and $\ct_2$ add up to a deficit angle, $\psi_{h\subset\p\ct_1}+\psi_{h\subset\p\ct_2}=\varepsilon_{h\subset(\ct_1\cup\ct_2)\setminus\p(\ct_1\cup\ct_2)}$, in the bulk of $\ct_1\cup\ct_2$. The Regge action is therefore additive. On the other hand, a discrete action where $\varepsilon_h$ is replaced by $\sin\varepsilon_h$ constitutes an example of a non-additive action.}

Most research on such systems has been performed in the Lagrangian setting. It is the goal of the present work to develop a canonical formalism for such general variational discrete systems which is equivalent to the Lagrangian formulation. This is achieved by using the action (or Hamilton's principal function) $S_n=S_n(x_{n-1},x_n)$, which depends on `old' variables $x_{n-1}$ and `new' variables $x_{n}$, as a generating function of the first kind for canonical time evolution \cite{marsdenwest,marsdenratiu,Gambini:2002wn,DiBartolo:2004cg,Dittrich:2011ke,Bahr:2009ku,Dittrich:2009fb}. The new canonical formalism will be applicable to general systems with evolving phase spaces and, in particular, provide a constraint analysis for variational discrete systems in analogy to Dirac's original continuum procedure \cite{Dirac,Henneaux:1992ig}. Finally, given that the formalism naturally handles Hamilton's principal functions, different  
choices for the underlying discrete lattice can be related to each other 
 and the formalism is naturally applicable to coarse graining methods (which, e.g., are relevant for the continuum limit), see for instance \cite{Bahr:2009qc,Bahr:2010cq,Dittrich:2012jq}.

%
\begin{Example}
\emph{An example which we will use throughout this paper to illustrate the general principles will be the 2D discretized massless scalar field on a triangulation with Euclidean geometry. To simplify the action we shall consider an equilateral triangulation. 
The scalar field $\phi$ is associated to the vertices of the lattice and $\phi^v$ shall denote the field at the vertex $v$.}

\emph{The action associated to one equilateral triangle is given by \cite{Sorkin:1975jz}
\ba\label{saction}
S_{\Delta}&=& \frac{1}{2} \left(\sum_{v\subset \Delta} (\phi^v)^2 - \sum_{e\subset \Delta} \phi^{s(e)}\phi^{t(e)} \right) \,=\,\frac{1}{4} \sum_{e\subset \Delta} (\phi^{s(e)}-\phi^{t(e)})^2,
\ea
where $v,e$ denote the vertices and edges of the triangle $\Delta$ and $s(e),t(e)$ are the source and target vertex of the edge $e$. The action associated to any larger triangulation is given by the sum of the actions associated to the triangles. This automatically includes the right boundary terms, if the fields on the boundary are held fixed.}
\end{Example}

\section{Global dynamics of variational discrete systems}\label{sec_global}

In this section we shall only consider {\it global} time evolution moves. Global moves are such that each of the variables at a given discrete time step is involved in the move and only occurs at this one time step, i.e.\ neighbouring time steps $n,n+1$ do not overlap, except for variables in a possible boundary of a fixed time step. For example, evolution moves in simplicial gravity which evolve between disjoint spatial hypersurfaces are global \cite{Dittrich:2011ke}. {\it Local} evolution moves which only evolve subsets of data of a given time step will be discussed in the subsequent section \ref{sec_local}.

\subsection{General properties and definitions}\label{sec_gen}

Consider three consecutive steps $n-1,n,n+1$ and the boundary value problem defined by the data at times $n-1$ and $n+1$. That is, we are given boundary data $x_{n-1}^i$ and $x_{n+1}^i$ and ought to extremize
\ba\label{app1}
S:=S_n(x_{n-1},x_n)+S_{n+1}(x_n,x_{n+1})
\ea
with respect to $x_n$. This yields the equations of motion
\ba\label{app2}
0&=&\frac{\partial S_n}{\p x_n}+\frac{\p S_{n+1}}{\p x_n}\,,
\ea
which may or may not be uniquely solvable for $x_n$ as a function of $x_{n-1},x_{n+1}$, depending on whether the system under consideration is regular or irregular (see sections \ref{sec_reg} and \ref{sec_irreg} below). An initial value problem can be treated in complete analogy by computing $x_{n+1}$ from $x_{n-1},x_n$ via (\ref{app2}).

\subsubsection{Lagrangian formulation}\label{sec_lagr}

The action contribution $S_n$ defines a mapping $S_n:\cq_{n-1}\times\cq_n\rightarrow\mathbb{R}$. On the other hand, in continuum mechanics, the Lagrangian $L(q,\dot{q})$ defines a mapping $L:T\cq\rightarrow\mathbb{R}$. That is, in the discrete the direct product of configuration manifolds $\cq_{n-1}\times\cq_n$---coordinatized by $x_{n-1},x_n$---assumes the role of the tangent bundle $T\cq$---coordinatized by $q,\dot{q}$---of the Lagrangian formulation of the continuum. In fact, if $\cq_{n-1}\cong\cq_n\cong\cq$, $\cq\times\cq$ is locally isomorphic to $T\cq$.

The variation of the discrete action (\ref{app1}) enables one to define the so--called Lagrangian one-- and two--forms \cite{marsdenwest}. In contrast to the continuum where only one Lagrangian one--form $\theta$ exists on $T\cq$, in the discrete {\it two} Lagrange one--forms on ${\cal Q}_n\times {\cal Q}_{n+1}$ arise from the boundary terms of the variation of the action.
Namely, varying $S_{n+1}$ as $\delta S_{n+1}=dS_{n+1}\cdot\delta_{q_{n(n+1)}}$ with $q_{n(n+1)}\in\cq_n\times\cq_{n+1}$ and some variation $\delta_{q_{n(n+1)}}\in T_{q_{n(n+1)}}\left(\cq_n\times\cq_{n+1}\right)$ yields 
\ba
dS_{n+1}=\theta_{n+1}^+-\theta_n^-\,,
\ea
where (summing over repeated indices $i,j$ is understood)
\ba\label{app5}
\theta^-_n(x_n,x_{n+1}) &=& -\frac{\p S_{n+1}}{\p x_n^i} dx_n^i \nn\\
\theta^+_{n+1}(x_n,x_{n+1}) &=& \frac{\p S_{n+1}}{\p x_{n+1}^i} dx_{n+1}^i\,.
\ea
However, $d\circ d S_{n+1}=0$, such that a single Lagrange two--form on ${\cal Q}_n\times {\cal Q}_{n+1}$ can be defined:
\ba\label{app6}
\Omega_{n+1}(x_n,x_{n+1})\,=\,-d\theta^+_{n+1}\,=\,-d\theta^-_n&=& -\frac{\p^{2} S_{n+1}}{\p x_n^i \p x_{n+1}^j} \, dx_n^i \, \wedge dx_{n+1}^j  \, .
\ea

The equations of motion (\ref{app2}) can be used to define a Lagrangian time evolution map $\cl_n:{\cal Q}_{n-1}\times {\cal Q}_n\rightarrow{\cal Q}_n\times {\cal Q}_{n+1}$
\ba\label{app4}
{\cal L}_n: (x_{n-1},x_n) \mapsto (x_n,x_{n+1}),
\ea
by solving for $x_{n+1}$ given $(x_{n-1},x_n)$. $\cl_n$ may neither be defined on all of ${\cal Q}_{n-1}\times {\cal Q}_n$, nor map to all of ${\cal Q}_n\times {\cal Q}_{n+1}$, nor be unique in the presence of constraints. This shall be a subject of the following subsections \ref{sec_reg} and \ref{sec_irreg}.


\subsubsection{Discrete Legendre transformations and canonical formulation}\label{sec_globcan}

In order to discuss the dynamics in a canonical language, we need to introduce {\it discrete Legendre transformations} which will carry us to suitable phase spaces. Recall that in the continuum formulation a single Legendre transformation $\mathbb{F}L:T\cq\rightarrow T^*\cq$ exists, which in coordinates reads $(q,\dot{q})\mapsto (q,\frac{\p L}{\p\dot{q}})$. On the other hand, two Legendre transformations can be defined in the discrete because we work with a cartesian product of two configuration manifolds $\cq_n\times\cq_{n+1}$, instead of $T\cq$. The precise definition of these discrete Legendre transformations which we will employ is motivated by, however, suitably differs from the continuum version (for details on the continuum Legendre transform see \cite{marsdenratiu}). 


Consider an arbitrary step $n$. Recall that $S_n:\cq_{n-1}\times\cq_n\rightarrow\mathbb{R}$ where $\cq_{n-1}\times\cq_n$ is a fibre bundle. Pick a point $q_{n-1}\in\cq_{n-1}$ (see figure \ref{fibder}). We denote the fibre over fixed $q_{n-1}$ by $\cf_{n}(q_{n-1}):=(\cq_{n-1}\times\cq_n)_{q_{n-1}}$. Notice that $\cf_{n}\cong\cq_n$. Choose a point $f_n\in\cf_{n}$ and a curve $\gamma(\varepsilon)$ in $\cf_{n}$ with curve parameter $\varepsilon$ such that $\gamma(0)=f_{n}$ and $\gamma'(0)=\frac{d}{d\varepsilon}\gamma(\varepsilon)\left|_{\varepsilon=0}\right.$. This allows us to provide the following

\begin{Definition}\label{def_postleg}
The discrete fibre derivative $\mathbb{F}^+S_n:\cq_{n-1}\times\cq_n\rightarrow T^*\cq_n$, defined by
\ba\label{postleg}
\mathbb{F}^+S_n(f_{n})\cdot\gamma'(0):=\frac{d}{d\varepsilon}S_n\left(\gamma(\varepsilon)\right)\Big|_{\varepsilon=0},
\ea
is called the \emph{post--Legendre transform}.
\end{Definition}

$\mathbb{F}^+S_n(f_{n})\cdot\gamma'(0)$ is the derivative of $S_n$ along the fibre $\cf_{n}$ at $f_{n}$ in the direction $\gamma'(0)$. Given that $\cf_{n}\cong\cq_n$, we have $\gamma'(0)\in T_{f_{n}}\cq_n$ and, thus, $\mathbb{F}^+S_n(f_{n})\in T^*_{f_{n}}\cq_n$.
\begin{figure}[hbt!]
\begin{center}
\psfrag{q}{$q_{n-1}$}
\psfrag{Q}{ $\cq_{n-1}$}
\psfrag{QQ}{\Large$\cq_{n-1}\times\cq_n$}
\psfrag{F}{ $\cf_n(q_{n-1})\cong\cq_n$}
\psfrag{f}{$f_n=\gamma(0)$}
\psfrag{g}{ $\gamma$}
\psfrag{g0}{$\gamma'(0)$}
\includegraphics[scale=.7]{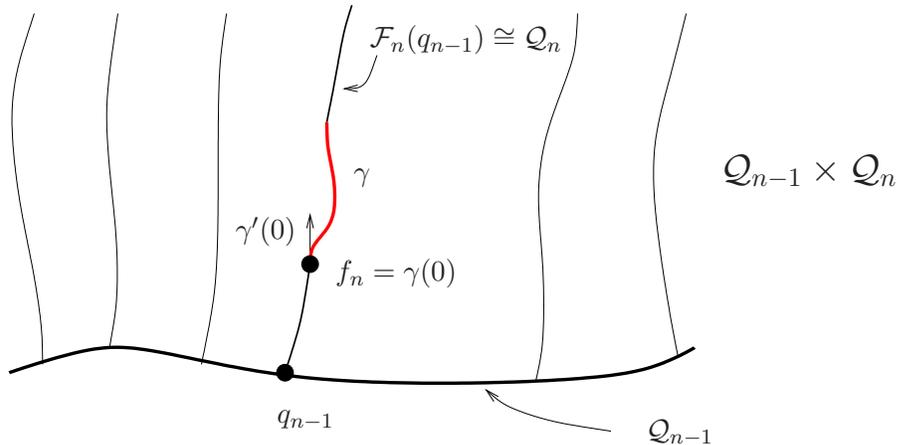}
\caption{\small Illustration of the fibre derivative used in the definition of the {\it post--Legendre transform}.}\label{fibder}
\end{center}
\end{figure}

Now exchange the roles of $\cq_{n-1}$ and $\cq_n$ and choose $f_{n-1}\in\cf_{n-1}(q_n)$. Let $\eta(\nu)$ be a curve in $\cf_{n-1}$ such that $\eta(0)=f_{n-1}$. In complete analogy, we give
\begin{Definition}\label{def_preleg}
The discrete fibre derivative $\mathbb{F}^-S_n:\cq_{n-1}\times\cq_n\rightarrow T^*\cq_{n-1}$, defined by
\ba
\mathbb{F}^-S_n(f_{n-1})\cdot\eta'(0):=-\frac{d}{d\nu}S_n\left(\eta(\nu)\right)\Big|_{\nu=0},\nn
\ea
is called the \emph{pre--Legendre transform}.
\end{Definition}

The cotangent bundles $\cp_{n-1}:=T^*\cq_{n-1}$ and $\cp_n:=T^*\cq_n$ are the phase spaces which we will henceforth work with.

For the remainder of this article, we will need the coordinate form of the pre-- and post--Legendre transform. Consider the post--Legendre transforms. Let $(x_{n-1},x_n)$ be the coordinates of $f_n$ in $\cq_{n-1}\times\cq_n$ and let $\delta x_n$ be the coordinate expression of $\gamma'(0)$. In a small neighbourhood of $f_n$ we can write $\gamma(\epsilon)$ in coordinates as $(x_{n-1},x_n+\epsilon\delta x_n)$. 
Inserting this in $S_n(x_{n-1},x_n)$, yields (\ref{postleg}) in the form
\ba
\frac{d}{d\varepsilon}S_n(x_{n-1},x_n+\varepsilon\delta x_n)\Big|_{\varepsilon=0}=\frac{\p S_n}{\p x_n}\delta x_n.
\ea
Hence, we obtain as in \cite{marsdenwest} (note that $x_n$ are the coordinates of $f_n$ in $\cq_n$),
\ba
\mathbb{F}^+S_n(x_{n-1},x_n)=\left(x_n,\frac{\p S_n}{\p x_n}\right).
\ea
The coordinate expression for the pre--Legendre transform is derived in complete analogy. In general, we write
\ba
&&\mathbb{F}^+S_n:  (x_{n-1},x_n) \mapsto (x_n,{}^{+}p^n)\,=\,\left(x_n,  \frac{\partial S_n}{ \partial x_{n}}\right) \label{app9}\\
&&\mathbb{F}^-S_n:  (x_{n-1},x_n) \mapsto (x_{n-1},{}^{-}p^{n-1})\,=\,\left(x_{n-1},  -\frac{\partial S_n}{ \partial x_{n-1}}\right).\label{app9b}
\ea
We shall refer to (\ref{app9}) as the {\it post--Legendre transformation} and to (\ref{app9b}) as the {\it pre--Legendre transformation}. 

The reason for the minus sign in definition \ref{def_preleg} (and thus of the definition of ${}^-p^{n-1}$ in (\ref{app9b})) is the following: using the coordinate forms (\ref{app9}, \ref{app9b}), it is straightforward to check that---in analogy to the continuum---the Lagrangian one-- and two--forms, (\ref{app5}) and (\ref{app6}), arise from pulling back the canonical one-- and two--forms, $\theta_n=p^n_idx^i_n$ and $ \omega_n=dx^i_n\wedge dp^n_i$, respectively, with the Legendre transformation. That is,\footnote{Clearly, the definition is coordinate independent and so, if a new coordinate system $x'_n(x_n)$ on $\cq_n$ is chosen, one finds $dx'_n\wedge dp'^n=dx_n\wedge dp^n.$}
\ba\label{legpullback}
\theta^+_{n}&=&(\mathbb{F}^+ S_n)^*\theta_n,\q\q\q\q\Omega_n=(\mathbb{F}^+S_n)^*\omega_n,\nn\\
\theta^-_{n}&=&(\mathbb{F}^- S_{n+1})^*\theta_n,\q\q\Omega_{n+1}=(\mathbb{F}^-S_{n+1})^*\omega_n.
\ea


In order to define the Legendre transforms, we employed the discrete action, or Hamilton's principal function. By noting that Hamilton's principal function is a generating function of the first kind (i.e.\ depends on the old and new configuration coordinates), one can define the canonical discrete time evolution via the equations
\ba\label{b3}
{}^{-}p^{n-1}:= -\frac{\partial S_{n}(x_{n-1},x_{n})}{ \partial x_{n-1}}    ,\q \q\q  {}^{+}p^{n}  := \frac{\partial S_n(x_{n-1},x_{n})}{ \partial x_{n}}.
\ea
We shall refer to the momenta ${}^{-}p$ (\ref{app9b}) as {\it pre--momenta} and to the momenta ${}^{+}p$ (\ref{app9}) as {\it post--momenta}. Equation (\ref{b3}) defines an implicit global Hamiltonian time evolution map $\ch_{n-1}:\cp_{n-1}\rightarrow\cp_n$,
\ba\label{app10}
{\cal H}_{n-1}: (x_{n-1},{}^-p^{n-1}) \mapsto (x_n, {}^+p^n ).
\ea
Namely, given $(x_{n-1},{}^-p^{n-1})$, one can use the equation for the pre--momenta in (\ref{b3}), in order to determine $x_n$ and, using this result and the post--momenta equation in (\ref{b3}), one determines ${}^+p^n$. It depends on the presence of constraints whether one can uniquely determine $x_n$ or not. This will be the subject of the following subsections.


We emphasize that the definition of $\ch_n$---in contrast to the one for $\cl_n$---itself does {\it not}{ require} the equations of motion. Nonetheless, the equations of motion are also necessary in the canonical picture for the following reason: In analogy to (\ref{b3}), we could equally well use $S_{n+1}$ as a generating function. Accordingly, at every time step we have both {\it pre--} and {\it post--momenta} (see figure \ref{discevol}):
\ba\label{b4}
{}^{-}p^{n}:= -\frac{\partial S_{n+1}}{ \partial x_{n}}  ,\q \q\q  {}^{+}p^{n}  := \frac{\partial S_{n}}{ \partial x_{n}}   .
\ea
\begin{wrapfigure}{h}{0.4\textwidth}
\psfrag{n}{\small$n$}
\psfrag{n1}{\small$n-1$}
\psfrag{n2}{\small$n+1$}
\psfrag{p1}{${{}^+p^n}$}
\psfrag{p2}{${{}^-p^n}$}
\psfrag{h1}{\small$\ch_{n-1}$}
\psfrag{h2}{\small$\ch_{n}$}
\psfrag{s}{$S_n$}
\psfrag{s2}{$S_{n+1}$}
\vspace{-.5cm}
\hspace*{1cm}\includegraphics[scale=.7]{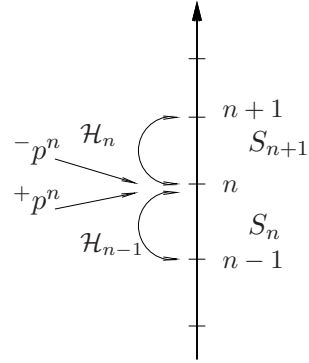}
\caption{\small The implicit global discrete Hamiltonian time evolution map and {\it pre--} and {\it post--momenta}.}\label{discevol}
\end{wrapfigure}
However, the equations of motion $\frac{\partial S_{n}}{ \partial x_{n}} + \frac{\partial S_{n+1}}{ \partial x_{n}}=0$ for the variables $x_n$ are equivalent to a {\it momentum matching} of pre-- and post--momenta,
\ba\label{mm}
{}^+p^n={}^-p^n,
\ea
such that (on--shell) there are unique momenta for the variables at step $n$. Henceforth, we will often omit the superindices $+$ and $-$ at the momenta, implicitly assuming that {\it momentum matching} holds.

In summary, we find that the off--shell (or kinematical) quantities are defined by the action associated to one time evolution move, while equations of motions are imposed by momentum matching. 

\begin{Example}
\emph{Consider the evolution of the massless scalar field with action (\ref{saction}) on a two--dimensional cylindrical space time $S^1\times [0,1]$, where we assign the time $n$ to the lower boundary $S^1\times \{0\}$ and time $(n+1)$ to the upper boundary $S^1 \times \{1\}$. We assume that this space time is triangulated such that we have $Q_n$ vertices at time $n$ and $Q_{n+1}$ vertices at time $(n+1)$ and such that we have only one layer of triangles as depicted in figure \ref{fig_cdt}. This corresponds to a triangulation as used in the Causal Dynamical Triangulation (CDT) programme \cite{Ambjorn:2012jv}; a vertex $i$ at time $n$ can be connected to a vertex $j$ at time $(n+1)$ by one edge. For very `small' triangulations a vertex $i$ may also be connected to $j$ by 2 edges. This defines an adjacency matrix $A^{n+1}_{ij}$ with $A^{n+1}_{ij}=E$ if vertex $i$ is connected with vertex $j$ by $E$ edges, where $E=0,1$ or $2$.  The canonical time evolution  from time $n$ to time $(n+1)$ can then be written as
\ba
{}^-\pi^n_i &=& -\left( \sum_{j=1}^{Q_{n+1}} A^{n+1}_{ij} ( \phi_{n}^i-\phi_{n+1}^j)\right)  \;-\; \phi^i_n +\tfrac{1}{2} \phi^{i+1}_n + \tfrac{1}{2} \phi^{i-1}_n ,\nn\\
{}^+\pi^{n+1}_j &=& \q\left(   \sum_{i=1}^{Q_n}   A^{n+1}_{ij} (\phi_{n+1}^j - \phi_{n}^i  )\right)  \,\,\, \;+\; \phi^j_{n+1} -\tfrac{1}{2} \phi^{j+1}_{n+1} - \tfrac{1}{2} \phi^{j-1}_{n+1} ,
\ea
where $\pi^n_i$ are the momenta conjugate to the fields $\phi_n^i$ and we assume a cyclic labeling of the vertices. For instance, if $i=Q_n$ then $i+1=1$ for objects $\phi_n^i$ and $\pi^i_n$.
}
\begin{SCfigure}
\psfrag{n}{\small$n$}
\psfrag{n1}{\small$n+1$}
\psfrag{1}{\footnotesize $\phi^1_{n+1}$}
\psfrag{2}{\footnotesize$\phi^2_{n+1}$}
\psfrag{3}{\footnotesize$\phi^3_{n+1}$}
\psfrag{4}{\footnotesize$\phi^4_{n+1}$}
\psfrag{5}{\footnotesize$\phi^5_{n+1}$}
\psfrag{6}{\footnotesize$\phi^6_{n+1}$}
\psfrag{7}{\footnotesize$\phi^7_{n+1}$}
\psfrag{8}{\footnotesize$\phi^1_{n+1}$}
\psfrag{a}{\footnotesize$\phi^1_n$}
\psfrag{b}{\footnotesize$\phi^2_n$}
\psfrag{c}{\footnotesize$\phi^3_n$}
\psfrag{d}{\footnotesize$\phi^4_n$}
\psfrag{e}{\footnotesize$\phi^1_n$}
\vspace{-1cm}
\includegraphics[scale=.8]{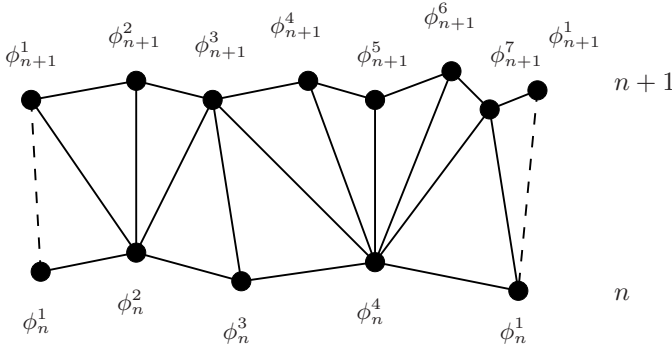}
\caption{\small A global time evolution move for a scalar field on a CDT like triangulation. Here, $Q_n=4$ and $Q_{n+1}=7$. Dashed lines represent periodic identification.}\label{fig_cdt}
\vspace{1cm}
\end{SCfigure}
\end{Example}

\subsection{Review: Regular discrete systems}\label{sec_reg}

We begin by reviewing regular discrete systems in which---by definition---constraints do not arise. The configuration spaces ${\cal Q}_n \cong {\cal Q}$ are of equal dimension at every time step $n$ and
the condition
\ba\label{app3}
\text{det}\frac{\p^{2} S_{n+1}}{\p x_n\p x_{n+1}} \neq 0\,
\ea
holds. A detailed exposition of regular discrete systems is given in \cite{marsdenwest} and so we shall be brief on this subject.

\subsubsection{Lagrangian formulation}

In this regular case, given $(x_{n-1},x_{n})$, (\ref{app2}) can be uniquely solved for $(x_n,x_{n+1})$ and the Lagrangian time evolution map $\cl_n$ in (\ref{app4}) is unique.


The Lagrangian two--form is preserved under the time evolution map (\ref{app4}), 
\ba\label{app7}
\Omega_n&=&{\cal L}_n^* \Omega_{n+1}\,,
\ea
where ${\cal L}_n^*$ represents the pull--back of ${\cal L}_n$. This is proven as follows: by substituting for $x_{n+1}$ the solutions $x_{n+1}(x_{n-1},x_n)$ of the time evolution map (\ref{app4}), the action $S$ in (\ref{app1}) can be written as a function on ${\cal Q}_{n-1}\times {\cal Q}_n$. The exterior derivative of $S$ on ${\cal Q}_{n-1}\times {\cal Q}_n$ yields only boundary terms on account of the equations of motion of the inner variable $x_n$,
\ba\label{app8}
dS\,(x_{n-1},x_n)&=& \frac{\p S_n}{\p x^i_{n-1}}dx^i_{n-1}\,\,+\,\, \frac{\p S_{n+1}}{\p x^j_{n+1}} \left( \frac{\p x^j_{n+1}}{\p x^i_{n-1}}dx^i_{n-1} + \frac{\p x^j_{n+1}}{\p x^i_n}dx^i_n\right) \nn\\
&=& -\theta_{n-1}^- \,+\, {\cal L}_n^* \theta_{n+1}^+ \, .
\ea
$d\circ d=0$ and commutativity of exterior derivatives and pull--backs then imply (\ref{app7}).
It is straightforward to generalize this argument to an arbitrary time step difference $(n_1,n_2)$.

\subsubsection{Canonical formulation}


As a consequence of condition (\ref{app3}), the Legendre transforms (\ref{app9}, \ref{app9b}) are invertible and the Hamiltonian time evolution map $\ch_{n-1}$ (\ref{app10}) generally possesses a unique solution. One can check \cite{marsdenwest} that $\ch_{n-1}=\mathbb{F}^+\circ\cl_{n-1}\circ(\mathbb{F}^+)^{-1}$, where $\cl_{n-1}$ is the Lagrangian time evolution map (\ref{app4}). This, together with the preservation of the Lagrangian two--form (\ref{app7}) can be used to show that the Hamiltonian time evolution is symplectic, i.e.\ preserves the canonical two--form
\ba
\omega_{n-1}=(\ch_{n-1})^*\omega_n.
\ea

The situation for regular discrete systems can be diagrammatically summarized as follows:
\begin{diagram}
& &\cq_{n-1}\times\cq_n&& \rTo^{\cl_n}&&\cq_n\times\cq_{n+1} &   &\\
&\ldTo^{{\mathbb{F}}^-}& &\rdTo^{{\mathbb{F}}^+}&&\ldTo^{{\mathbb{F}}^-}  && \rdTo^{{\mathbb{F}}^+}&\\
 \cp_{n-1}&& \rTo^{\ch_{n-1}}& &\cp_n&& \rTo^{\ch_{n}}&&  \cp_{n+1}.
\end{diagram}

\subsection{Irregular discrete systems}\label{sec_irreg}

Next, we examine the general scenario in which regularity condition (\ref{app3}) is violated. Nonetheless, let $\cq_{n-1}\cong\cq_n$ $\forall\, n$, such that an equal number of left and right null vectors $L_n^i$ and $R^i_{n+1}$, respectively, of the Lagrange two--form
\ba\label{app13}
L_n^i \frac{\p^{2} S_{n+1}}{\p x_n^i \p x_{n+1}^j}\;=0  , \q\q      \frac{\p^{2} S_{n+1}}{\p x_n^i \p x_{n+1}^j} R_{n+1}^j\;=0,   \q
\ea
arises in an open neighborhood in ${\cal Q}_n\times {\cal Q}_{n+1}$. (In this section only, a further index labeling the null vectors will be ignored for the sake of better legibility.) As we shall see in subsection \ref{sec_psext} below, it is possible to formulate systems with temporally varying numbers of degrees of freedom in this irregular fashion.

\subsubsection{The Lagrangian formulation}\label{sec_singlang}

A specific solution $x_{n+1}$ as a function of $(x_{n-1},x_n)$ is not uniquely determined because $x_{n+1}+\varepsilon R_{n+1}$ then constitutes another (infinitesimally displaced) solution of (\ref{app2}). In consequence, arbitrariness in the form of free parameters arises. That is, there exists and $N$ parameter family of Lagrangian time evolution maps $\cl_{n}^{\lambda^l}$, $l=1,\ldots,N$, if there are $N$ independent null vectors $R_{n+1}$. 
A given $\cl_n$, obtained after fixing $N$ {\it a priori} free parameters $\lambda_n^l$ maps at most onto a $(2Q-N)$--dimensional space where $Q$ is the dimension of $\cq$. Therefore, $\cl_n$ is either not injective and/or only defined on some constraint submanifold of $\cq_{n-1}\times\cq_n$. 

On account of (\ref{app13}), the Lagrangian two--form (\ref{app6}) is, obviously, degenerate and possesses $2N$ null directions $L_n^i$ and $R_{n+1}^i$. Nevertheless, if there are no constraints  on $\cq_{n-1}\times\cq_n$ and, hence, a given $\cl_n^{\lambda^l}$ can be defined on all of $\cq_{n-1}\times\cq_n$, the Lagrangian two--form is preserved under the discrete time evolution. The corresponding argument is identical to the one spelled out below (\ref{app7}). 

The same holds for the time reverse of $\cl_n$ \footnote{More precisely, by the time reverse of $\cl_n$ we mean the map defined by the equations of motion (\ref{app2}), given suitable `final data' on $\cq_n\times\cq_{n+1}$.} if some null vectors $L_{n-1}$ are present. It is as this stage useful to distinguish between two types of free data:

\begin{Definition}\label{dfn_free}
We refer to a free parameter
\begin{itemize} \parskip -2mm
\item $\lambda_{n+1}$ at step $n+1$ which cannot be predicted by $\cl_n$, given the Lagrangian data at steps $n-1,n$, as \emph{a priori} free, and
\item $\mu_{n-1}$ at step $n-1$ which cannot be postdicted by the time reverse of $\cl_n$, given the Lagrangian data at steps $n,n+1$, as \emph{a posteriori} free.
\end{itemize}
\end{Definition}

As we shall see in section \ref{sec_conrole} below, these free parameters $\lambda_n,\mu_n$ and null vectors $L_n,R_n$ at some $n$ do not necessarily correspond to gauge freedom. 

\subsubsection{The canonical formulation}\label{sec_singcan}

Recall that in the continuum the Legendre transform $\mathbb{F}L$ fails to be an isomorphism if and only if $\det\left(\frac{\p^2L}{\p\dot{q}^i\p\dot{q}^j}\right)=0$, i.e.\ if and only if the Lagrangian two--form is degenerate. In this case a single primary constraint surface arises \cite{Henneaux:1992ig}. 

The analogous state of affairs holds true in the discrete. Definitions \ref{def_postleg} and \ref{def_preleg} of the {\it pre--} and {\it post--Legendre transforms}
are directly applicable to singular systems. In this case the rank of both Legendre transformations
is $2Q-N$, as a consequence of the $N$ left and $N$ right null vectors (\ref{app13}). The Legendre transformations thus fail to be onto,\footnote{In analogy to the continuum, $\mathbb{F}^\pm S_{n+1}$ simultaneously fail to be isomorphisms if and only if the Lagrangian two--form (\ref{app6}) is degenerate.} and
%
%
%
their images form $(2Q-N)$--dimensional submanifolds in the two phase spaces ${\cal P}_n$ and ${\cal P}_{n+1}$. We shall denote these submanifolds by ${\cal C}_n^-$ and ${\cal C}_{n+1}^+$, respectively. Notice that $\dim\cc^-_n=\dim\cc^+_{n+1}$.

\begin{Definition}\label{def_pcons}
The image of the pre--Legendre transform, $\cc_n^-:=\text{\emph{Im}}(\mathbb{F}^-S_{n+1})\subset\cp_n$, is called the \emph{pre--constraint surface}. The image of the post--Legendre transform, $\cc_{n+1}^+:=\text{\emph{Im}}(\mathbb{F}^+S_{n+1})\subset\cp_{n+1}$, is called the \emph{post--constraint surface}.
\end{Definition}

In contrast to the continuum, we obtain two kinds of constraint surfaces in the discrete since there are two kinds of Legendre transforms. The {\it pre--constraints} governing ${ \cal C}_n^-$ are satisfied by the pre--momenta ${}^-p^n$ by construction because these arise from the definition of the Legendre transform. However, 
the pre--constraints  impose non--trivial conditions on the post--momenta ${}^+p^n$ at $n$, resulting from the previous time evolution move $(n-1)\rightarrow n$. 

If there are $N$ independent left null vectors  of the Lagrangian form (\ref{app13}) (in an open neighbourhood), the pre--constraint hypersurface will be of co--dimension $N$ in ${\cal P}_n$ and  we choose $N$ constraints ${}^-C^{n}_l$, $l=1,\ldots,N$, to describe this pre--constraint hypersurface. Thus, this constraint set is irreducible\footnote{The gradients of the constraints are linearly independent.} \cite{Henneaux:1992ig}. 

%

The  $N$ (irreducible) {\it post--constraints} ${}^+C^{n+1}_r$, $r=1,\ldots,N$, defining ${ \cal C}_{n+1}^+$ 
are automatically satisfied by the momenta ${}^+p^{n+1}$ (for all initial values) after having performed the evolution move $n\rightarrow (n+1)$. By momentum matching, these post--constraints will provide non--trivial conditions for the pre--momenta ${}^-p^{n+1}$ of the next evolution move $(n+1)\rightarrow (n+2)$.  

The {\it pre--} and {\it post--constraint surfaces} at fixed step $n$ generally do {\it not} coincide, $\cc^+_n\neq\cc^-_n$. Making sure that both {\it pre--} and {\it post--constraints} are satisfied at each step such that we restrict to $\cc^+_n\cap\cc^-_n$ is the non--trivial challenge and will be further discussed in section \ref{sec_conmatch}. 


For singular systems, the Hamiltonian time evolution ${\cal H}_n$ (\ref{app10}) can only be defined as a map from  ${\cal C}_n^-$  to $\cc^+_{n+1}$. That is, $\ch_n:\cc^-_n\rightarrow\cc^+_{n+1}$. Nevertheless, $\ch_n$ is generated by the discrete action and implicitly defined by (\ref{b3}). 
%
%
%
However, as a consequence of the right null vectors in (\ref{app13}), we are no longer able to solve the pre--momentum equation in (\ref{b3}) uniquely for $x_{n+1}$: 
therefore, as in the Lagrangian picture, 
we must specify
$N$ {\it a priori} free parameters $\lambda^r_{n+1}$, $r=1,\ldots,N$ 
to uniquely determine $x_{n+1}(x_n,\lambda_{n+1},{}^-p^n)$. Moreover, the post--momenta ${}^+p^{n+1}$ will in general depend on the parameters $\lambda_{n+1}$.

Likewise, on account of the 
left null vectors $L_n$, we can no longer uniquely express $x_n$ as a function of $x_{n+1},{}^+p^{n+1}$ only and $N$ {\it a posteriori} free parameters $\mu^l_n$, $l=1,\ldots,N$, that cannot be postdicted via (the time reverse of) $\ch_n$ and the canonical data at ${n+1}$ need be specified in order to determine $x_n(x_{n+1},\mu_n,{}^+p^{n+1})$. 

Via the momentum matching (\ref{mm}), 
 the {\it a priori} and {\it a posteriori} free parameters $\lambda_n,\mu_n$ at step $n$ of the Lagrangian time evolution map coincide with those of the Hamiltonian time evolution map. We have one {\it a priori} free $\lambda_n$  per post--constraint and one {\it a posteriori} free $\mu_n$ per pre--constraint.


In analogy to the continuum situation, arbitrariness in the form of free  parameters $\lambda,\mu$ thus appears in the canonical discrete evolution of singular systems. But as in the continuum the preservation of the constraints may lead to a fixing of some of these parameters $\lambda,\mu$. This will be discussed in sections \ref{sec_conmatch}--\ref{sec_varobs}. 


In consequence of the Hamiltonian time evolution $\ch_n:\cc^-_n\rightarrow\cc^+_{n+1}$ restricting to the constraint hypersurfaces, $\ch_n$ cannot be a symplectic map. However, it defines a pre--symplectic map: using the embeddings $\iota_n^-:{\cal C}_n^- \rightarrow {\cal P}_n$ and $\iota_{n+1}^+:{\cal C}_{n+1}^+ \rightarrow {\cal P}_{n+1}$, the canonical two--forms
\ba\label{app19}
\omega_n=dx_n^j\,\wedge\, dp^n_j  , \q\q\omega_{n+1}=dx_{n+1}^j\,\wedge\,dp^{n+1}_j .
\ea
can be pulled back to two--forms on the constraint surfaces ${\cal C}^-_n$ and ${\cal C}^+_{n+1}$, respectively. The resulting two--forms are pre--symplectic forms\footnote{As will be explained later, the pre--constraints in ${\cal P}_n$ and the post--constraints in ${\cal P}_{n+1}$ each define a first class set of constraints. The pull back of the symplectic form to a constraint hypersurface described by first class constraints is a pre--symplectic form with one null direction per (independent) constraint. }  and Hamiltonian time evolution preserves these pre--symplectic forms as stated in the following theorem. 

\begin{Theorem}\label{thm_presymp}
The discrete global Hamiltonian time evolution map $\ch_n:\cc^-_n\rightarrow\cc^+_{n+1}$ satisfies
\ba
(\iota_n^-)^*\omega_n = {\cal H}_n^*  (\iota_{n+1}^+)^*\omega_{n+1}  .
\ea
\end{Theorem}

\begin{proof}
See Theorem 6.1 in \cite{Dittrich:2011ke}.
\end{proof}

Let us conclude the treatment of singular discrete systems with a diagrammatic symmary: 
\begin{diagram}
& &\cq_{n-1}\times\cq_n&& \rTo^{\cl_n}&&\cq_n\times\cq_{n+1} &   &\\
&\ldTo^{{\mathbb{F}}^-}& &\rdTo^{{\mathbb{F}}^+}&&\ldTo^{{\mathbb{F}}^-}  && \rdTo^{{\mathbb{F}}^+}&\\
 \cp_{n-1}\supset\cc^-_{n-1}&& \rTo^{\ch_{n-1}}& &\cc^+_n\cap\cc^-_n&& \rTo^{\ch_{n}}&&  \cc^+_{n+1}\subset\cp_{n+1}.
\end{diagram}

\subsection{Evolving phase spaces and configuration and phase space extensions}\label{sec_psext}

In order to describe a growing/shrinking lattice, we must be able to deal with numbers of variables which vary in discrete time, such that the concept of an evolving phase space becomes necessary. 

Note that the Definitions \ref{def_postleg} and \ref{def_preleg} of the {\it pre--} and {\it post--Legendre transforms} are 
also  applicable to a situation where $\cq_{n-1}\ncong\cq_n$.  In this case the transforms obviously fail to be isomorphisms and we are necessarily in the situation of a singular system. The mapping from $\cq_{n-1}\times\cq_n$ to the phase spaces $\cp_{n-1}=T^*\cq_{n-1}$ and $\cp_n=T^*\cq_n$ is well defined  also in the case  $\dim\cq_{n-1}\neq\dim\cq_n$. Therefore, all the results of section \ref{sec_gen} apply to a $\cq_n$ which varies with $n$. The action contribution $S_n$ is a generating function of the first kind as before, just not for a canonical transformation (as the corresponding Hamiltonian map is not symplectic but rather pre--symplectic), but for a singular global time evolution.
 


Nevertheless, it will generally be convenient to work with suitable configuration and phase space extensions such that one obtains phase spaces of equal dimension at $(n-1)$ and $n$. This simplifies the discussion of the consequences for the preservation of the symplectic structure under $\ch_{n-1}$. 

The prescription for the phase space extension is simple. Assume there exist `old' configuration variables $x^o_{n-1}$ which appear at time step $(n-1)$ but do not have counterparts $x^o_n$ at step $n$. We would like to extend the phase space at time $n$ such as to include the pairs $(x^o_n,p^n_o)$. To this end, extend the configuration manifold $\cq_n$ to a new configuration space 
$\bar{\cq}_n:=\cq_n\times\cq^{ext}_n$, where $\cq^{ext}_n$ is a suitable configuration manifold of appropriate dimension 
coordinatized by the desired $x^o_n$.\footnote{For instance, in Regge Calculus one has {\it a priori}, i.e., before imposing any equations of motion or generalized triangle inequalities, $\cq_n=\mathbb{R}_+\times\cdots\times\mathbb{R}_+$ (the lengths cannot be negative). One can, therefore, simply extend $\cq_n$ by some additional products of $\mathbb{R}_+$ coordinatized by the lengths $l^o_n$ of `old edges' (which actually do not appear in the triangulation at step $n$). See \cite{Dittrich:2011ke} for more details.} In this way the coordinate form of the Lagrangian two--form (\ref{app6}) will be again a square matrix, however, with (additional) left and right null vectors corresponding to the formally added variables $x^o_n$.  The action $S_n$ as well as $S_{n+1}$ do {\it not} depend on these variables as these were only introduced for book keeping purposes. The $x^o_n$ are thus both {\it a priori} and {\it a posteriori} free parameters $\lambda^o_n=\mu^o_n$ according to definition \ref{dfn_free}. Upon extension to $\bar{\cq}_n$ we are therefore in the case of singular discrete dynamics and all the results of subsections \ref{sec_singlang} and \ref{sec_singcan} apply. The extension for `new' variables $x^{new}_n$ which only appear at step $n$ but do not have counterparts at $(n-1)$ proceeds analogously.

On the extended configuration spaces $\bar{\cq}_{n-1}\times\bar{\cq}_n$, where now $\dim\bar{\cq}_{n-1}=\dim\bar{\cq}_n$, one may then perform the discrete Legendre transformations (\ref{app9}, \ref{app9b}), yielding a mapping to the extended phase spaces $\bar{\cp}_{n-1}:=T^*\bar{\cq}_{n-1}$ and $\bar{\cp}_n:=T^*\bar{\cq}_n$. Denote by $(x^{new}_{n-1},p^{n-1}_{new})$ and $(x^{o}_n,p^n_{o})$ the pairs by which the phase spaces at $n-1$ and $n$ have been extended, respectively. The Hamiltonian time evolution map $\ch_{n-1}$ is extended to the formally added variables by simply using the action $S_n$ as a generating function of the first kind with trivial dependence on $x^{new}_{n-1},x^o_n$, resulting in the evolution equations
\ba\label{c4}
{}^-p^{n-1}_{new}\,&=&\, -\frac{\partial S_{n}}{\partial x^{new}_{n-1}}\,=\,0 ,\q\q\q\q  {}^+p^{n}_{new}\,=\, \frac{\partial S_{n}}{\partial x^{new}_{n}},\nn\\
{}^-p^{n-1}_{o}\,&=&\, -\frac{\partial S_{n}}{\partial x^{o}_{n-1}} ,\q\q\q\q\q\q \q \,{}^+p^{n}_{o}\,=\, \frac{\partial S_{n}}{\partial x^{o}_{n}}=0.
\ea
That is, additional {\it pre--} and {\it post--constraints}
\ba\label{npcon}
{}^-C^{n-1}_{new}={}^-p^{n-1}_{new}=-\f{\p S_n}{\p x^n_{n-1}}=0,\q\q\q{}^+C^n_o={}^+p^n_{o}=\f{\p S_n}{\p x^{o}_n}=0
\ea
now arise at $n-1$ and $n$, respectively, which are in number equal to the difference in dimension between the extended and unextended configuration spaces, respectively.
 The total resulting numbers of pre--constraints at $n-1$ and post--constraints at $n$ are identical as the Lagrangian two--form $\Omega_n$ (\ref{app6})  corresponds to a square matrix which has equally many left and right null vectors $L_{n-1}, R_n$ (\ref{app13}). 

Thus, the added variables---being both {\it a piori} and {\it a posteriori} free---cannot be determined and their conjugate momenta are constrained.\footnote{If one also extended the phase spaces at previous time steps by the added variables $x^{new}$, the momentum $p_{new}$ would still be vanishing. The same holds for the variables $x^o$ and later time steps.} We are free, however, to make the `gauge choices' $x^{new}_{n-1}=x^{new}_n$ and $x^o_{n-1}=x^o_n$. Since $S_n$ does not depend on $x^{new}_{n-1},x^{o}_n$, neither does any of the already existing pre-- or post--constraints at $n-1$ and $n$, respectively, such that both new sets in (\ref{npcon}) are first class. Hence, returning to the unextended (`reduced') phase spaces $\cp_{n-1}$ and $\cp_n$ is simply performed by a partial reduction procedure consisting of imposing ${}^-p^{n-1}_{new}=0$ and ${}^+p^n_{o}=0$ and factoring out the corresponding first class flow. This amounts to simply dropping the added pairs $(x^{new}_{n-1},p^{n-1}_{new})$ and $(x^o_n,p^n_o)$ from their respective phase spaces.

In that sense, the phase space extension is a trivial extension which simply adds constrained canonical pairs. Nevertheless, precisely this allows us to directly apply the discussion of singular systems of the previous section to systems with varying number of variables.



\begin{Example}\label{examp0}
\emph{Coming back to our massless scalar field, we found the time evolution equations from step $n$ to $(n+1)$ as
\ba
{}^-\!\pi^n_i &=& -\left( \sum_{j=1}^{Q_{n+1}} A^{n+1}_{ij} ( \phi_{n}^i-\phi_{n+1}^j)\right)  \;-\; \phi^i_n +\tfrac{1}{2} \phi^{i+1}_n + \tfrac{1}{2} \phi^{i-1}_n  , \nn\\
{}^+\!\pi^{n+1}_j &=&\q \left(   \sum_{i=1}^{Q_n}   A^{n+1}_{ij} (\phi_{n+1}^j - \phi_{n}^i ) \right)   \;+\; \phi^j_{n+1} -\tfrac{1}{2} \phi^{j+1}_{n+1} - \tfrac{1}{2} \phi^{j-1}_{n+1}  .
\ea
with $A^{n+1}_{ij}$ encoding the adjacency relations between the vertices $i$ at time $n$ and the vertices $j$ at time $(n+1)$. The number of configuration variables is $Q_n$ and $Q_{n+1}$ for times $n,(n+1)$ respectively. Thus, if these numbers disagree we can introduce auxiliary additional field variables. Since the corresponding dynamics is trivial we will refrain from doing so explicitly. }

\emph{The Lagrangian two--form agrees with the adjacency matrix $A_{ij}^{n+1}$. If $Q_{n+1}=Q_n+N$ with $N>0$ we have at least $N$ right null vectors $(R_{n+1})^j_r,r=1,\ldots, N$ of the Lagrangian two--form (and $N$ trivial left null vectors  if $N$ auxilliary variables $\phi_n^{o}$ are introduced).  The corresponding $N$ post--constraints are given as
\ba
{}^+\! C_r^{n+1} &=& \sum_j^{Q_{n+1}} \left( {}^+\!\pi^{n+1}_j  - \left(\left(\sum_{i=1}^{Q_{n}}  A^{n+1}_{ij}\right) \,\,+1 \right)  \phi^j_{n+1} \,+\,\tfrac{1}{2} \phi^{j-1}_{n+1} \,+\,\tfrac{1}{2} \phi^{j+1}_{n+1}\right) (R_{n+1})^j_r  . \q\q
\ea}
\begin{SCfigure}
\psfrag{n}{\small$n$}
\psfrag{n1}{\small$n+1$}
\psfrag{1}{\footnotesize $\phi^1_{n+1}$}
\psfrag{2}{\footnotesize$\phi^2_{n+1}$}
\psfrag{3}{\footnotesize$\phi^3_{n+1}$}
\psfrag{4}{\footnotesize$\phi^4_{n+1}$}
\psfrag{5}{\footnotesize$\phi^5_{n+1}$}
\psfrag{6}{\footnotesize$\phi^6_{n+1}$}
\psfrag{7}{\footnotesize$\phi^7_{n+1}$}
\psfrag{8}{\footnotesize$\phi^1_{n+1}$}
\psfrag{a}{\footnotesize$\phi^1_n$}
\psfrag{b}{\footnotesize$\phi^2_n$}
\psfrag{c}{\footnotesize$\phi^3_n$}
\psfrag{d}{\footnotesize$\phi^4_n$}
\psfrag{e}{\footnotesize$\phi^1_n$}
\vspace{-1cm}
\includegraphics[scale=.6]{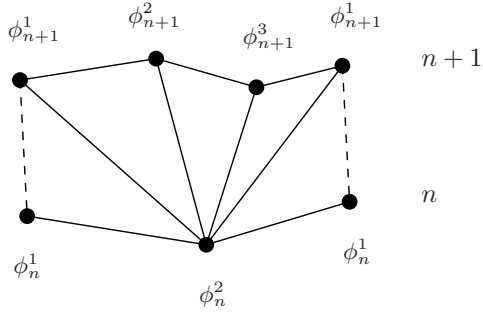}
\hspace*{1cm}
\caption{\small A global time evolution move for a scalar field on a triangulation with $Q_n=2$ and $Q_{n+1}=3$. Dashed lines represent periodic identification.}\label{fig_field23}
\vspace{1cm}
\end{SCfigure}

\emph{For instance, for the triangulation in figure \ref{fig_field23} with two vertices at time $n$ and three vertices at time $(n+1)$   
 we have the adjacency matrix and right null vector 
 \ba
 A^{n+1}&=&\left(         
 \begin{matrix}
 1 & 0 & 0 \\
 2& 1 & 1 
 \end{matrix}
 \right)  ,\q\q 
 R_{n-1}\;=\; \left(
 \begin{matrix}
 0\\1\\-1
 \end{matrix}
 \right)  .
 \ea
This yields the post--constraint at time step $(n+1)$
\ba
{}^+C^{n+1}&=& \pi^{n+1}_2 -\pi^{n+1}_3 \,-\tfrac{5}{2} \left( \phi_{n+1}^2 -  \phi_{n+1}^3 \right)  ,
\ea
corresponding to an {\it a priori} free parameter $\lambda_{n+1}=\phi_{n+1}^2-\phi_{n+1}^3$.}

\emph{Time--reversing the triangulation, i.e.\ just exchanging $n$ and $(n+1)$, one obtains a pre--constraint at time $n$
\ba
{}^-\! C^n &=&  \pi^{n}_2 -\pi^{n}_3 \,+\tfrac{5}{2} \left( \phi_{n}^2 -  \phi_{n}^3 \right)  ,
\ea
corresponding to an a posteriori free parameter $\mu_n=\phi_n^2-\phi_n^3$.}
\end{Example}

\subsection{Bulk variables and Hamilton's principal function}\label{sec_bulk}

Some systems involve bulk variables---which we shall denote by $x^t_n$---that only appear in a single action contribution $S_n$.  Thus, the corresponding equations of motion simply read $\frac{\p S_n}{\p x^t_n}=0$. Examples appear in the context of evolution schemes in Regge Calculus which, e.g., involve fat slices \cite{Dittrich:2011ke} or tent moves \cite{Dittrich:2011ke,Bahr:2009ku,Dittrich:2009fb,Barrett:1994ks}. For example, figure \ref{fatslices} depicts two adjacent fat slices bounded by non-intersecting hypersurfaces $\Sigma_n$ in a triangulation. The variables are given by the lengths of the edges and all the lengths of the dashed edges will appear as bulk variables either in $S_n$ or in $S_{n-1}$.

\begin{SCfigure}
\psfrag{n}{ $n$}
\psfrag{n+1}{$n+1$}
\psfrag{sn0}{$\Sigma_{n-1}$}
\psfrag{sn}{ $\Sigma_n$}
\psfrag{sn1}{ $\Sigma_{n+1}$}
\psfrag{Sn}{$S_n$}
\psfrag{Sn1}{$S_{n+1}$}
\includegraphics[scale=.4]{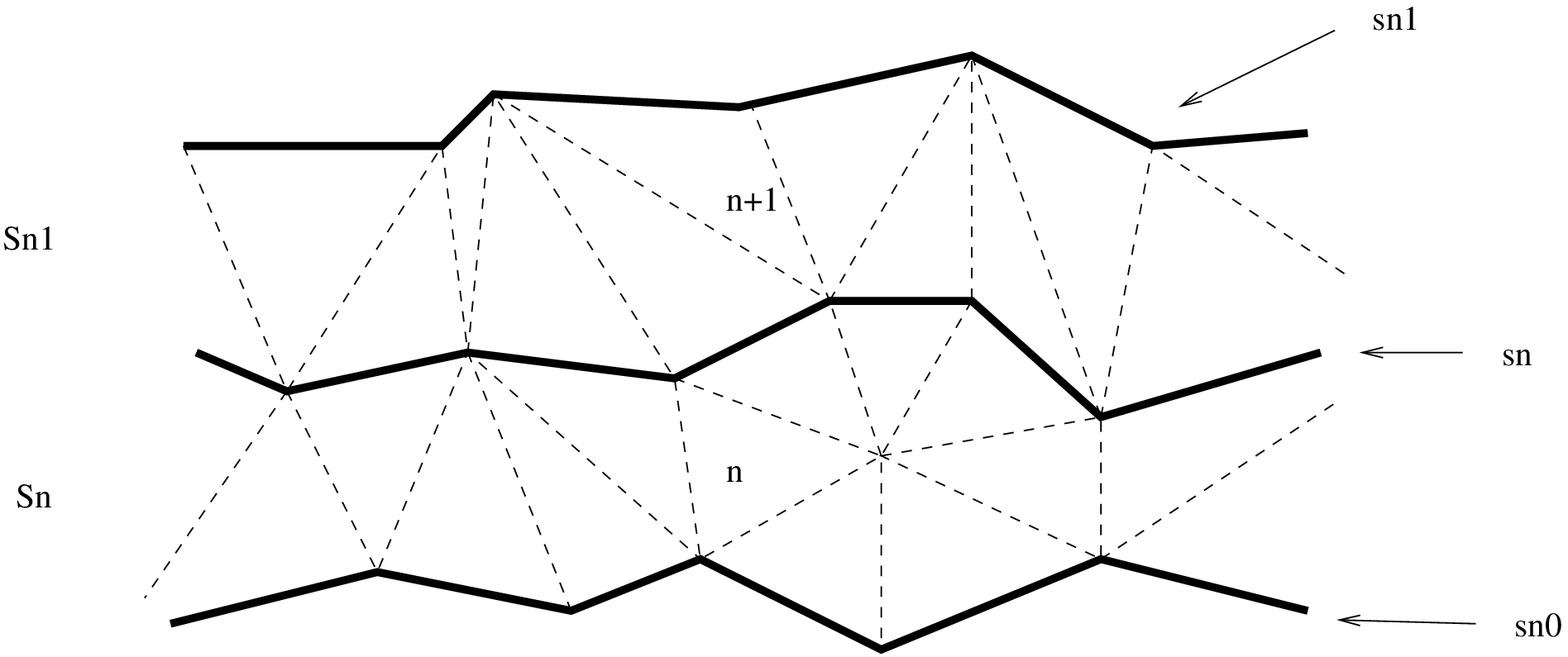}
\caption{\small A fat slicing of a triangulation.}\label{fatslices}
\end{SCfigure}
\noindent

Such bulk variables also appear when 
performing 
`effective' evolution moves $n\rightarrow n+k$   which consist of multiple intermediate basic time steps $n+1,n+2,\ldots,n+k-1$. 
The variables at these intermediate steps are then the bulk variables of the action $S_{n,n+k}=\sum^{n+k}_{l=n+1}S_l$. Integrating out these variables leads to Hamilton's principal function, describing the effective evolution $n\rightarrow n+k$, which, in fact, is a function of the boundary variables only. In this case, we shall relabel the time steps so that the effective evolution move is now labeled as $n\rightarrow (n+1)$ and all the bulk variables of $S_{n,n+k}$ are assigned to the new time step $n+1$.



By virtue of this splitting of the configuration variables into bulk variables $x^t_n$ and boundary variables $x^e_n$, which in general will appear in two action contributions, the evolution equations defining $\ch_{n-1}$ (\ref{b3}) become
\ba\label{b6}
{}^{-}p^{n-1}_e&:=& -\frac{\partial S_{n}}{ \partial x^e_{n-1}}  ,\q \q\q \q\q \,{}^{+}p^{n}_e  := \frac{\partial S_n}{ \partial x^e_{n}}   \nn\\
{}^{-}p^{n-1}_t&:= &-\frac{\partial S_{n}}{ \partial x^t_{n-1}} \,=\,0   ,\q \q\q  {}^{+}p_t^{n}  := \frac{\partial S_n}{ \partial x^t_{n}} .
\ea
As a result, one obtains the pre--constraints ${}^-C^{n-1}_t={}^-p_t^{n-1}=0$ ($S_n$ does {\it not} depend on $x^t_{n-1}$) which, by momentum matching, leads to ${}^+p_t^{n-1}=0$. 


Thus, also the equations of motion for the bulk variables are implemented by momentum matching, in this case by a constraint, 
\be\label{mommat}
0={}^-p^{n+1}_t={}^+p^{n}_t=\frac{\partial S_{n}}{ \partial x^t_{n}} .
\ee
This, in fact, is a general attribute of the present formalism: {\it equations of motion manifest themselves as canonical constraints $p^n=0$}. 

Denote by $\chi^t_n(x^e_{n-1},x^{e'}_n,\kappa^m)$ the solutions to the equations of motion (\ref{mommat}) for the internal variables $x^t_n$ where the $\kappa^m$ denote possible free parameters. Substituting this into $S_n$, one can define an effective action $\tilde S_n$ (or Hamilton's principal function) which only depends on the `true boundary' variables $x^e_{n-1}$ and $x^e_{n}$. The effective action defined in this way does not depend on the free parameters $\kappa^m$: The solutions $\chi^t_n(x^e_{n-1},x^{e'}_n,\kappa^m)$  extremize the action and the effective action is equal to the value of this extremum. The free parameters appear only if the extremum is actually not a proper extremum, rather there are constant directions along which the action (extremum) does not change. Below, however, we will redefine the effective action for the case that the equations of motion for the bulk variables imply constraints on the boundary variables.


The effective action defines the effective dynamics of the remaining degrees of freedom and, in analogy to definitions \ref{def_postleg} and \ref{def_preleg}, the {\it effective pre--} and {\it post--Legendre transformations} $\tilde{\mathbb{F}}^-\tilde{S}_n:\tilde{\cq}_{n-1}\times\tilde{\cq}_n\rightarrow T^*\tilde{\cq}_{n-1}$ and $\tilde{\mathbb{F}}^+\tilde{S}_n:\tilde{\cq}_{n-1}\times\tilde{\cq}_n\rightarrow T^*\tilde{\cq}_n$, respectively, where $\tilde{\cq}_n$ denotes the configuration manifold at step $n$ after solving (\ref{mommat}). The following theorem shows that the dynamics defined by the effective action $\tilde{S}_n$ is equivalent to the dynamics defined by $S_n$.
\begin{Theorem}
Let ${\cq}_n:=\tilde{\cq}_n\times\cq^t_n$, where $\tilde{\cq}_n$ and $\cq^t_n$ are coordinatized by $x^e_n$ and $x^t_n$, respectively. The following two diagrams commute
\begin{displaymath}
\begin{diagram}
{\cq}_{n-1}\times{\cq}_n&& \rTo^{\,\mathbb{F}^+}&&T^*{\cq}_n\\
\dTo^{\text{eom (\ref{mommat})}}& &&&\dTo^{}{{}^-p^n_t={}^+p^n_t}\\
\tilde{\cq}_{n-1}\times\tilde{\cq}_n&& \rTo^{\,\tilde{\mathbb{F}}^+}&& T^*\tilde{\cq}_n
\end{diagram}
\q\q\q\q\q
\begin{diagram}
{\cq}_{n-1}\times{\cq}_n&& \rTo^{\,\mathbb{F}^-}&&T^*{\cq}_{n-1}\\
\dTo^{\text{eom (\ref{mommat})}}& &&&\dTo^{}{{}^-p^n_t={}^+p^n_t}\\
\tilde{\cq}_{n-1}\times\tilde{\cq}_n&& \rTo^{\,\tilde{\mathbb{F}}^-}& &T^*\tilde{\cq}_{n-1}.
\end{diagram}\end{displaymath}
This also holds in the presence of constraints on the boundary data, i.e.\ on $\tilde{\cq}_{n-1}\times\tilde{\cq}_n$.
\end{Theorem}

\begin{proof}
 To show that the effective action defines the same dynamics as the action we started with, we have to convince ourselves that the momenta and constraints agree if we perform either of the following two procedures:\\
(a) We first define the canonical time evolution (\ref{b6}) and then impose momentum matching ${}^-p_t^{n}=0={}^+p_t^{n}$, which implements the equation of motions for the bulk variables $x^t_{n}$ (right--down in the diagrams). \\
(b) We first solve the equations of motion for the bulk variables and define the effective action. We use this effective action to then define the canonical time evolution from $(n-1)$ to $n$ (down--right in the diagrams).

~\\
The equivalence of the two methods in the absence of constraints on $\tilde{\cq}_{n-1}\times\tilde{\cq}_n$ follows from
\ba
{}^{-}\tilde p^{n-1}_e &=& -\frac{ \partial \tilde S_n}{\partial x^e_{n-1}} \;=\; - \frac{\partial S_n}{\partial x^e_{n-1}} - \frac{\partial S_n}{\partial x^t_n} \frac{\partial \chi^t_n}{x^e_{n-1}} \;=\;  -\frac{\partial S_n}{\partial x^e_{n-1}} \;=\;{}^-p^{n-1}_e,
\ea
where the last equation holds due to the equations of motion $ \frac{\partial S_n}{\partial x^t_n}=0$ for the bulk variables. A similar calculation shows equality of the post--momenta.

However, this argument is not sufficient if the equations of motion for the bulk variables impose constraints between the boundary variables $x^e_{n}$ and $x^e_{n-1}$. These can come as proper canonical constraints involving only configuration variables at one time---either $x^e_n$ or $x^e_{n-1}$.  Such constraints are called holonomic and arise in this form after momentum matching for the bulk momenta. The other type of constraints that can occur are relations between variables at time $(n-1)$ and time $n$ and are thus not proper canonical constraints. (Canonical constraints are equations of motions involving canonical data of only one time step.) We will refer to both kinds of constraints as boundary data constraints. 

Let us assume that, after solving for the bulk variables, one finds $M$ independent  boundary data constraints $H_h(x^e_{n-1},x^e_{n})=0$ (forming an irreducible set). The form of these constraints is not unique---but locally these constraints can be solved for $M$ variables $x^h=\chi^h(x^{\bar{e}}_{n-1},x^{\bar{\bar{e}}}_n )$ out of the $x^e_{n-1},x^e_n$, where $\bar{e}$ denotes the remaining variables at $(n-1)$ and $\bar{\bar{e}}$ the remaining variables at $n$. We furthermore assume that the matrix of derivatives
$\frac{\partial H_h}{\partial x^{h'}} $
is (locally) invertible. This allows us to express the derivative of $\chi^{h'}$ with respect to $x^{\bar{e}}_{n-1}$
\ba
\frac{\partial \chi^{h'}}{\partial x^{\bar{e}}_{n-1}} &=& - \left(  \frac{\partial H_{h} }{\partial x^{h'}}  \right)^{-1} \frac{\partial H_{h}}{\partial x^{\bar{e}}_{n-1}}  ,
\ea
which can be derived by taking the derivative of $H_h(x^{\bar{e}}_{n-1},x^{\bar{\bar{e}}}_n,\chi^{h'}(x^{\bar{e}}_{n-1},x^{\bar{\bar{e}}}_n ))$. 

Also, if we have the equations of motion for the bulk variables $x^t_n$ resulting in $M$ (independent and irreducible) relations restricting the boundary data, we will have at least $M$ free parameters $\kappa^h$ (among the $\kappa^m$ introduced below (\ref{mommat})) for the solutions $\chi^t_n$.  
We will take these $\kappa^h=\kappa^h_n$ as additional bulk variables for the effective action, acting as Lagrange multipliers.  If there are more than $M$ free parameters a suitable choice has to be made. Below we will derive a condition for this choice.

Thus, for the effective action to impose the boundary data constraints we define
\ba\label{wed20}
\tilde S_n(x^e_{n-1},x^e_n, \kappa^h_n)&:=& S_n(x^{\bar{e}}_{n-1},x^{\bar{\bar{e}}}_n, \chi^h(x^{\bar{e}}_{n-1},x^{\bar{\bar{e}}}_n ), \chi^t_n(x^{\bar{e}}_{n-1},x^{\bar{\bar{e}}}_n, \kappa^m)) +\nn\\&&\q\q\q\q\q\q\q  \Lambda^h(\kappa^{h'}_n,x^{\bar{e}}_{n-1},x^{\bar{\bar{e}}}_n) H_h(x^e_{n-1},x^e_{n}) \, .
\ea
The functions $\Lambda^h$ will be determined such that the momenta obtained by the two different methods (a) and (b) coincide.  But we need to assume that $\frac{\partial\Lambda^h{}'}{\partial\kappa^h_n}$ is invertible so that the equations of motion for the set $\{\kappa^h_n\}$ impose the constraint set  $\{H^{h}\}$. (The first term on the right hand side of (\ref{wed20}), as previously argued, does not depend on $\kappa^m$.) 

To see this, consider the momenta at $(n-1)$ as defined from the effective action (we omit terms proportional to $H_h$ as these vanish if the equations of motion for the $\kappa^h_{n}$ are imposed)
\ba
{}^{-}\tilde p_{\bar{e}}^{n-1} &=&
-\frac{\partial S_n}{\partial x^{\bar{e}}_{n-1} } -\frac{\partial S_n}{\partial x^h} \frac{\partial \chi^h}{\partial x^{\bar{e}}_{n-1}} - \Lambda^h \frac{\partial H_h}{\partial x^{\bar{e}}_{n-1}}
\nn\\
&=&
-\frac{\partial S_n}{\partial x^{\bar{e}}_{n-1} } - \frac{\partial H_h}{\partial x^{\bar{e}}_{n-1}} \left(  \Lambda^h -    \frac{\partial S_n}{\partial x^{h'}}  \left(  \frac{\partial H_{h} }{\partial x^{h'}}  \right)^{-1}   \right) \, ,\nn\\
{}^{-}\tilde p_{\bar{h}}^{n-1} &=& -\Lambda^h \frac{\partial H^h}{\partial x^{\bar h}_{n-1}} \q 
\ea
where $x^{\bar h}_{n-1}$ are those variables among the $x^h$ which are associated to the time $(n-1)$. 

Comparing with the momenta obtained from the original action,
\ba
{}^{-} p_{\bar{e}}^{n-1} \;=\; -\frac{\partial S_n}{\partial x^{\bar{e}}_{n-1} }   ,\q\q
{}^{-} p_{\bar{h}}^{n-1} \;=\;  -\frac{\partial S_n}{\partial x^{\bar{h}}_{n-1} }  ,
\ea
(evaluated on the surface $\{x^t_n=\chi^t_n,H^h=0\}$), we see that we need to set 
\ba\label{wed1}
\Lambda^h= \frac{\partial S_n}{\partial x^{h'}}  \left(  \frac{\partial H_{h} }{\partial x^{h'}}  \right)^{-1}   \q 
\ea
to obtain for the pre--momenta $\tilde p^{n-1}=p^{n-1}$. The same  line of arguments can be made for the post--momenta.

Using the solutions 
$x^t_n=\chi^t_n(x^{\bar{e}}_{n-1},x^{\bar{\bar{e}}}_n, \kappa^m)$ and $x^h=\chi^h(x^{\bar{e}}_{n-1},x^{\bar{\bar{e}}}_n )$ in (\ref{wed1}) will make the $\Lambda^h$ dependent on $\kappa^h_n$.
The assumption on the invertibility of  $\frac{\partial\Lambda^h{}'}{\partial\kappa^h_n}$ translates then into the assumption that
\ba
\left(\frac{\partial}{\partial x^t_n} \frac{\partial S_n}{\partial x^{h'}} \right)\frac{\partial \chi^t_n}{\partial \kappa^h}
\ea
is invertible. This, in particular, might restrict the choice of the $M$ parameters $\kappa^h$ among the free parameters $\kappa^m$ for the solutions $x^t_n=\chi^t_{n}$. 

This finally shows that the dynamics defined by the original action and the effective action is equivalent. 
\end{proof}

In the case of constraints on the boundary variables one has to be careful and incorporate these constraints into the effective action. 
For instance, a constraint on the boundary data arises if 
we have an action of the form
\ba
S_n(   x^e_{n-1},x^e_n,x^t_n) &=&     S'_n(   x^e_{n-1},x^e_n   )  +            x^t_n H(x^e_{n-1},x^e_n)  .
\ea
 In this case $x^t_n$ will be undetermined by the equations of motion for $x^t_n$. Thus, we will have $x^t_n=\kappa$. For $\Lambda$ we find according to (\ref{wed1})
\ba
\Lambda&=& \kappa + \frac{\partial S'_n}{\partial x^{h'}}  \left(  \frac{\partial H_{h} }{\partial x^{h'}}  \right)^{-1}   .
\ea
Hence, $\Lambda$  replaces the bulk variable that acted as a Lagrange multiplier in the original action.

Let us briefly return to an effective evolution move $n\rightarrow n+k$, where we treat all intermediate variables as bulk variables.  
 The previous discussion implies that it does not matter \\
(i) in which ordering one integrates out intermediate evolution steps, and \\
(ii) in which ordering one engages Legendre transformations and equations of motion\\
in order to evolve from $n$ to $n+k$.


Finally, in appendix \ref{app_bulk}, we show that, under the simplifying assumption of absence of boundary data constraints, 
the effective canonical two--form $\tilde{w}_n$ on $T^*\tilde{\cq}_n$ and the effective Lagrangian two--forms $\tilde{\Omega}_n$ on $\tilde{\cq}_{n-1}\times\tilde{\cq}_n$ and $\tilde{\Omega}_{n+1}$ on $\tilde{\cq}_{n}\times\tilde{\cq}_{n+1}$ are, indeed, correctly related by pull--back 
\ba\label{omeff}
\tilde{\Omega}_n=(\tilde{\mathbb{F}}^+\tilde{S}_n)^*\tilde{\omega}_n,\q\q\q\tilde{\Omega}_{n+1}=(\tilde{\mathbb{F}}^-\tilde{S}_n)^*\tilde{\omega}_n
\ea
via the effective post-- and pre--Legendre transformations $\tilde{\mathbb{F}}^\pm\tilde{S}_n$, respectively. 


\section{Local dynamics of variational discrete systems}\label{sec_local}

In the previous section we have only considered {\it global} time evolution moves, i.e.\ time evolution moves such that no subsets of variables at any two neighbouring time steps $n, n+1$ coincide (except possibly in the boundary). In this section, let us now also consider {\it local} evolution moves which only evolve subsets of the canonical data of a given time step and which for distinction from the global moves we shall label by $k\in\mathbb{Z}$, instead of $n\in\mathbb{Z}$. That is, different discrete time steps $k,k+1$, may partially overlap such that they involve coinciding subsets of canonical variables. Such time evolution moves occur, e.g.\ in simplicial gravity or lattice field theory when only small regions of a discrete hypersurface $\Sigma$ are evolved in discrete time such that $\Sigma_k\cap\Sigma_{k+1}\neq\emptyset$; the Pachner moves \cite{Dittrich:2011ke} and the tent moves \cite{Dittrich:2009fb,Dittrich:2011ke} are particular examples of such local moves (see these reference for their canonical formulation). 
 Such local evolution moves define the most basic evolution steps from which more complicated time evolutions can be constructed. In particular, the Pachner moves constitute an elementary and ergodic class of local moves applicable to arbitrary triangulations such that one can map between any (finite) triangulations of fixed topology by finite sequences of these moves \cite{pachner1,pachner2}.

%



\subsection{Coinciding subsets of variables and {\it momentum updating}}

We will formulate the evolution equations corresponding to the local evolution moves directly in the Hamiltonian picture. The reason is twofold:
\begin{itemize}
\item[(i)] In the discrete the Lagrangian picture (see section \ref{sec_lagr}) necessitates the complete configuration data of two consecutive time steps. But the action contribution $S_{k+1}$ governing a local evolution move $k\rightarrow k+1$ 
only contains data which is directly involved in the evolution move. 
\item[(ii)] Subsets of data in $\cq_k$ and $\cq_{k+1}$ coincide. Accordingly, the contribution $S_{k+1}$ alone does {\it not} suffice to define the Legendre transformations (definitions \ref{def_postleg} and \ref{def_preleg}) from $\cq_{k}\times\cq_{k+1}$ to the phase spaces $\cp_k,\cp_{k+1}$. 
\end{itemize}
In the Hamiltonian picture 
the configuration data of the second time step are replaced by the momenta which are defined at the same time step as the configuration data. 
That is, in the Hamiltonian picture we need only one time step to encode the canonical data; if subsets of variables coincide at consecutive steps, the canonical data merely need to be appropriately updated in the course of the evolution move (while the momenta themselves are defined using the global structure of section \ref{sec_globcan}).

In section \ref{sec_globcan}, we saw that the action contribution $S_{n+1}$ of a global evolution move $n\rightarrow n+1$ is the generating function of the first kind of the corresponding global Hamiltonian time evolution map $\ch_n$. In this section, we shall introduce the generating functions of the local Hamiltonian time evolution maps, which, for distinction, we shall denote by $\fh_k:\cp_k\rightarrow\cp_{k+1}$, corresponding to local time evolution moves $k\rightarrow k+1$. 

Firstly, suppose some set of variables $x^b$ appears in several time steps, e.g., such that $x^b_{k+1} \equiv x^b_{k}$. Assume $x^b$ is not involved in the local time evolution such that
\ba\label{c1}
 x^b_{k+1}\,=\,x^b_k \,,\q  \q\q p_b^{k+1}\,=\,p_b^k  .
 \ea
These evolution equations cannot be generated by the action, since $S_{k+1}$ depends on neither $x^b_k$ nor on $x^b_{k+1}$, given that $x^b$ is not included in the local dynamics. Instead, the identity transformation---either governed by a generating function of the second (depending on old configuration and new momentum variables) or the third kind (depending on new configuration and old momentum variables)---is appropriate:
 \ba\label{c2}
 &&G_2(x^k_b,p^{k+1}_b)\,=\,-x^b_k p_b^{k+1} ,\q\q\q \,p_b^k\,=\,-\frac{\partial G_2}{\partial x_b^k}\,=\,p_b^{k+1} ,\q\,\,\,
 x_b^{k+1}=\,-\frac{\partial G_2}{\partial p_b^{k+1}}\,=\,x_b^k   \nn\\
 &&G_3(x^{k+1}_b,p^{k}_b)\,=\,\;\, x^b_{k+1} p_b^{k} ,\q\q\,\,\, p_b^{k+1}\,=\,\;\,\frac{\partial G_3}{\partial x^b_{k+1}}\,=\,p_b^{k} ,\q\q\;\,\;\,
 x_b^{k}\,=\,\,\;\frac{\partial G_3}{\partial p^b_{k}}\,=\,x^b_{k+1}\, . \nn
 \ea

Next, let us consider the case in which some configuration variables do {\it not} evolve $x^e_{k}=x^e_{k+1}$, yet in which either $x^e_k$ or $x^e_{k+1}$ appear in $S_{k+1}$ (but not both for the same index $e$) such that their conjugate momenta are transformed. The additivity of the action implies either  
 \ba\label{c3}
  p_e^{k}\,=\,p_e^{k+1}-\frac{\partial S_{k+1}(x_k)}{\partial x_{k}^e} \q \q \text{or}\q\q
 p_e^{k+1}\,=\,p_e^k+\frac{\partial S_{k+1}(x_{k+1})}{\partial x_{k+1}^e} ,
 \ea
which we shall call (\ref{c3}) {\it momentum updating}. Both ways of {\it momentum updating} (\ref{c3}) can be summarized into 
\ba
p_e^{k+1}=p_e^k + \frac{\partial S_{k+1}(x_{k+1})}{\partial x_{k+1}^e} + \frac{\partial S_{k+1}(x_k)}{\partial x_{k}^e}  .
\ea
because one of the derivatives will always be zero.

The generating function for {\it momentum updating} is of second or third kind: either add $G_2$ or $G_3$, respectively, to $S_{k+1}$, the latter of which either depends only on the old configuration variables $x_k^e$ or only on the new configuration variables $x_{k+1}^e$. 
The same construction also applies to the case of varying numbers of canonical variables when formulated on extended phase spaces, as we shall see in the following subsection.

\subsection{{ Preservation of symplectic structures and constraints under momentum updating}}\label{sec_momup}

The preservation of the symplectic structure, as proven in theorem \ref{thm_presymp} for singular systems and discussed in section \ref{sec_psext} for varying phase spaces, holds for {\it global} evolution moves where variables at $n$ and $(n+1)$ do {\it not} coincide. 
We still have to investigate the preservation of the symplectic structure under local evolution moves which involve coinciding subsets of variables and proceed by {\it momentum updating} for {\it all} canonical pairs on the extended phase space. To circumvent the problem of changing phase space dimensions, we shall work on extended phase spaces $\bar{\cp}_k,\bar{\cp}_{k+1}$ as introduced in section \ref{sec_psext}. 

At this stage we shall distinguish {four} types of local evolution moves. { Any other conceivable kind of local moves can be treated in complete analogy.}\footnote{In particular, we shall ignore the case where `bulk' variables are involved in the local move. These cases can be treated via the recipe provided in section \ref{sec_bulk}.} Let us detail the corresponding evolution equations case by case and subsequently consider the symplectic structure.


~\\
\underline{\bf Type I:} The move introduces `new variables' but does not remove `old variables.' Assume that $K$ `new variables' arise.  Accordingly, extend the phase space at time $k$ by $K$ pairs $(x_k^n,p^k_n)$ which correspond to the $K$ canonical pairs $(x_{k+1}^n,p^{k+1}_n)$ at time $k+1$. Additionally, pairs $(x^b,p_b)$ occur which do not change during this evolution move and variables $(x^e,p_e)$ for which only the momenta are updated. The local Hamiltonian evolution map $\fh_k$ for type I is thus given by the following momentum updating
\ba\label{anh1}
x^b_k&=&x^b_{k+1}  \, ,\q\q\q p^{k+1}_b\,=\,p^k_b \; , \\
x^e_k&=&x^e_{k+1}\,,\q\q\q p^{k+1}_e\,=\,p^k_e+\frac{\partial S_{k+1}(x^e_{k+1}, x^n_{k+1})}{\partial x^e_{k+1}}  
\; ,    \label{anh1b}  \\
p^{k}_n&=&0 \,,\q\q\q\q \,\,\,p^{k+1}_n\,=\, \frac{\partial S_{k+1}(x^e_{k+1},x^n_{k+1})  }{\partial x^n_{k+1}}\,. \label{anh1c} 
\ea
We choose $S_{k+1}$ to be a function of the variables at time $k+1$. Equations (\ref{anh1c}) contain $K$ constraints { $C_n^k=p^k_n$ which are simultaneously pre-- and post--constraints, as well as} $K$ post--constraints ${}^+C_n^{k+1}=p^{k+1}_n- \frac{\partial S_{k+1}(x^e_{k+1},x^n_{k+1})  }{\partial x^n_{k+1}}$. Both $x^n_k,x^n_{k+1}$ are undetermined.

\begin{Example}
\emph{Consider a scalar field living on the vertices of a 2D space-time triangulation. The evolution move given by the 1--2 Pachner move (see figures \ref{fig_12}, \ref{fig_12b}) introduces one new vertex $v$ and thus one `new' field variable $\phi_{k+1}^v$ at $k+1$, while preserving all other vertices of $\Sigma_k$. Hence, this move is of type I. Using (\ref{saction}), the momentum updating map reads 
\ba
\phi^b_k&=&\phi^b_{k+1}  \, ,\q\q\q \pi^{k+1}_b\,=\,\pi^k_b\,,\q\q\q\q\q\q\q\q\q\q\q\q\q\q\q b=1,4,5 \; ,\nn \\
\phi^e_k&=&\phi^e_{k+1}\,,\q\q\q \pi^{k+1}_e\,=\,\pi^k_e+\phi^e_{k+1}-\f{1}{2}\left(\phi^v_{k+1}+\phi^{e+1}_{k+1}\right)\,,\q\q\q\, e=2,3\; ,    \nn  \\
\pi^{k}_v&=&0 \,,\q\q\q\q \,\,\,\pi^{k+1}_v\,=\, \phi^v_{k+1}-\f{1}{2}\left(\phi^2_{k+1}+\phi^3_{k+1}\right)\,, \nn
\ea
where $e+1=3$ if $e=2$ and $e+1=2$ if $e=3$ (see figure \ref{fig_12} for further notation).
}

\emph{Furthermore, the 1--3 Pachner move in 3D Regge Calculus and the 1--4 and 2--3 Pachner moves in 4D Regge Calculus are of type I (see \cite{Dittrich:2011ke} for details).}
\begin{center}
\begin{figure}[htbp!]
\psfrag{s}{\small$\Sigma_k$}
\psfrag{sk}{\small$\Sigma_{k+1}$}
\psfrag{v}{\small$v$}
\psfrag{fv}{\footnotesize$\phi^v_{k+1}$}
\psfrag{1}{\footnotesize $\phi^1_{k+1}$}
\psfrag{2}{\footnotesize$\phi^2_{k+1}$}
\psfrag{3}{\footnotesize$\phi^3_{k+1}$}
\psfrag{4}{\footnotesize$\phi^4_{k+1}$}
\psfrag{5}{\footnotesize$\phi^5_{k+1}$}
\psfrag{a}{\footnotesize$\phi^1_k$}
\psfrag{b}{\footnotesize$\phi^2_k$}
\psfrag{c}{\footnotesize$\phi^3_k$}
\psfrag{d}{\footnotesize$\phi^4_k$}
\psfrag{e}{\footnotesize$\phi^5_k$}
\hspace{-.5cm}
\begin{subfigure}[b]{.22\textwidth}
\centering
\includegraphics[scale=.5]{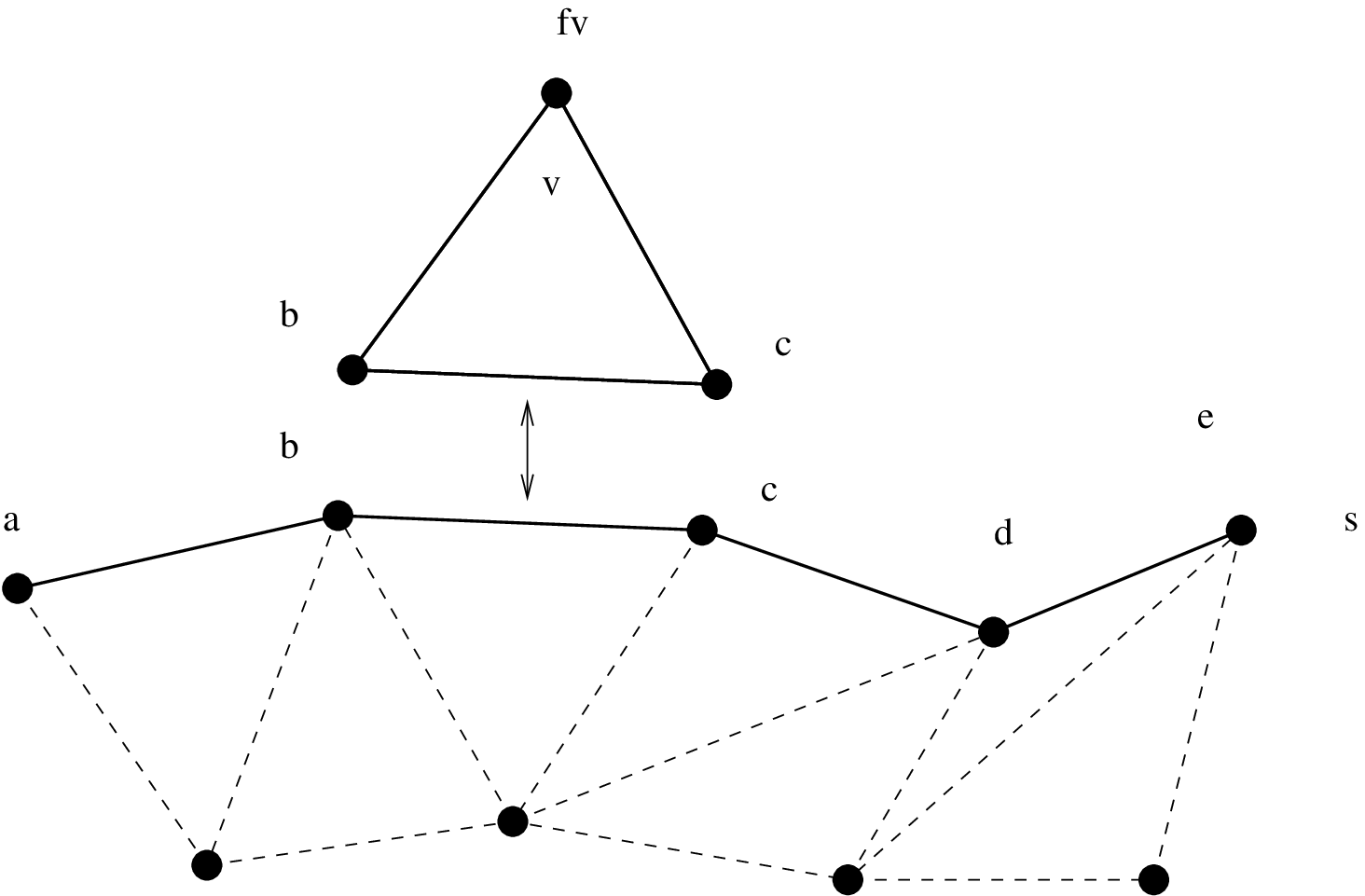}
\centering
\caption{\small }
\end{subfigure}
\hspace*{4.8cm}
\begin{subfigure}[b]{.22\textwidth}
\centering
\includegraphics[scale=.5]{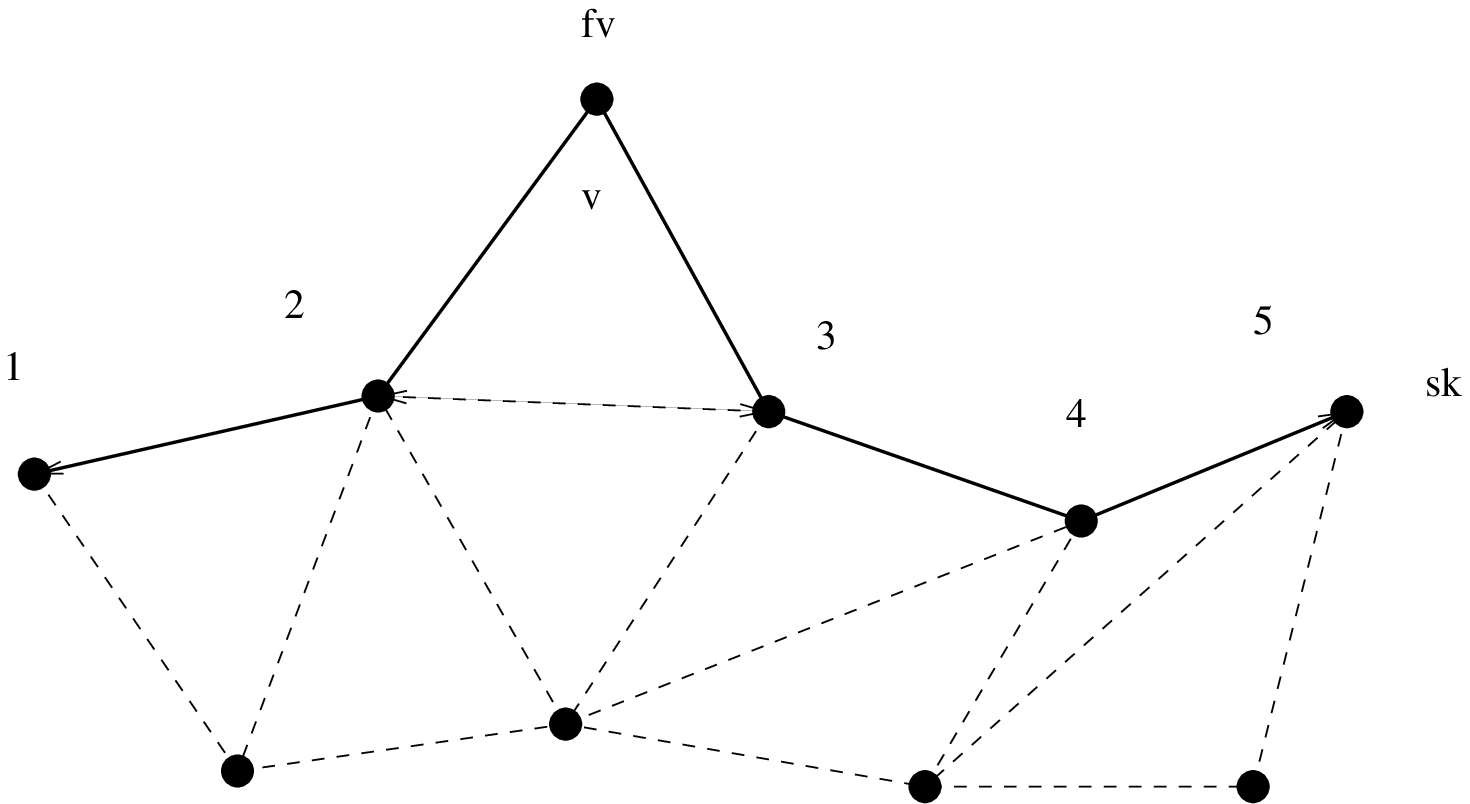}
\caption{\small }
\end{subfigure}
\caption{\small The 1--2 Pachner move glues a triangle onto the 1D hypersurface $\Sigma_k$ and introduces one new field variable $\phi^v_{k+1}$ at vertex $v$ at step $k+1$ and is of type I (see figure \ref{fig_12b} for the 1D perspective).}\label{fig_12}
\end{figure}
\end{center}
\begin{SCfigure}
\psfrag{k}{}
\psfrag{k1}{}
\includegraphics[scale=.5]{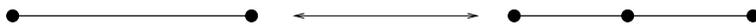}
\caption{\small The 1--2 and 2--1 Pachner moves in 1D.}\label{fig_12b}
\end{SCfigure}
\end{Example}

\vspace*{-1cm}
\underline{\bf Type II:} The move removes `old variables' but does not introduce `new variables'. This is the time reverse of type I. If $K$ old variables are removed, extend the phase space at step $k+1$ by $K$ pairs $(x^o_{k+1},p^{k+1}_o)$. We choose $S_{k+1}$ to be a function of the variables of time $k$. The local Hamiltonian evolution or momentum updating map $\fh_k$ for type II reads
\ba\label{case2}
x^b_k&=&x^b_{k+1}  \, ,\q\q\q p^{k+1}_b\,=\,p^k_b \; , \\
x^e_k&=&x^e_{k+1}\,,\q\q\q\q p^{k}_e\,=\,p^{k+1}_e-\frac{\partial S_{k+1}(x^e_{k}, x^o_{k})}{\partial x^e_{k}}  
\; ,    \label{case2b}  \\
p^{k+1}_o&=&0 \,,\q\q\q\q\q \,\,\,p^{k}_o\,=\, -\frac{\partial S_{k+1}(x^e_{k},x^o_{k})  }{\partial x^o_{k}}\,. \label{case2c} 
\ea
Equations (\ref{case2c}) contain $K$ pre--constraints ${}^-C^k_o=p^{k}_o+\frac{\partial S_{k+1}(x^e_{k},x^o_{k})  }{\partial x^o_{k}}$ and $K$ constraints ${}^+C^{k+1}_o=p^{k+1}_o$, {which are simultaneously pre-- and post--constraints}, while $x^o_{k+1}$ remains undetermined. {We shall denote the partial pre--constraint surface defined {\it only} by the ${}^-C^k_o$ in the extended phase space $\bar{\cp}_k$ by $\ck^-_k\subset\bar{\cp}_k$. No further pre--constraints are created in the course of the move.\footnote{Subsequent evolution moves may lead to additional pre--constraints on the same data.}}

\begin{Example}
\emph{Consider, again, the scalar field living on the vertices of a 2D triangulation. The 2--1 Pachner evolution move (see figures \ref{fig_21} and \ref{fig_12b}) is the time reverse of the 1--2 Pachner move and thus of type II. It removes a vertex $v^*$ and, accordingly, an `old' field variable $\phi^{v^*}_k$ from $\Sigma$. Employing (\ref{saction}) yields its momentum updating map
\ba
\phi^b_k&=&\phi^b_{k+1}  \, ,\q\q\q \pi^{k+1}_b\,=\,\pi^k_b\,,\q\q\q\q\q\q\q\q\q\q\q\q\q\q\,\,\, b=1,4 \; ,\nn \\
\phi^e_k&=&\phi^e_{k+1}\,,\q\q\q\q \pi^{k}_e\,=\,\pi^{k+1}_e-\phi^e_{k}+\f{1}{2}\left(\phi^{v^*}_{k}+\phi^{e+1}_{k}\right)\,,\q\q\q e=2,3\; ,   \nn  \\
\pi^{k+1}_{v^*}&=&0 \,,\q\q\q\q\q \,\,\,\pi^{k}_{v^*}\,=\, -\phi^{v^*}_{k}+\f{1}{2}\left(\phi^2_{k}+\phi^3_{k}\right)\,,\nn
\ea
where $e+1=3$ if $e=2$ and $e+1=2$ if $e=3$ (see figure \ref{fig_21} for further notation).}

\emph{Moreover, the 3--1 Pachner move in 3D Regge Calculus and the 3--2 and 4--1 Pachner moves in 4D Regge Calculus are of type II (see \cite{Dittrich:2011ke} for details).}
\begin{center}
\begin{figure}[htbp!]
\psfrag{s}{\small$\Sigma_k$}
\psfrag{sk}{\small$\Sigma_{k+1}$}
\psfrag{v}{\small$v^*$}
\psfrag{fv}{\footnotesize$\phi^{v^*}_{k}$}
\psfrag{1}{\footnotesize $\phi^1_{k+1}$}
\psfrag{2}{\footnotesize$\phi^2_{k+1}$}
\psfrag{3}{\footnotesize$\phi^3_{k+1}$}
\psfrag{4}{\footnotesize$\phi^4_{k+1}$}
\psfrag{5}{\footnotesize$\phi^5_{k+1}$}
\psfrag{a}{\footnotesize$\phi^1_k$}
\psfrag{b}{\footnotesize$\phi^2_k$}
\psfrag{c}{\footnotesize$\phi^3_k$}
\psfrag{d}{\footnotesize$\phi^4_k$}
\psfrag{e}{\footnotesize$\phi^5_k$}
\begin{subfigure}[b]{.22\textwidth}
\centering
\includegraphics[scale=.5]{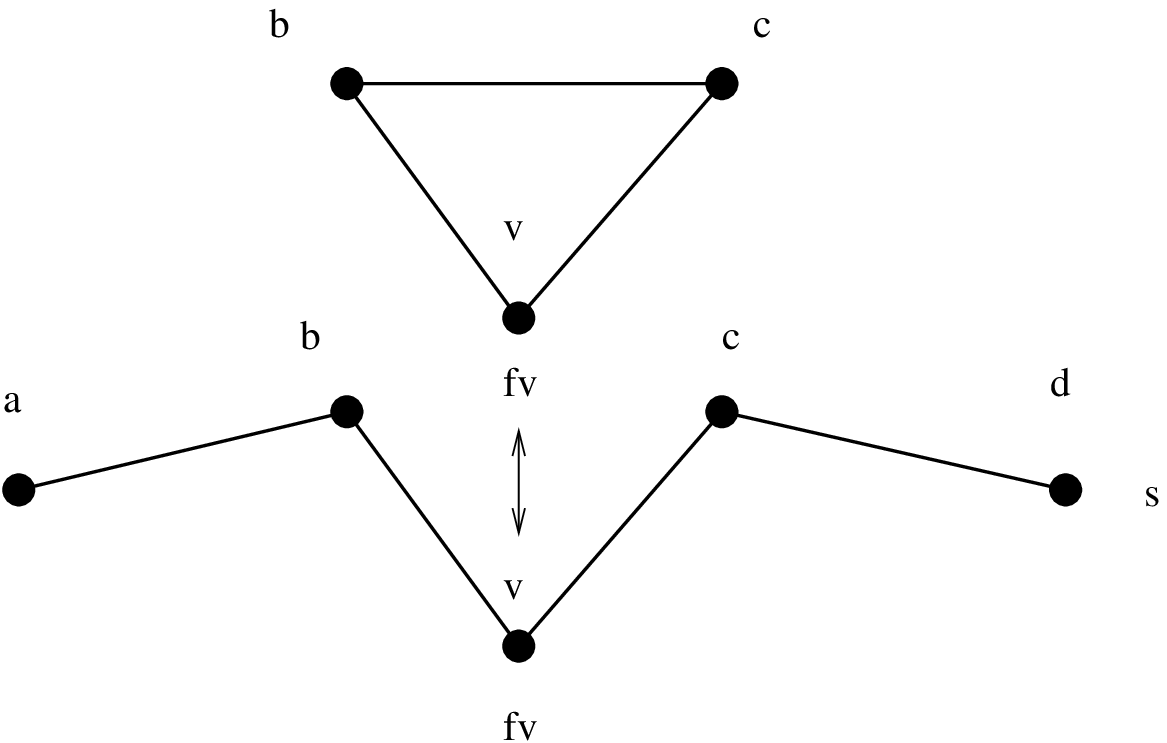}
\centering
\caption{\small }
\end{subfigure}
\hspace*{4.4cm}
\begin{subfigure}[b]{.22\textwidth}
\centering
\includegraphics[scale=.5]{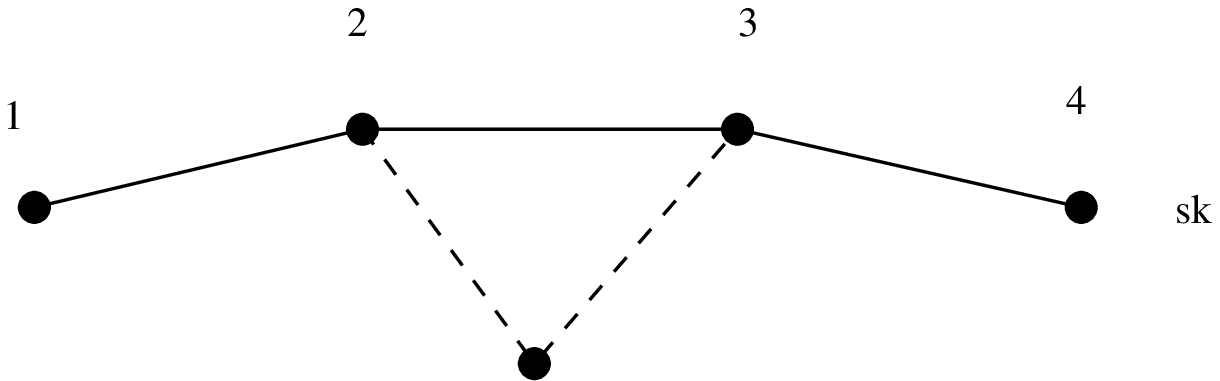}
\vspace*{.005cm}
\caption{\small }
\end{subfigure}
\caption{\small The 2--1 Pachner move glues a triangle onto the 1D hypersurface $\Sigma_k$ such that vertex $v^*$ disappears at step $k+1$ together with the variable $\phi^{v^*}_{k}$ and is of type II (see figure \ref{fig_12b} for the 1D perspective).}\label{fig_21}
\end{figure}
\end{center}
\end{Example}

\vspace*{-1cm}
\underline{\bf Type III:} The move removes $K$ `old variables' and also introduces $K$ `new variables'. 
Extend the initial phase space by $K$ pairs $(x^n_k,p^k_n)$ and the final phase space by $K$ pairs $(x^o_{k+1},p^{k+1}_o)$. The momentum updating map $\fh_k$ is given by
\ba\label{anh21}
x^b_k&=&x^b_{k+1}  \, ,\q\q\q p^{k+1}_b\,=\,p^k_b \, , \\
x^e_k&=&x^e_{k+1}\,,\q\q\q p^{k+1}_e\,=\,p^k_e+\frac{\partial S_{k+1}(  x^e_{k+1},x^o_k,x^n_{k+1})  }{\partial x^e_{k+1}} 
\, ,\\
p^{k}_n&=&0 \,,\q\q\q\q\,\, \,p^{k+1}_n\,=\, \frac{\partial S_{k+1}(  x^e_{k+1},x^o_k,x^n_{k+1}) }{\partial x^n_{k+1}}\,\,
 , \label{anh21c}\\
p^{k+1}_o &=& 0 \,,\q\q\q\q\,\, \,p^{k}_o\q\,=\, -\frac{\partial S_{k+1} (  x^e_{k+1},x^o_k,x^n_{k+1}) }{\partial x^o_{k}}\,
 . \label{anh21d}
\ea
(We could equally well have chosen to let $S_{k+1}$ depend on $x^e_k$, instead of $x^e_{k+1}$.) {The $p^k_n=0$ and the $p^{k+1}_o=0$ each constitute $K$ constraints that result from phase space extensions and are each simultaneously pre-- and post--constraints. Furthermore, if the $K\times K$ matrix $\f{\p^2S_{k+1}}{\p x^o_k\p x^n_{k+1}}$ is of rank $K-\kappa$, then $\kappa$ non--trivial pre--constraints at $k$ and $\kappa$ non--trivial post--constraints at $k+1$ will arise from the second equations in (\ref{anh21c}) and (\ref{anh21d}), respectively. Among the $x^n_{k+1}$, $K-\kappa$ variables will be determined via the second equation in (\ref{anh21d}), while the remaining $\kappa$ variables will remain ({\it a priori}) free. Furthermore, the $x_k^n$ and $x_{k+1}^o$ remain undetermined and can be gauge fixed to arbitrary values, for instance $x^n_k=x^n_{k+1}$ and $x^o_{k+1}=x^o_k$. The above $\kappa$ non--trivial pre--constraints at $k$ alone define the partial pre--constraint surface $\ck^-_k\subset\bar{\cp}_k$.}

\begin{Example}
\emph{Consider a scalar field living on a 2D space-time quadrangulation where each $\Sigma_k$ is given by a `zig-zag line'. Adding a square as a fundamental building block annihilates one `old' vertex $v^*$, introduces a `new' vertex $v$ and preserves the remaining vertices in the move $k\rightarrow k+1$ (see figure \ref{fig_typeiii}). Hence, this evolution move annihilates one field variable $\phi_{k+1}^{v}$, introduces one new field variable $\phi_{k}^{v^*}$ and is of type III. The action of a scalar field on a (rectangular) square reads 
\ba
S_\Box=\sum_{i=1}^4\left((\phi^i)^2-\phi^i\phi^{i+1}\right),\nn
\ea
where $i$ labels the vertices of the square and $i+1=5$ coincides with $i=1$ due to periodicity. The corresponding momentum updating map is thus given by
\ba
\phi^b_k&=&x^b_{k+1}  \, ,\q\q\q \pi^{k+1}_b\,=\,\pi^k_b\,,\q\q\q\q\q\q\q\q\q\q\q\q\q b=1,4,5 \, , \nn\\
\phi^e_k&=&x^e_{k+1}\,,\q\q\q \pi^{k+1}_e\,=\,\pi^k_e+2\phi^e_{k+1}-\phi^{v^*}_k-\phi^v_{k+1}\,,\q\q\q\, e=2,3\, ,\nn\\
\pi^{k}_v&=&0 \,,\q\q\q\q\,\, \,\pi^{k+1}_v\,=\, 2\phi^v_{k+1}-\phi^2_{k+1}-\phi^3_{k+1}\,\,
 ,\nn \\
\pi^{k+1}_{v^*} &=& 0 \,,\q\q\q\q\q\,\pi^{k}_{v^*}\,=\, -2\phi^{v^*}_k+\phi^2_k+\phi^3_k\,\nn
\ea
(see figure \ref{fig_typeiii} for further notation). Notice that ${}^+C^{k+1}_v:=\pi^{k+1}_v-2\phi^v_{k+1}+\phi^2_{k+1}+\phi^3_{k+1}$ defines a post--constraint at $k+1$, while ${}^-C^k_{v^*}:=\pi^{k}_{v^*} +2\phi^{v^*}_k-\phi^2_k-\phi^3_k$ constitutes a pre--constraint at $k$.}

\emph{Furthermore, the 2--2 Pachner move in 3D Regge Calculus and the tent moves (after solving the tent pole equation) in any dimension are of type III (see \cite{Dittrich:2011ke,Dittrich:2009fb} for details).}
\begin{center}
\begin{figure}[htbp!]
\psfrag{s}{\small$\Sigma_k$}
\psfrag{sk}{\small$\Sigma_{k+1}$}
\psfrag{v}{\small$v$}
\psfrag{vs}{\small$v^*$}
\psfrag{fv}{\footnotesize$\phi^v_{k+1}$}
\psfrag{fvs}{\footnotesize$\phi^{v^*}_k$}
\psfrag{1}{\footnotesize $\phi^1_{k+1}$}
\psfrag{2}{\footnotesize$\phi^2_{k+1}$}
\psfrag{3}{\footnotesize$\phi^3_{k+1}$}
\psfrag{4}{\footnotesize$\phi^4_{k+1}$}
\psfrag{5}{\footnotesize$\phi^5_{k+1}$}
\psfrag{a}{\footnotesize$\phi^1_k$}
\psfrag{b}{\footnotesize$\phi^2_k$}
\psfrag{c}{\footnotesize$\phi^3_k$}
\psfrag{d}{\footnotesize$\phi^4_k$}
\psfrag{e}{\footnotesize$\phi^5_k$}
\begin{subfigure}[b]{.22\textwidth}
\centering
\includegraphics[scale=.5]{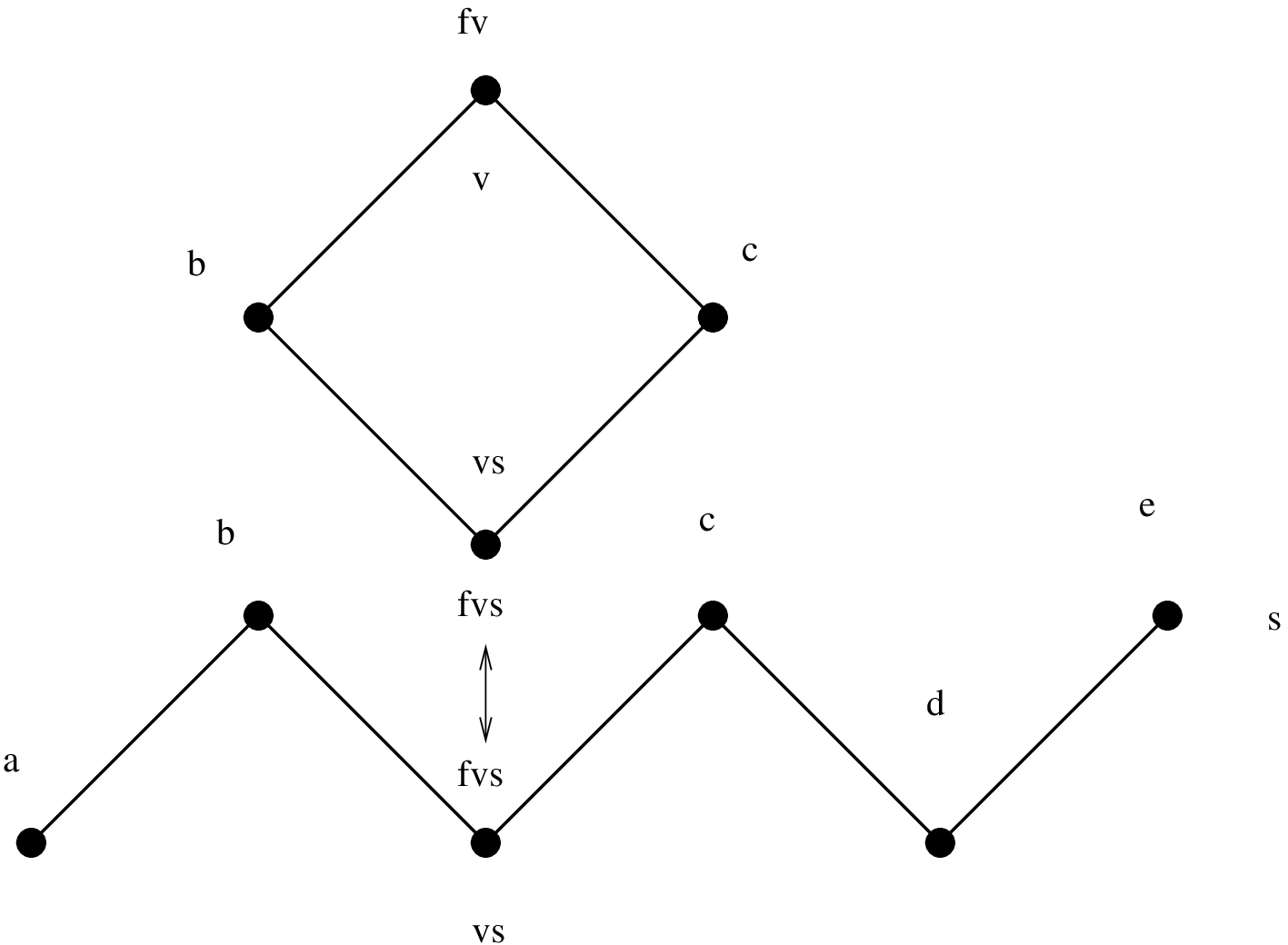}
\centering
\caption{\small }
\end{subfigure}
\hspace*{4.8cm}
\begin{subfigure}[b]{.22\textwidth}
\centering
\includegraphics[scale=.5]{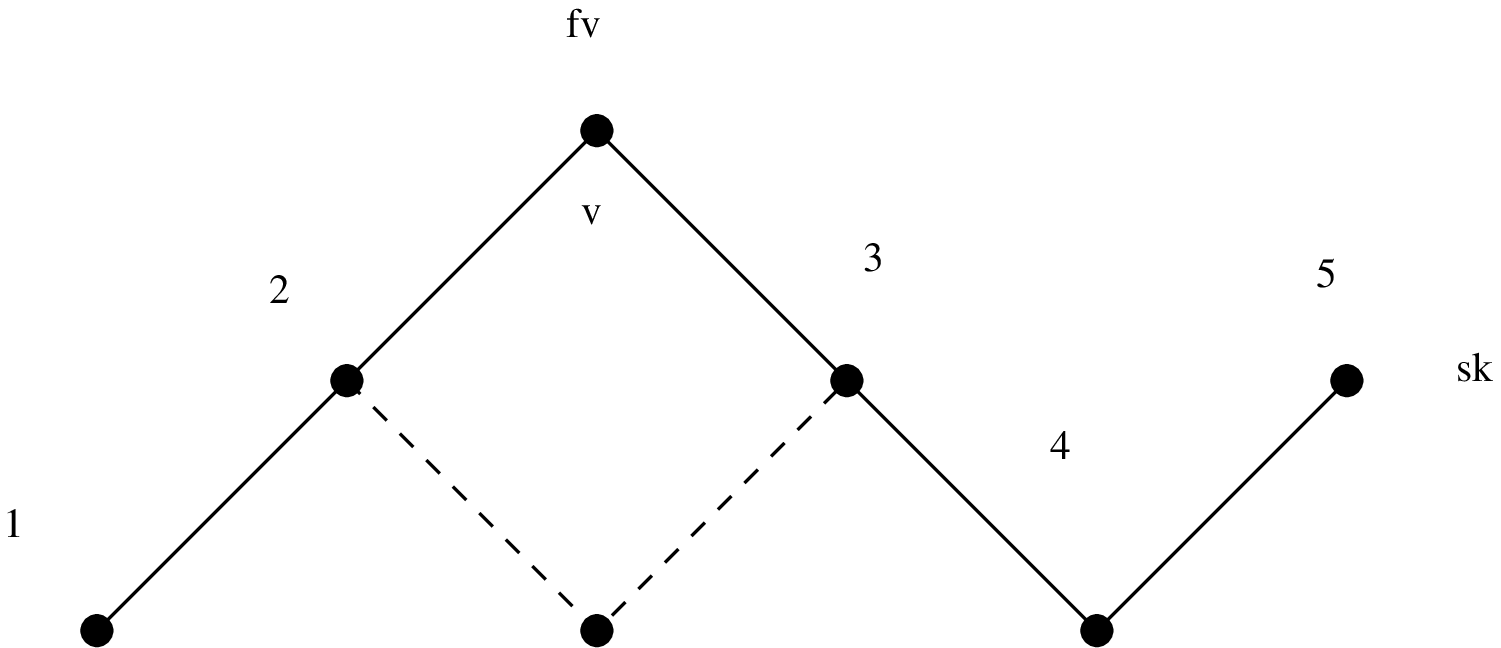}
\caption{\small }
\end{subfigure}
\caption{\small Gluing a square onto the 1D hypersurface $\Sigma_k$ annihilates a vertex $v^*$ and its field variable $\phi^{v^*}_k$ and introduces one new vertex $v$ and its field variable $\phi^v_{k+1}$. This move is of type III.}\label{fig_typeiii}
\end{figure}
\end{center}
\end{Example}


\vspace*{-.6cm}
{\underline{\bf Type IV:} No `new variable' is introduced and no `old variable' is removed in the course of the move. The momentum updating map $\fh_k$ simply reads
\ba
x^b_k&=&x^b_{k+1}  \, ,\q\q\q p^{k+1}_b\,=\,p^k_b \, , \label{typeiv1}\\
x^e_k&=&x^e_{k+1}\,,\q\q\q p^{k+1}_e\,=\,p^k_e+\frac{\partial S_{k+1}(  x^e_{k+1})  }{\partial x^e_{k+1}} .\label{typeiv2}
\ea
($S_{k+1}$ may equally well depend on $x^e_k$, rather than $x^e_{k+1}$.) No pre-- or post--constraints arise.

\begin{Example}
\emph{The 2--2 Pachner evolution move for a scalar field living on the vertices of a 3D triangulation is of type IV. No vertex is annihilated or introduced so neither is any field variable (see figure \ref{fig_typeiv}). We abstain from detailing the corresponding momentum updating map. }
\begin{center}
\begin{figure}[htbp!]
\psfrag{s}{$\Sigma_k$}
\psfrag{sk}{$\Sigma_{k+1}$}
\psfrag{fv}{\footnotesize$\phi^v_{k+1}$}
\psfrag{fvs}{\footnotesize$\phi^{v^*}_k$}
\psfrag{1}{\footnotesize $\phi^1_{k+1}$}
\psfrag{2}{\footnotesize$\phi^2_{k+1}$}
\psfrag{3}{\footnotesize$\phi^3_{k+1}$}
\psfrag{4}{\footnotesize$\phi^4_{k+1}$}
\psfrag{5}{\footnotesize$\phi^5_{k+1}$}
\psfrag{6}{\footnotesize$\phi^6_{k+1}$}
\psfrag{7}{\footnotesize$\phi^7_{k+1}$}
\psfrag{8}{\footnotesize$\phi^8_{k+1}$}
\psfrag{9}{\footnotesize$\phi^9_{k+1}$}
\psfrag{a}{\footnotesize$\phi^1_k$}
\psfrag{b}{\footnotesize$\phi^2_k$}
\psfrag{c}{\footnotesize$\phi^3_k$}
\psfrag{d}{\footnotesize$\phi^4_k$}
\psfrag{e}{\footnotesize$\phi^5_k$}
\psfrag{f}{\footnotesize$\phi^6_k$}
\psfrag{g}{\footnotesize$\phi^7_k$}
\psfrag{h}{\footnotesize$\phi^8_k$}
\psfrag{i}{\footnotesize$\phi^9_k$}
\begin{subfigure}[b]{.21\textwidth}
\centering
\includegraphics[scale=.5]{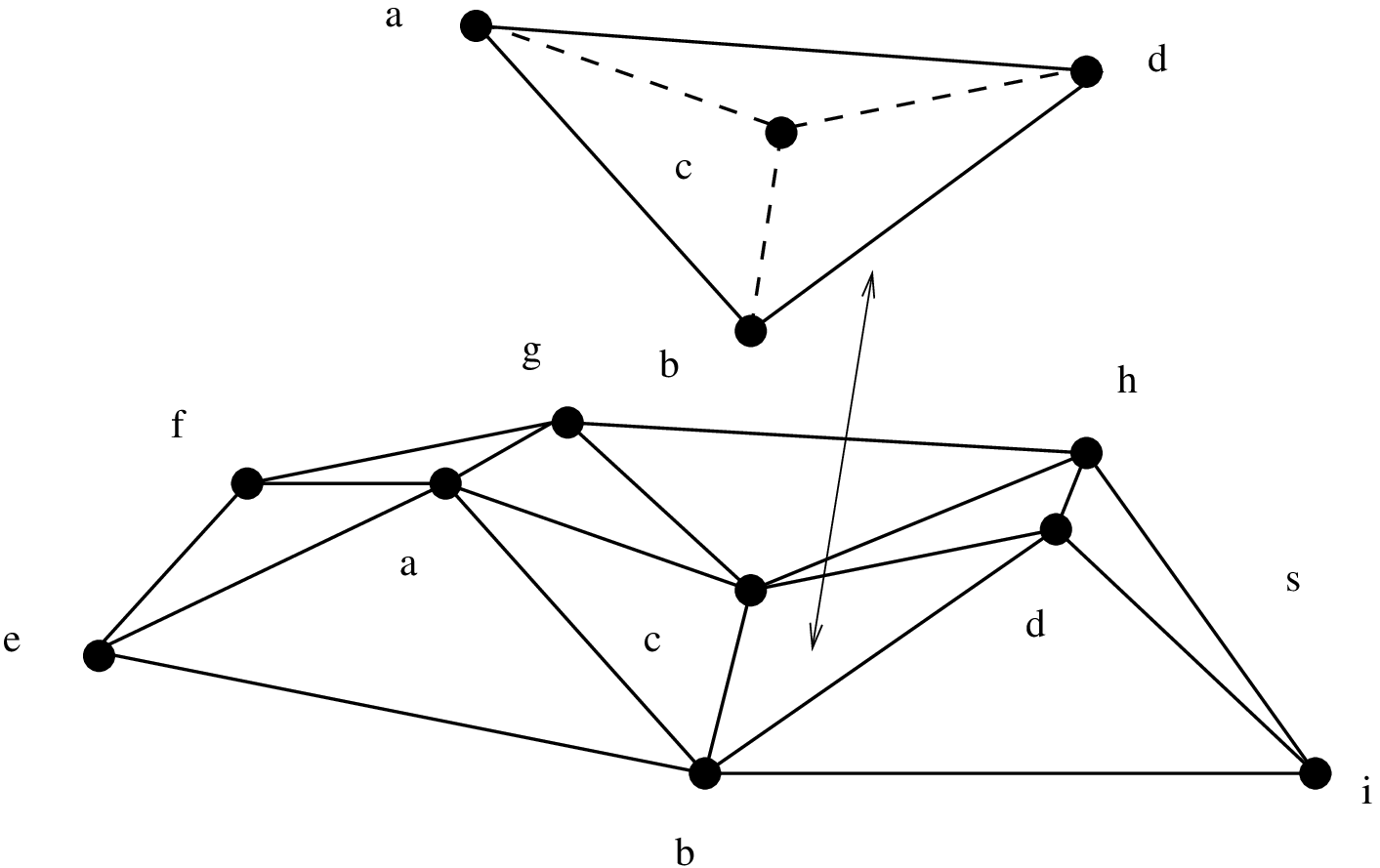}
\centering
\caption{\small }
\end{subfigure}
\hspace*{4.5cm}
\begin{subfigure}[b]{.21\textwidth}
\centering
\includegraphics[scale=.5]{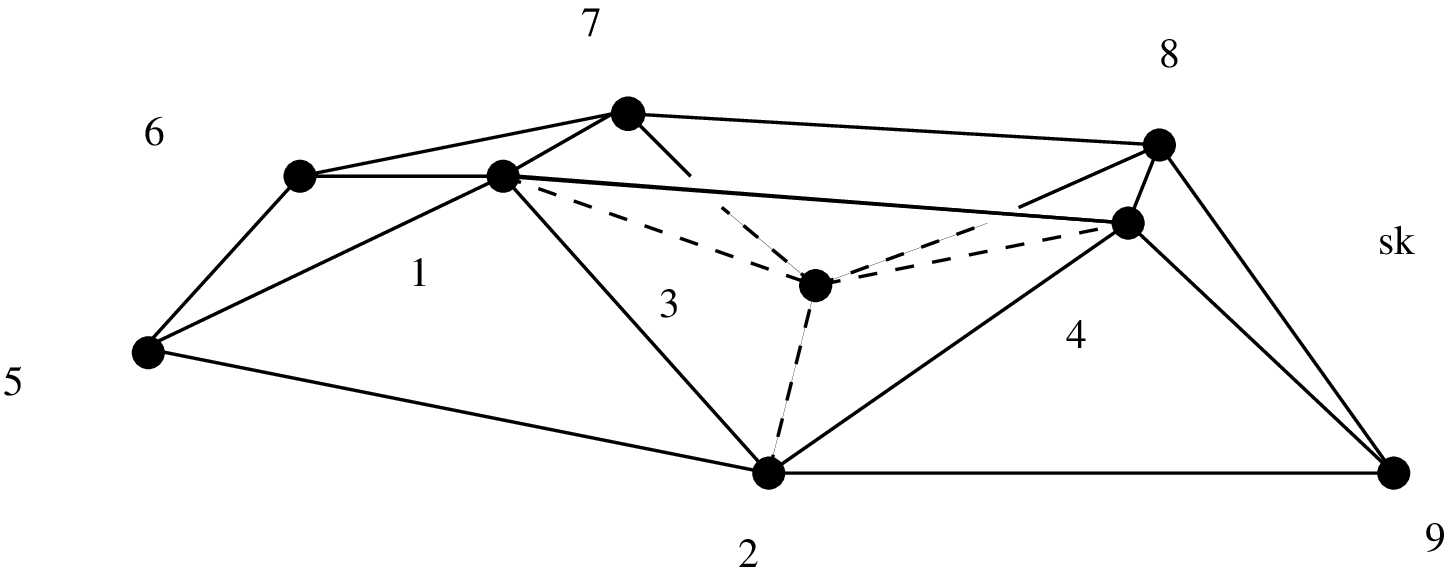}
\caption{\small }
\end{subfigure}
\caption{\small Performing a 2--2 Pachner move in the 2D hypersurface $\Sigma$ corresponds to gluing a tetrahedron onto the hypersurface as shown. This move neither removes nor introduces any vertices or field variables and is of type IV. However, the momenta conjugate to $\phi^e_k$, $e=1,\ldots,4$, must be udpated.}\label{fig_typeiv}
\end{figure}
\end{center}
\end{Example}


\vspace*{-.5cm}
The following theorem shows that {\it momentum updating} 
preserves {the post--constraints and} the symplectic structure restricted to the post--constraint surface in the move $k\rightarrow k+1$. 
However, for types II and III this constraint surface at $k$ is crucially further reduced by intersection with the partial pre--constraint surface $\ck^-_k$.{\footnote{For types I and IV $\cc^+_k\cap\ck^-_k\equiv\cc^+_k$ because no pre--constraints independent of the post--constraints arise in the move. In particular, for type IV $\ck^-_k\equiv\bar{\cp}_k$.}} 
Hence, {\it momentum updating}---as presently formulated---is a pre--symplectic transformation. In fact, there also exists an alternative way of formulating {\it momentum updating} as a canonical transformation on the full extended phase space which we shall briefly introduce in appendix \ref{sec_momupcan}.


\begin{Theorem}\label{thm_momup}
Let $\omega_k,\omega_{k+1}$ be the symplectic forms on the (extended) phase spaces $\bar{\cp}_k,\bar{\cp}_{k+1}$ and let $\cc^+_k,\cc^+_{k+1}$ be the {\it post--constraint surfaces} at steps $k$ and $k+1$, respectively. 
The \emph{momentum updating map} $\fh_k$ preserves the symplectic structure as follows
\ba
{\fh}_k^*(\iota_{k+1})^*\omega_{k+1}=(\iota_{k})^*\omega_{k},\nn
\ea
where for  
~\\
\emph{\underline{\bf Types I}} and \emph{\underline{\bf IV:}} $\fh_k:\cc^+_k\rightarrow\cc^+_{k+1}$ and $\iota_{k/k+1}:\cc^+_{k/k+1}\hookrightarrow\bar{\cp}_{k/k+1}$ are embedding maps.
~\\
\emph{\underline{\bf Types II}} and \emph{\underline{\bf III:}} $\fh_k:\cc^+_k\cap\ck^-_k\rightarrow\cc^+_{k+1}$ and $\ck^-_k\subset\bar{\cp}_k$ is the partial {\it pre--constraint surface} {at $k$} 
{and $\iota_k:\cc^+_k\cap\ck^-_k\hookrightarrow\bar{\cp}_k$ and $\iota_{k+1}:\cc^+_{k+1}\hookrightarrow\bar{\cp}_{k+1}$ are embedding maps.

In particular, $\fh_k$ maps all post--constraints on $\bar{\cp}_k$ to post--constraints on $\bar{\cp}_{k+1}$.}
\end{Theorem}


\begin{proof} The proof is given in appendix \ref{app_momup}.
\end{proof}

{{\bf Consequence:} The preceding theorem has important repercussions for the preservation of the rank of the symplectic structure and the number of constraints throughout the \emph{entire} discrete evolution. While the local evolution moves of type I and IV preserve the symplectic structure restricted to the post--constraint surfaces before and after the move, local moves of type II and III only preserve the symplectic structure further reduced by the partial pre--constraint surface $\ck^-_k$. This, in particular, means that type II and III moves will, in general, \emph{not} preserve the post--constraint surfaces dimension. Rather, \emph{the number of post--constraints at $k+1$ can only be equal or higher than the number of post--constraints at $k$}:\footnote{The 2--2 Pachner move in 3D Regge Calculus---which is of type III---actually does not lead to any non--trivial pre--constraints at $k$ (see \cite{Dittrich:2011ke} for details). That is, $\kappa=0$ for this move and $\ck^-_k\equiv\bar{\cp}_k$ such that the number of post--constraints at $k$ and $k+1$ are identical.}  the number is equal if the pre--constraints defining $\ck^-_k$ are such that they coincide with post--constraints at $k$ or if they are of second class together with the post--constraints at $k$ (second class constraints do not further reduce the rank of the symplectic form \cite{Henneaux:1992ig}). The number increases if some of the pre--constraints defining $\ck^-_k$ do \emph{not} coincide with post--constraints at $k$ and neither are rendered second class. As a result, \emph{on an evolving slice, the number of post--constraints will either remain constant or grow, but cannot decrease}.

This also has consequences for the number of constraints at some fixed global step $n$. Theorem \ref{thm_presymp} showed that a \emph{global} evolution between $n-1$ and $n$ preserves the symplectic structures restricted to $\cc^-_{n-1}$ and $\cc^+_{n}$ because the number of pre-- and post--constraints coincided. In conjunction with the present theorem \ref{thm_momup}, it implies that, if step $n$ is evolved forward by a local evolution move $k\rightarrow k+1$, where the local step $k$ now coincides with the global step $n$, the symplectic structure restricted to the pre--constraint surface at $(n-1)$ and the new post--constraint surface at $k+1$ must again be preserved. However, if the local move was of type II or III, it may have increased the number of post--constraints at $k+1$ as compared to $n=k$. Consequently, also the number of pre--constraints at $(n-1)$ must actually have increased through the evolution from $n=k$ to $k+1$. That is, in general the number of constraints at fixed $k$ can only stay constant or increase by further evolution. This is a consequence of solving equations of motion in the course of the moves; solving more equations of motion can lead to more, but not less constraints because the equations of motion can act as secondary constraints. This has important ramifications for the notion of propagating degrees of freedom and the reduced phase space. We shall discuss all of this amply in the next section \ref{sec_conrole}.}


%

\section{Constraint analysis in the discrete}\label{sec_conrole}

A discrete time evolution proceeds in discrete steps and therefore, in contrast to the continuum, is {\it not}  generated by a set of constraints via a Poisson bracket structure which necessarily has an infinitesimal action.
Rather, it is the time evolution moves which generate the discrete evolution. The exception are theories in which the diffeomorphism symmetry of the continuum is preserved as a vertex translation symmetry \cite{Dittrich:2008pw,Bahr:2009ku,Dittrich:2012qb,Dittrich:2009fb} (and one considers a time evolution which preserves the connectivity of the lattice). This is the case for 3D Regge gravity and also for 4D Regge gravity at vertices embedded in a flat geometry. In general, however, the symmetry is broken in  4D Regge gravity \cite{Bahr:2009ku} such that continuous time evolution has to be replaced by discrete steps, as also proposed in the consistent discretization program \cite{Gambini:2002wn,DiBartolo:2004cg,Gambini:2005vn}. Understanding the complicated status of constraints and symmetries in discrete 4D Regge gravity is one of the motivations for the present work.

Although generally not generating the time evolution, the constraints in the discrete should otherwise assume similar roles to those of their continuum analogues \cite{Henneaux:1992ig}, namely:
\begin{table}[h!]
\begin{center}
\begin{tabular}{ c l }
\hline
  \multicolumn{2}{c}{role of constraints in the continuum}   \vspace*{.1cm}\\
\hline  \hline
    (i) & guarantee correct dynamics\\
  (ii) & generate symmetries\\
  (iii) & classify degrees of freedom\\
  (iv) & generate time evolution\\
  \hline
\end{tabular}
\hspace*{1.5cm}\begin{tabular}{ c l }
\hline
  \multicolumn{2}{c}{role of constraints in the discrete}  \vspace*{.1cm} \\
  \hline  \hline
    (i) & guarantee correct dynamics\\
  (ii) & generate symmetries\\
  (iii) & classify degrees of freedom\\
  &\\
    \hline
\end{tabular}
\end{center}
\end{table}

It is the goal of the present section to show that this is, indeed, the case. To this end, a constraint analysis for variational discrete systems, analogous to the continuum Dirac procedure \cite{Dirac,Henneaux:1992ig}, needs to be developed. We shall discuss roles (i)--(iii) of the constraints in the discrete below after first considering the preservation of constraints by the discrete evolution. For this purpose it is necessary to discuss global evolution moves, but towards the end we shall also refer to local ones. Since we describe systems with evolving phase spaces on extended phase spaces of equal dimension, no generality is lost by restricting to singular systems where $\dim\cq_n=Q$ $\forall\,n$.


\subsection{Preservation of the constraints}\label{sec_conmatch}

In the previous sections we have seen that {\it a priori} free parameters $\lambda_n$ (or {\it a posteriori} free parameters $\mu_n$) and, consequently, arbitrariness in the canonical evolution arise in the presence of post--constraints (or pre--constraints). However, some of this {\it a priori} ({\it a posteriori}) arbitrariness may get fixed {\it a posteriori} ({\it a priori}) by the condition of preservation of the constraints. 

\begin{center}
\underline{\bf Continuum}
\end{center}

Firstly recall that in the continuum the free Lagrange multipliers $\lambda^m$ of some primary constraints $\phi_m$ may become fixed by the condition of preservation of the constraints under evolution, which reads $\dot{\phi}_m=\{\phi_m,H+\lambda^{m'}\phi_{m'}\}\simeq 0$, where $H$ is the (non--vanishing) true Hamiltonian of the system and $\simeq$ refers to the fact that this equation needs only to hold weakly, i.e.\ on the constraint hypersurface \cite{Henneaux:1992ig}. This condition can \begin{itemize}\parskip -.5mm
\item[(a)] be automatically satisfied, in which case no new condition arises and the $\lambda^m$ remain free, 
\item[(b)] lead to secondary constraints which are independent of the $\phi_m$ and $\lambda^m$, or 
\item[(c)] lead to restrictions on the $\lambda^m$ in which case some of the {\it a priori} free parameters get fixed.
\end{itemize} Case (c) is only possible in the presence of second class constraints. A similar, yet slightly different situation arises in the discrete where time evolution is {\it not} generated by a total Hamiltonian and one has to cope with {\it two} constraint surfaces in a given phase space.

\begin{center}
\underline{\bf Discrete}
\end{center}

Since time evolution in the discrete is generated by the time evolution moves and not via the Poisson structure, we need to discuss the behaviour of the constraints under the former in order to discuss their preservation.

For translation invariant systems\footnote{By translation invariant systems we mean systems governed by an action such that $S_n(x_{n-1},x_n)$ is in form identical for all $n$.} one requires preservation of the constraints under time evolution. In this case the primary pre/post--constraints are the same for every time step, but it is possible that $\cc^-_n\neq\cc^+_n$. One has to make sure that the image of the total constraint hypersurface, namely the intersection of the post-- with the pre--constraint hypersurface at time $n$, $\cc_n:=\cc^-_n\cap\cc^+_n$, is a subset of the (same) constraint hypersurface at time $(n+1)$, i.e.\ that the constraints are preserved under time evolution. Otherwise, one has to add further (secondary) constraints, which for translation invariant systems again will be the same at each time step. 
A related discussion of constraints for these special systems has been given in \cite{DiBartolo:2004cg} in the context of  consistent discretizations.

For a 
non--translation invariant system, a preservation of constraints in the above narrow sense does not need to hold. 
Consider a (bare) evolution move $(n-1)\rightarrow n$ and another (bare) evolution move $n\rightarrow(n+1)$. Recall that $\ch_{n-1}:\cc^-_{n-1}\rightarrow\cc^+_n$ and $\ch_n:\cc^-_n\rightarrow\cc^+_{n+1}$. 
When considering the composition of the two moves, $(n-1)\rightarrow n\rightarrow (n+1)$, we have to solve the equations of motion, i.e.\ match the momenta, at $n$. Consequently, at step $n$ the pre--constraints now also have to hold for the post--momenta and the post--constraints have to hold for the pre--momenta. 
That is, the image of the pre--constraint hypersurface $\cc^-_{n-1}$ under time evolution $\ch_{n-1}$ must be a subset of the total constraint hypersurface $\cc_n:=\cc^-_n\cap\cc^+_n$ at time step $n$. If this is not the case, secondary pre--constraints have to be added at step $(n-1)$ to appropriately restrict the data at $(n-1)$ to map under $\ch_{n-1}$ only into $\cc_{n}$ because not all canonical solutions $(n-1)\rightarrow n$ can be extended to solutions for $(n-1)\rightarrow n\rightarrow(n+1)$. These secondary pre--constraints at time step $(n-1)$ arise from the pre--image under time evolution $\ch_{n-1}$ of the pre--constraint surface $\cc^-_n$.\footnote{They arise from the pre--image of $\cc^-_n$ because the pre--image under $\ch_{n-1}$ of $\cc^+_n$ is by construction already the primary pre--constraint surface $\cc^-_{n-1}$.} Likewise, for consistency of the evolution $(n-1)\rightarrow n\rightarrow (n+1)$, the pre--image of the post--constraint surface $\cc^+_{n+1}$ under $\ch_n$ must be a subset of the total constraint hypersurface $\cc_n$ at $n$. Otherwise, secondary post--constraints at the final step $(n+1)$ arise from the image under time evolution $\ch_n$ of the post--constraint hypersurface $\cc^+_n$. Consequently, constraints can `propagate' forward and backward in discrete time (see also the examples in section \ref{sec_correctdyn}).

\begin{figure}
        \centering
        \begin{subfigure}[b]{0.45\textwidth}
        \psfrag{P}{ $T^*\cq_n$}
      \psfrag{cm}{}
   \psfrag{cp}{}
   \psfrag{c}{ $\cc_n={\cc^+_n}\cap{\cc^-_n}$}              
     \centering
                \includegraphics[width=.8\textwidth]{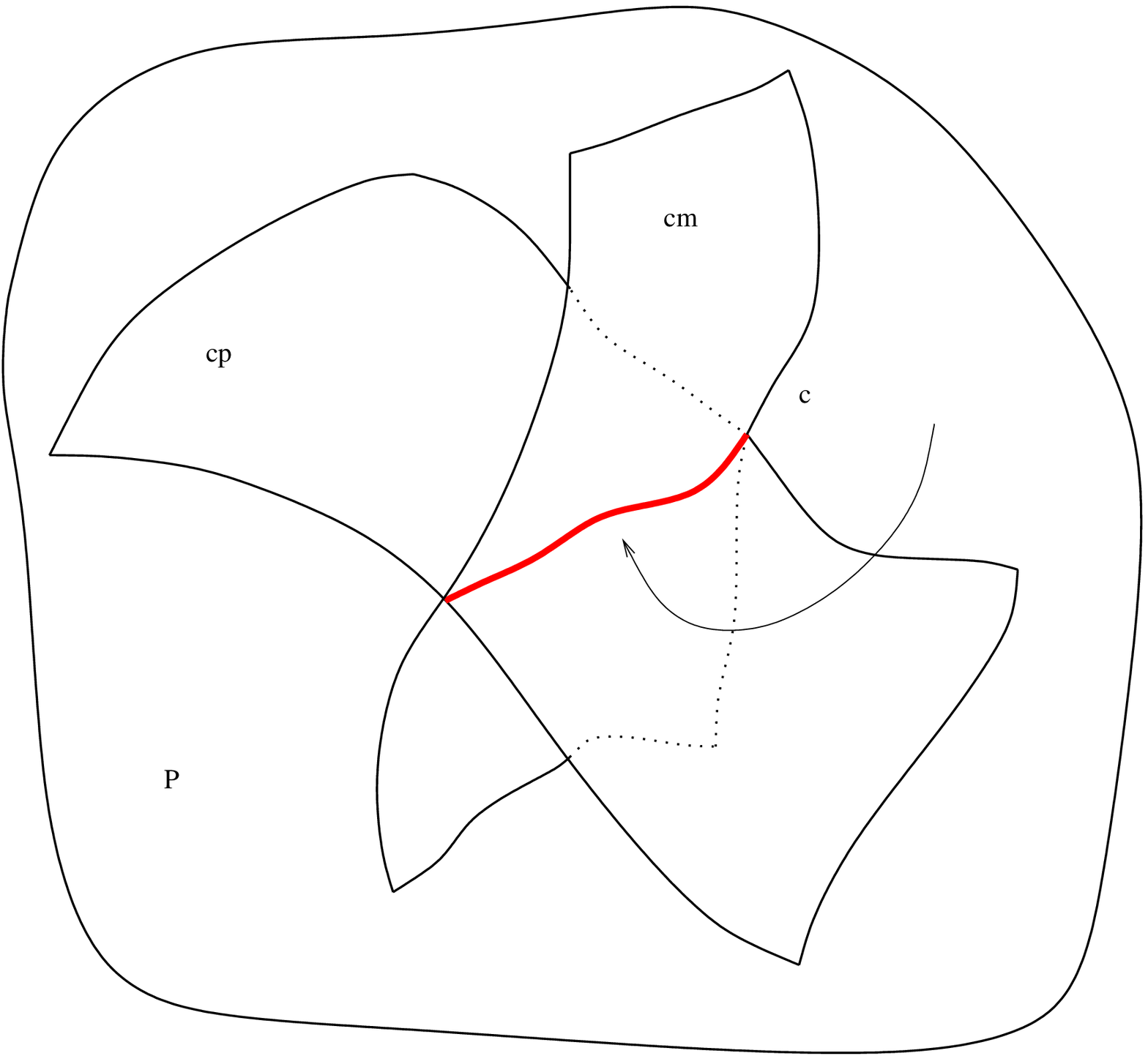}
                \caption{\small }
        \end{subfigure}%
        \hspace*{1cm}
        \begin{subfigure}[b]{0.45\textwidth}
        \psfrag{k}{\tiny $n$}
\psfrag{k1}{\tiny $n-1$}
\psfrag{k2}{\tiny $n+1$}
\psfrag{c2}{\tiny $\textcolor{red}{\cc^+_{n+1}}$}
\psfrag{c1}{\tiny $\textcolor{blue}{\cc^-_{n-1}}$}
\psfrag{ck}{\tiny $\textcolor{red}{\cc^+_n}$}
\psfrag{ckm}{\tiny $\textcolor{blue}{\cc^-_n}$}
\psfrag{p2}{\tiny $x_{n+1},\textcolor{red}{{}^+p^{n+1},\lambda_{n+1}}$}
\psfrag{p1}{\tiny $x_{n-1},\textcolor{blue}{{}^-p^{n-1},\mu_{n-1}}$}
\psfrag{pk}{\tiny $x_{n},\textcolor{red}{{}^+p^{n},\lambda_{n}}$}
\psfrag{pkm}{\tiny $x_{n},\textcolor{blue}{{}^-p^{n},\mu_{n}}$}
\psfrag{m}{\footnotesize match}
\psfrag{h1}{\tiny $\ch_n$}
\psfrag{h2}{\tiny $\ch_{n+1}$}
\psfrag{s1}{\tiny $S_n$}
\psfrag{s2}{\tiny $S_{n+1}$}
        
                \centering
                \includegraphics[width=.8\textwidth]{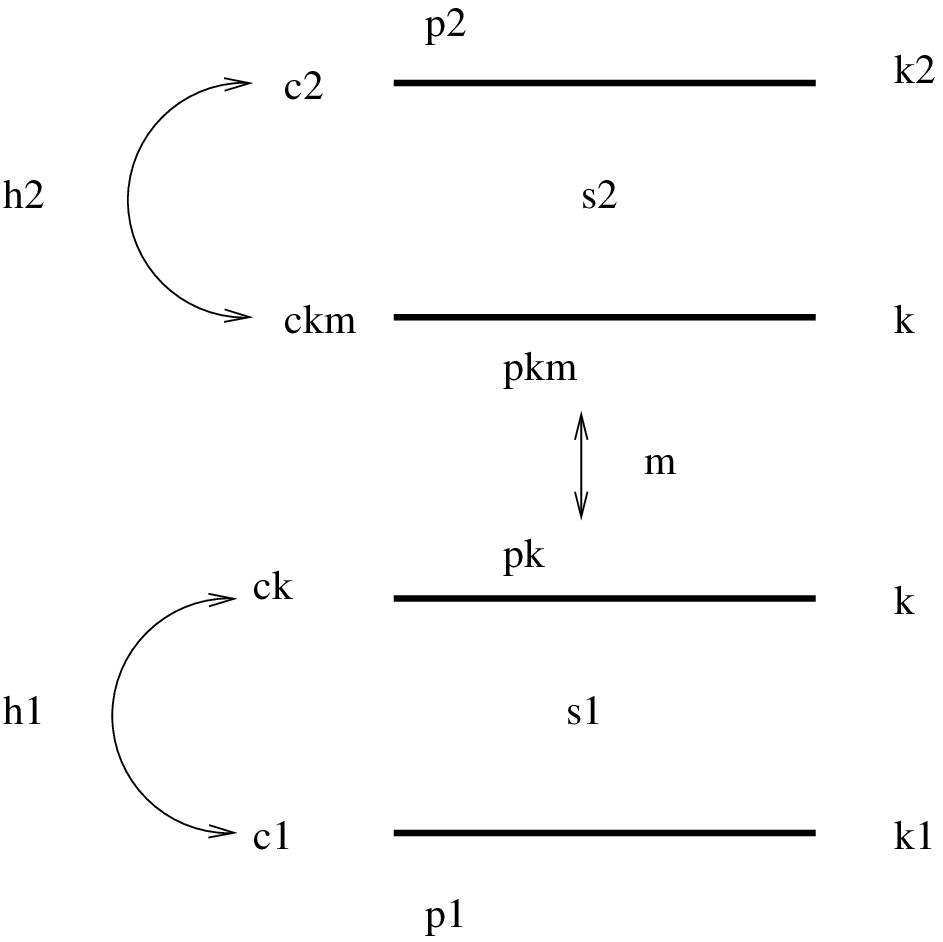}
                \caption{\small }
        \end{subfigure}
                \caption{(a) \small The {\it pre--constraint surface} $\cc^-_n$ and the {\it post--constraint surface} $\cc^+_n$ in the phase space $T^*\cq_n$, in general, do not coincide, i.e.\ $\cc^-_n\neq\cc^+_n$. In order to ensure the correct dynamics, we have to impose both the pre-- and post--constraints at step $n$ and thus must restrict to the intersection $\cc_n=\cc^-_n\cap\cc^+_n$. (If $\cc_n=\emptyset$, the dynamics is inconsistent.) (b) Matching symplectic structures at step $n$.}\label{conmat}
\end{figure}

For a larger sequence of evolution moves (for which in general both $\dim\cc^-_n\neq\dim\cc^+_n$, as well as $\dim\cc_n\neq\dim\cc_{n+1}$ are possible), this process is repeated until all constraint hypersurfaces are stable or an inconsistency has been reached. For translation invariant systems, this prescription is equivalent to the preservation of constraints given in \cite{DiBartolo:2004cg}. However, in contrast to the translationally invariant case, the secondary constraints are, in general, not the same at each time step and, moreover, crucially depend on the choice of initial $n'=i$ and final $n''=f$ time step. 
For instance, onsider an effective evolution move $i\rightarrow n$ described by $S_{in}$ and another effective evolution move $n\rightarrow f$ described by $S_{nf}$, where $S_{in},S_{nf}$ are effective actions. All constraints at $n$---including possibly propagated constraints, that would arise via the elementary evolution moves contained in both $i\rightarrow n$ and $n\rightarrow f$---will arise as the post--constraints of $S_{in}$ and as the pre--constraints of $S_{nf}$. That is, the effective actions automatically contain all propagated constraints.

Therefore, if we are using effective actions, consistency of the evolution simply requires to take the intersection of the post--constrained surface defined by $S_{in}$ with the pre--constraint hypersurface defined by $S_{nf}$ as the total constraint hypersurface at time $n$, $\cc_n:=\cc^-_n\cap\cc^+_n$, for {\it any} elementary or effective evolution $i\rightarrow n\rightarrow f$ (see figure \ref{conmat}). In addition, secondary pre--constraints at the initial step $i$ may arise as the pre--image under time evolution of $\cc_n$ and, likewise, secondary post--constraints at the final step $f$ may arise as the image under time evolution of $\cc_n$. If, however, we now were to evolve further to $i'<i$ and/or $f'>f$, additional propagated secondary constraints at any of the time steps may appear.

As regards the canonical data at $n$, the imposition of pre--constraints in addition to the post--constraints leads to conditions  which either
\begin{itemize} \parskip -.5mm
\item[(a)] are automatically satisfied (i.e.\ the pre--constraints are dependent on the post--constraints), 
\item[(b)] are {\it not} automatically satisfied, 
yet which do {\it not} fix the flows and {\it a priori} free parameters $\lambda_n$ of the post--constraints, 
\item[(c)] fix some 
{\it a priori} free parameters $\lambda_n$ (likewise for the {\it a posteriori} free parameters $\mu_n$), or
\item[(d)] cannot be simultaneously satisfied such that $\cc_n=\emptyset$ and the dynamics is inconsistent.
\end{itemize}
This enforces restrictions on the discrete time evolution and the amount of arbitrariness in the data. The consequences of cases (a)--(c) for the dynamics, symmetries and observables, as well as roles (i)--(iii) of the constraints shall be discussed in detail in the remainder of this section.

\subsection{Restricting the dynamics}\label{sec_correctdyn}

Let us begin by briefly discussing role (i) of the constraints: ensuring the `correct dynamics'. 
If we ensure that both the pre-- and post--constraints (including possibly propagated constraints) are satisfied at each step, we obtain the correct dynamics because by momentum matching all equations of motion will be implemented.

In non--translation invariant systems the numbers of pre-- or post--constraints at fixed $n$, in general, depends on the initial and final step between which one evolves: secondary constraints at $n$ may arise as a consequence of imposing equations of motion at neighbouring steps. This, in fact, can only happen if case (b) above occurs at neighbouring steps.  

\begin{Example}\label{examp1}
\emph{
{\bf(A)} Consider the massless scalar field evolving in two moves as depicted in figure \ref{fig_exABC} (a).}
\begin{center}
\begin{figure}[htbp!]
\psfrag{n}{\small$n=1$}
\psfrag{n1}{\small$n=2$}
\psfrag{n2}{\small$n=0$}
\psfrag{1}{\footnotesize 1}
\psfrag{2}{\footnotesize 2}
\psfrag{3}{\footnotesize 3}
\centering
\begin{subfigure}[b]{.25\textwidth}
\centering
\includegraphics[scale=.5]{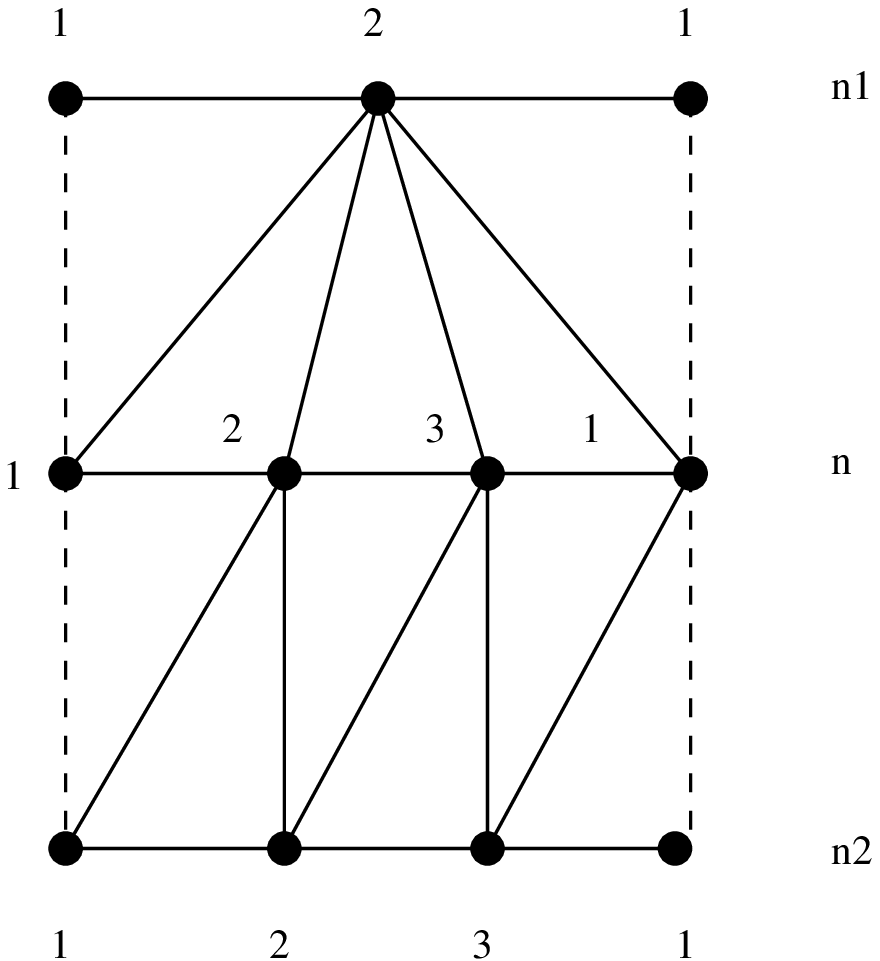}
\caption{\small }
\end{subfigure}
\hspace*{1.5cm}
\begin{subfigure}[b]{.25\textwidth}
\centering
\includegraphics[scale=.5]{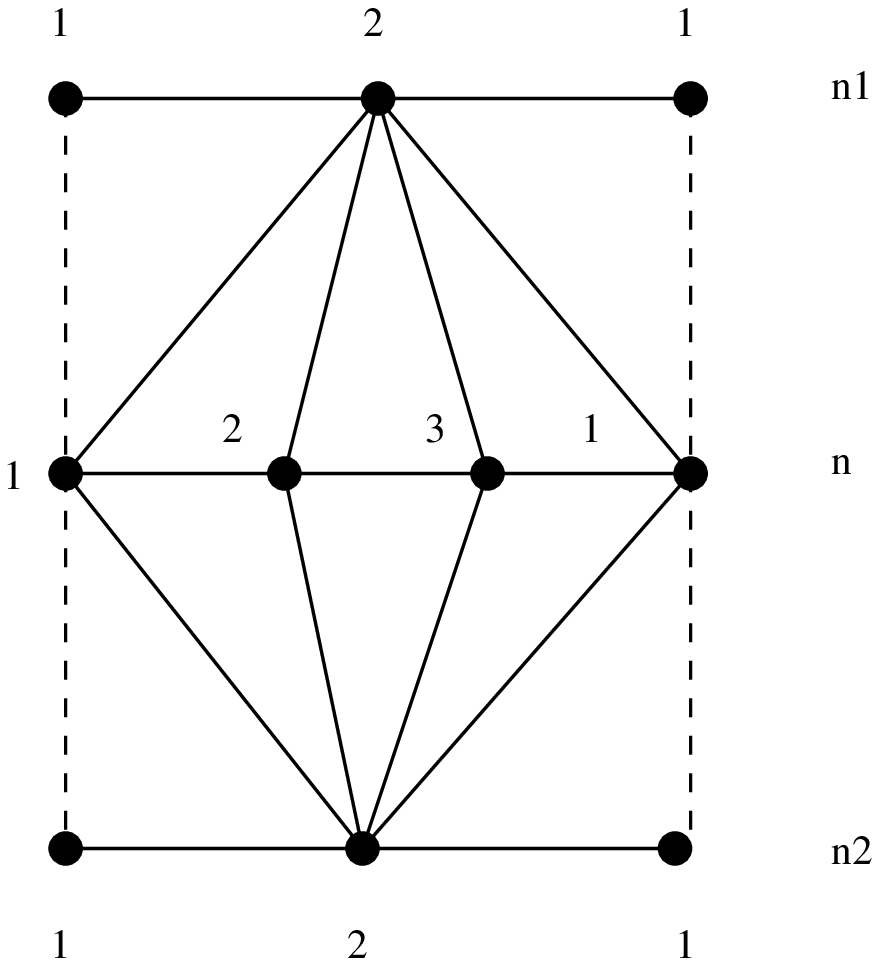}
\caption{\small }
\end{subfigure}
\hspace*{1.2cm}
\begin{subfigure}[b]{.25\textwidth}
\centering
\includegraphics[scale=.5]{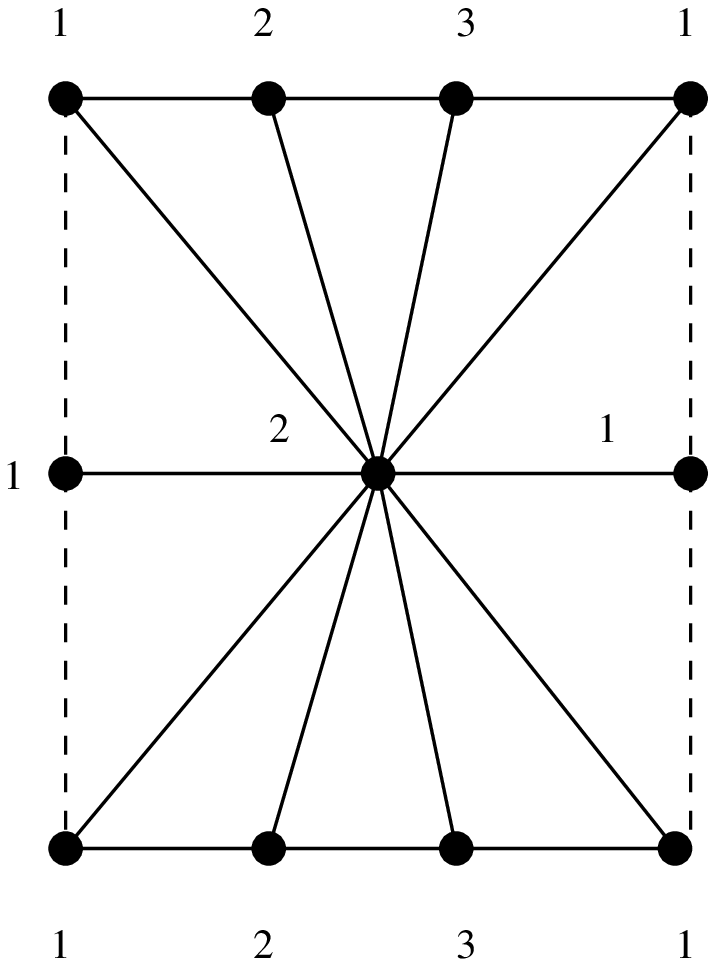}
\caption{\small }
\end{subfigure}
\caption{\small Schematic illustration of (a) example \ref{examp1} (A), (b) example \ref{examp1} (B), and (c) example \ref{examp1} (C). The numbers label the vertices and dashed lines indicate periodic identification.}\label{fig_exABC}
\vspace{1cm}
\end{figure}
\end{center}
\vspace*{-1.5cm}
\emph{There are three scalar field variables at time $n=0$, three at time $n=1$ and two scalar field variables at the  final time $n=2$.  The second move from $n=1$ to $n=2$ corresponds to the time reversed version we discussed in example \ref{examp0} in section \ref{sec_psext}. Thus, we have one pre--constraint at time $n=1$
\ba\label{blablabla}
{}^-\! C^1 &=&  \pi^{1}_2 -\pi^{1}_3 \,+\tfrac{5}{2} \left( \phi_{1}^2 -  \phi_{1}^3 \right)  ,
\ea
which leads to the free parameter $\mu_1=\phi^2_1-\phi^3_1$. 
For the evolution from time $n=0$ to $n=1$ we have to consider the adjacency matrix
\ba
A^1=\left(\begin{matrix}
1 & 1 & 0 \\
0 & 1 & 1 \\
1 & 0 & 1 \end{matrix}
\right)
\ea
which is invertible. Therefore, no pre-- or post--constraints arise from the evolution $0 \rightarrow 1$. The pre--constraint at time $n=1$ has to be propagated back to a (pre--) constraint at time $n=0$ because the data at $n=0$ has to be restricted such that it satisfies (\ref{blablabla}) upon propagation with $\ch_1$. This will also result in a free parameter $\mu_0$, whose value cannot be postdicted by the canonical data at time $n=2$. The so obtained constraint at $n=0$ will coincide with the pre--constraint of the effective time evolution from $n=0$ to $n=2$, as illustrated in figure \ref{fig_seq1}.
\begin{center}
\begin{figure}[htbp!]
         \psfrag{k}{\tiny $1$}
\psfrag{k1}{\tiny $0$}
\psfrag{k2}{\tiny $2$}
\psfrag{c2}{\tiny $\textcolor{red}{\cc^+_{2}}$}
\psfrag{c1}{\tiny $\textcolor{blue}{T^*\cq_{0}}$}
\psfrag{ck}{\tiny $\textcolor{red}{T^*\cq_1}$}
\psfrag{ckm}{\tiny $\textcolor{blue}{\cc^-_1}$}
\psfrag{p2}{\tiny $x_{2}=\textcolor{red}{\lambda_{2},{}^+p^{2}}$}
\psfrag{p1}{\tiny $x_{0},\textcolor{blue}{{}^-p^{0}}$}
\psfrag{pk}{\tiny $x_{1},\textcolor{red}{{}^+p^{1}}$}
\psfrag{pkm}{\tiny $x_{1}=\textcolor{blue}{\mu_{1},{}^-p^{1}}$}
\psfrag{m}{\footnotesize match}
\psfrag{h1}{\tiny $\ch_1$}
\psfrag{h2}{\tiny $\ch_{2}$}
\psfrag{s1}{\tiny $S_1$}
\psfrag{s2}{\tiny $S_{2}$}
\centering
\begin{subfigure}[b]{.25\textwidth}
\centering
{\includegraphics[scale=.5]{figure11b.eps}} 
\caption{\small The two individual evolution moves $0\rightarrow1$ and $1\rightarrow2$.}
\end{subfigure}
\hspace*{1cm}
\psfrag{ck}{\tiny $\cc_1=\textcolor{blue}{\cc^-_1}$}
\psfrag{i}{\tiny integrate}
\begin{subfigure}[b]{.25\textwidth}
\centering
{\includegraphics[scale=.5]{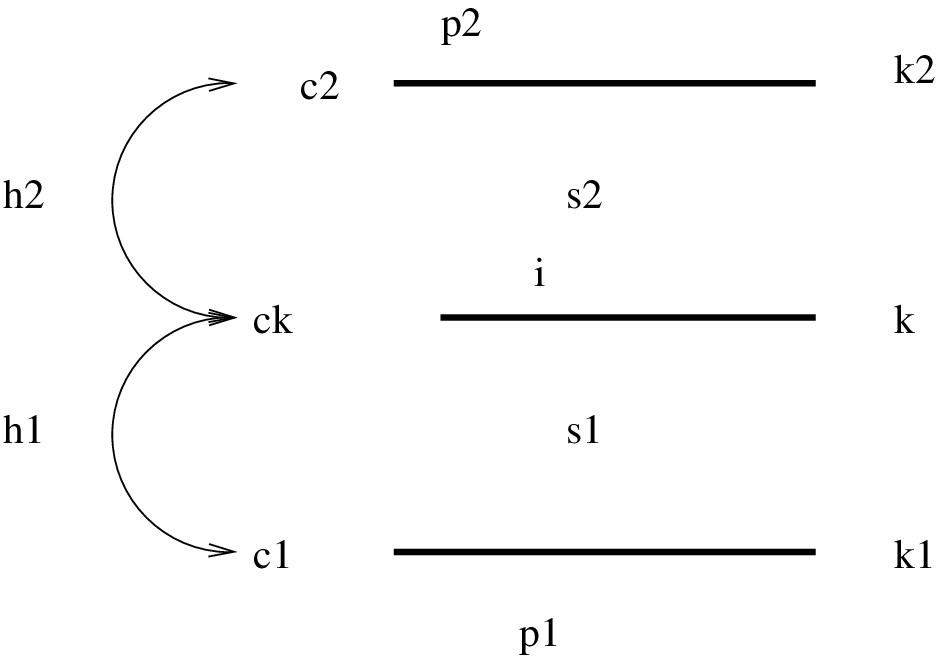}} 
\vspace*{.5cm}
\caption{\small Matching the symplectic structures at $n=1$.}
\end{subfigure}
\hspace*{1cm}
\psfrag{h}{\tiny$\tilde{\ch}_{02}$}
\psfrag{s}{\tiny $\tilde{S}_{02}$}
\psfrag{c4}{\tiny $\textcolor{red}{\tilde{\cc}^+_{2}}$}
\psfrag{c3}{\tiny $\textcolor{blue}{\tilde{\cc}^-_{0}}$}
\psfrag{p1}{\tiny $x_{0}=\textcolor{blue}{\tilde{\mu}_{0},{}^-p^{0}}$}
\begin{subfigure}[b]{.25\textwidth}
\centering
{\includegraphics[scale=.5]{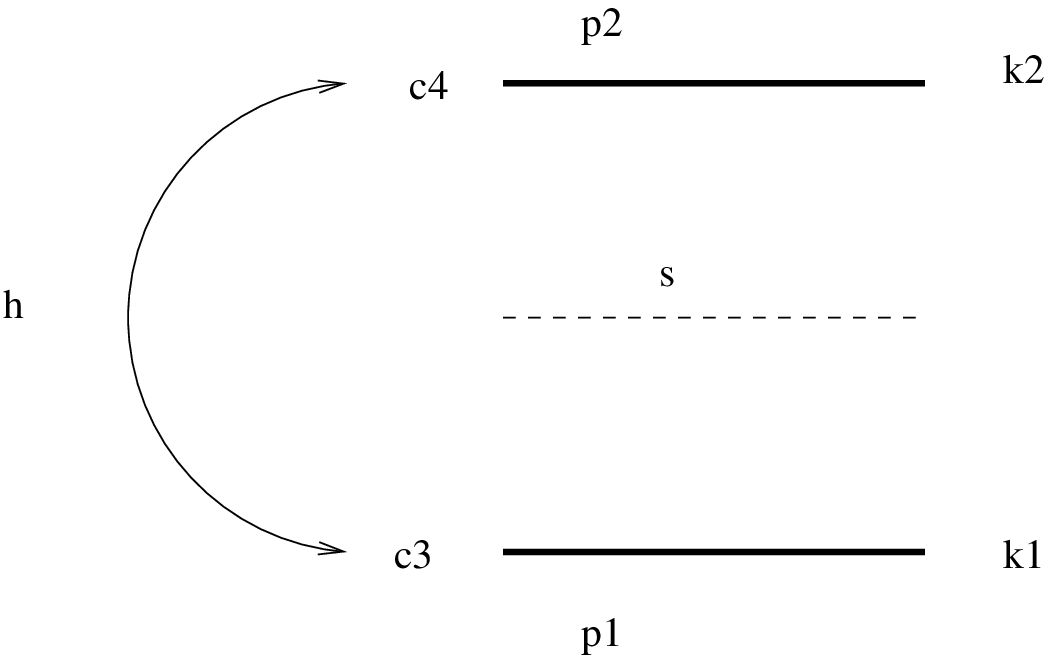}} 
\caption{\small An effective evolution move $0\rightarrow2$ results with a new pre--constraint surface $\tilde{C}^-_0$ at $n=0$.}
\end{subfigure}
\caption{\small Schematic illustration of example \ref{examp1} (A).}\label{fig_seq1}
\end{figure}
\end{center}
\vspace*{-1cm}
~\\
{\bf (B)} On the other hand, consider the example in figure \ref{fig_exABC} (b). Here we replace the time evolution move $0\rightarrow 1$ from the previous example by the time reversed time evolution move $1\rightarrow 2$. From the considerations of the single moves in example \ref{examp0} of section \ref{sec_psext}, one obtains both a pre--constraint and a post--constraint at time $n=1$
\ba
{}^-\! C^1 &=&  \pi^{1}_2 -\pi^{1}_3 \,+\tfrac{5}{2} \left( \phi_{1}^2 -  \phi_{1}^3 \right) , \nn\\
{}^+\! C^1 &=&  \pi^{1}_2 -\pi^{1}_3 \,-\tfrac{5}{2} \left( \phi_{1}^2 -  \phi_{1}^3 \right)  .
\ea
These constraints are second class with respect to each other and indeed fix the a priori and a posteriori parameter $\lambda_1=\mu_1=\phi^2_1-\phi^3_1=0$. The constraint hypersurface can equivalently be described by the two constraints $C_1^1=\phi^2_1-\phi^3_1$ and $C^1_2= \pi^{1}_2 -\pi^{1}_3$. Propagating the constraints forward to $n=2$ and backward to $n=0$ does not lead to any constraints at these time steps. The effective time evolution $0\rightarrow 2$ is thus regular.\\
~\\
{\bf (C)} Having the same number of variables at final and initial time steps does not necessarily guarantee a regular time evolution. If we exchange the two time evolution moves of the previous example (see figure \ref{fig_exABC} (c)), we will obtain a pre--constraint at time $n=0$ and a post--constraint at time $n=2$
\ba
{}^-\! C^0 &=&  \pi^{0}_2 -\pi^{0}_3 \,+\tfrac{5}{2} \left( \phi_{0}^2 -  \phi_{0}^3 \right)  , \nn\\
{}^+\! C^3 &=&  \pi^{3}_2 -\pi^{3}_3 \,-\tfrac{5}{2} \left( \phi_{3}^2 -  \phi_{3}^3 \right)  
\ea
with the corresponding free parameters $\mu_0$ and $\lambda_3$.}
\end{Example}

One might be surprised about the feature that constraints propagate backwards in time evolution. Therefore, let us provide an intuitive picture: we can imagine our evolution on the cylinder also as a radial evolution on an annulus region in two--dimensional space.  If we proceed from time slices with more vertices to time slices with less vertices, we can interpret this as evolving from the outer boundary to the inner boundary of the annulus. Proceeding further and further with the evolution, more and more equations of motions in the ball defined by the outer boundary of  the annulus have to be satisfied. Thus, we have to expect more and more constraints fixing the momenta as function of the boundary scalar fields. Closing the inner boundary we expect a totally constrained system at the outer boundary because the values of the scalar field on the spherical boundary should be sufficient to specify a solution uniquely. On the other hand, all the fields at the outer boundary correspond to free parameters $\mu$, since the fields at the outer boundary cannot be determined from the zero--dimensional phase space which arises after one has closed the annulus to a sphere.

These examples highlight how constraints can severely restrict the space of solutions and how additional constraints {\it at fixed} $n$ can arise. 

Now consider a general evolution up to step $n$. Assume a new evolution move $n\rightarrow n+1$ is to be performed and one finds that a pre--constraint arising in this move is in conflict with the underlying canonical data at $n$. Then there are four options:
\begin{enumerate}\parskip -1mm
\item Accept that we cannot perform this evolution move. In this case, perform some other evolution move.\footnote{For instance, a type I local move (in particular, in Regge Calculus a 1--$D$ Pachner move, where $D>2$ is the dimension of the spacetime triangulation \cite{Dittrich:2011ke}) is always possible because no pre--constraints arise.}
\item \begin{itemize} \parskip -.5mm
\item[(a)] Change the underlying data by varying parameters which are {\it a priori} free up to step $n$ such that the attempted move is possible, otherwise
\item[(b)] restrict the space of initial data leading to $n$ such that the attempted move becomes possible.
\end{itemize}
\item Neither of 1 or 2 is possible and the evolution becomes inconsistent and stops.
\end{enumerate}
Option 1 is what one would choose if one solved an initial value problem, while options 2 (a) and (b) are what one would choose in case one attempted to solve some boundary value problem.  

It is worthwhile to add a few remarks concerning the consistency of the discrete evolution from an initial value problem. Firstly, it should be noted that there is no principle, which dictates either the choice of the set of evolution moves or the particular sequence of these moves in the evolution. Instead, the choice of both the set of moves and their sequence has to be put in by hand. It is the constraints, which subsequently determine whether a given evolution move or a given sequence of moves is allowed or not: if the constraints can be satisfied, the sequence will generate a solution to the discrete equations of motion. In general, there will exist a whole plethora of consistent choices of moves and sequences. Broadly, one may choose between:
\begin{itemize}
\item[(1)] Choose a completely general (and elementary) set of evolution moves and leave their sequences open such that one may generate all possible solutions arising from a given initial data set, or
\item[(2)] Fix the set of evolution moves and their sequences at the outset and thereby restrict the space of solutions arising from a given initial data set. 
\end{itemize}

In a discrete gravitational context, the Pachner moves are an ergodic set of evolution moves, which can generate all possible spacetime triangulations (of a given topology) and are therefore an example of choice (1). For an implementation of these moves in canonical language, see \cite{Dittrich:2011ke}. On the other hand, an example of a restricted set of moves and sequences (2) are the tent moves in discrete gravity, which have been amply discussed in \cite{Bahr:2009ku, Dittrich:2009fb}.

Secondly, the question arises whether different consistent choices of evolution moves and their sequences generate equivalent or distinct solutions from a given initial data set. The answer to this question depends on the presence or absence of symmetries. For instance, in a discrete gravitational context, in particular, in Regge Calculus, an initial data hypersurface that leads to flat solutions will permit infinitely many different choices and sequences of evolution moves, all of which generate equivalent solutions to the Regge equations, namely flat triangulations. This is a consequence of the continuum diffeomorphism symmetry being preserved for flat triangulations \cite{Dittrich:2008pw, Bahr:2009ku}. That is, in this case, different solutions arising from different sequences and sets of moves, but from the same initial data set, can be mapped into each other through the symmetries. On the other hand, for curved Regge solutions, the diffeomorphism symmetry becomes broken \cite{Dittrich:2008pw, Bahr:2009ku}, such that different choices and sequences of evolution moves applied to the same initial data set will lead to distinct solutions in the sense that they can no longer be mapped into each other by a symmetry transformation. Consequently, the discrete dynamics will, in the general case of broken continuum symmetries, be non-hyperbolic: a fixed initial data set cannot uniquely predict the future solution and a multitude of distinct solutions are compatible with it. For further discussion of this, see also \cite{Dittrich:2011ke}.

This situation bears some loose analogy to the continuum, at least in a gravitational context. In a sense, the different choices and sequences of evolution moves in discrete gravity can be viewed as different lapse and shift choices because they determine how the `spatial' triangulated hypersurface is pushed forward in `time'. Not being able to perform a given move because of a conflict with the constraints, corresponds in the continuum simply to the evolution not running off the constraint surface. If the continuum symmetry is present, as for flat solutions, then the different choices of moves and their sequences are essentially gauge choices, in analogy to the continuum lapse and shift. By contrast, when the symmetries are broken, these different choices are no longer equivalent and determine distinct solutions---a situation which does not arise in the continuum. It should be noted, however, that the notion of lapse and shift in Regge Calculus is not only associated to the choice and sequence of moves, but also to the lengths of edges interpolating between different hypersurfaces in the evolution. For more details on this in the context of tent moves see \cite{Bahr:2009ku, Dittrich:2009fb}, or in the general context of Pachner moves, see \cite{Dittrich:2011ke}.

\subsection{Constraints and symmetries}\label{sec_symrole}

Next, let us elaborate on role (ii) of the constraints in the discrete, namely, generating gauge symmetries (if present). 
To begin with, consider the Hessian of the action, which is the matrix of second derivatives of the action with respect to `bulk' variables. At time step $n$ it reads \ba\label{hess}
H^n_{ij}=\frac{\p^2 S_{n}}{\p x_n^j \p x_n^i}+\frac{\p^2 S_{n+1}}{\p x_n^j \p x_n^i}\,.\label{Hessian}
\ea
Null vectors of the Hessian define flat directions at the extrema of the action and thereby define genuine gauge directions. The Hessian thus plays a central role in the determination 
of gauge symmetry.  Here we consider the Hessian with boundary data at $(n-1)$ and $(n+1)$ fixed. One should be aware that under further evolution, i.e.\ treating these boundary data as dynamical, the Hessian becomes a larger matrix (including blocks corresponding to the `bare' Lagrangian two--forms at $n$ and $(n+1)$, etc.), and the previous null vectors may cease to be null vectors of the effective Hessian. 

For non--linear theories, null vectors may only exist for the Hessian evaluated on solutions and the space of solutions may change if further evolution is taken into account (as a result of propagated constraints). Even more so, null vectors may exist only for special solutions. This is actually the case in 4D Regge Calculus, where null vectors exist for flat but not for curved solutions  \cite{Bahr:2009ku}. In the specification of irregular systems (\ref{app13}), however, we assumed a constant rank of the Lagrangian two--form in an open neighbourhood in the space of initial data. We shall, therefore, continue to make the assumption that the rank of the Hessians and Lagrangian two--forms is locally constant. Nevertheless, all considerations remain valid, even if a situation as in 4D Regge Calculus occurs, if one linearizes the theory around a solution (with additional symmetries).

Recall that an {\it a priori} free parameter $\lambda_n$ and an {\it a posteriori} free parameter $\mu_n$ were defined as functions of the configuration data at time step $n$, whose values cannot be pre-- or postdicted from the canonical data at $(n-1)$ or $(n+1)$, respectively. A gauge mode\footnote{generally only defined in the linearized theory} at time $n$ is given by a free parameter $\lambda_n=\mu_n$ which is both {\it a priori} and {\it a posteriori} free and that can never be pre-- or postdicted by any other data. 
 
Thus, given a solution $s$ satisfying the equations of motion
 \ba
\frac{\partial (S_{n-1} +S_n)}{\partial x_{n-1}}=0  ,\q\q \frac{\partial (S_n+ S_{n+1})}{\partial x_n}=0 ,\q\q   \frac{\partial (S_{n+1}+S_{n+2})}{\partial x_{n+1}}=0 ,
 \ea
  the free parameter describes an infinitesimally displaced solution $(s)_n+\epsilon V_n, V_n \in T{\cal Q}_n$ at time $n$, which also satisfies the equations of motion. Taking the derivatives of the equations of motion in the direction of $V_n$, we see that  to each gauge mode there corresponds a vector $V_n$ at step $n$, which is a (i) null vector  of the Hessian at $n$, (ii) a  right null vector of the Lagrangian two--form $\Omega_n$,  and (iii) a left null vector of the Lagrangian two--form $\Omega_{n+1}$.




%

Next, let us study the relation between the presence of gauge symmetry and the first and second class nature of the pre-- and post--constraints. To this end, we label the $N^n_R$ right null vectors $(R_n)^i_r$ of $\Omega_n$ and their corresponding post--constraints by $r=1,\ldots,N^n_R$ and the $N^{n}_L$ left null vectors $(L_n)^i_l$ of $\Omega_{n+1}$ and their corresponding pre--constraints by $l=1,\ldots,N^{n}_L$.

\begin{Theorem}\label{thm_conalgeneral}
The set of pre--constraints at fixed $n$ and the set of post--constraints at fixed $n$ each form a first class Poisson sub--algebra
\ba
\{{}^-C^n_l,{}^-C^n_{l'}\}  \underset{{\cal C}^-_n}{=}0\,,\q\q\q\{{}^+C_r^n,{}^+C_{r'}^n\}   \underset{{\cal C}^+_n}{=}0\,,\q\q\forall\, l,l',r,r',
\ea
where the first equations holds on the pre--constraint surface ${\cal C}^-_n$ and the second on the post--constraint surface ${\cal C}^+_n$.

Furthermore, 
\ba
\{{}^-C^n_l,{}^+C_r^n\} \simeq \gamma_l{}^{l'}(L_n)^i_{l'}\,H^n_{ij}\,\rho_r{}^{r'}(R_n)^j_{r'}\,,\q\q\forall\,l,r\,,\nn
\ea
where $\gamma_l{}^{l'}(x_n,p^n)$ and $\rho_r{}^{r'}(x_n,p^n)$ are coefficient functions determined by the gradients of $\cc_n^-$ and $\cc_n^+$, respectively and $\simeq$ denotes a weak equation valid on the constraint hypersurface ${\cal C}^-_n \cap {\cal C}^+_n  $.
\end{Theorem}

\begin{proof}
The $N^{n}_L$ pre--constraints ${}^-{C}_l^n$ and the $N^n_R$ post--constraints ${}^+C_r^{n}$ are defined through the equations
\ba\label{app16a}
0&=&{}^-{C}_l^n(x_n,{}^-p^n)\Big|_{  {}^-p^n=  -\frac{\p S_{n+1}}{\p x_n}(x_n,x_{n+1})   }\,,\nn\\
0&=&{}^+{C}_r^{n}(x_{n},{}^+p^{n})\Big|_{  {}^+p^{n}=  \frac{\p S_{n}}{\p x_{n}}(x_{n-1},x_{n})   }\,,
\ea
for arbitrary $x_{n-1},x_n$ and $x_{n+1}$. Differentiating the first equations with respect to $x_{n+1}$ and $x_n$, respectively, and the second one with respect to $x_{n-1}$ and $x_n$, respectively, yields
\ba\label{friday1}
0& =& \frac{\p {}^-C_l^n}{\p {}^-p_j^n} \frac{\p^{2} S_{n+1}}{\p x_n^j \p x_{n+1}^i}\,,\q\q\q 0 = \frac{\p {}^-C_l^n}{\p x^j_n} \,-\,  \frac{\p {}^-C_l^n}{\p {}^-p_i^n} \frac{\p^{2} S_{n+1}}{\p x_n^j \p x_n^i}\,,\label{app16b}\\
0&=& \frac{\p {}^+C_r^{n}}{\p {}^+p_j^{n}} \frac{\p^{2} S_{n}}{\p x_n^i \p x_{n-1}^j}\,,\q\q\q 0 = \frac{\p {}^+C_r^{n}}{\p x^j_{n}} \,+\,  \frac{\p {}^+C_r^{n}}{\p {}^+p_i^{n}} \frac{\p^{2} S_{n}}{\p x_{n}^j \p x_{n}^i}\,\label{app16bb} .
\ea

Note that the right hand equations also imply non--vanishing gradients $\frac{\p {}^-C_l^n}{\p {}^-p_j^n}$ and $\frac{\p {}^+C_r^{n}}{\p {}^+p_j^{n}}$, because otherwise the constraints would be constants. These are equations between functions on ${\cal Q}_{n} \times {\cal Q}_{n+1}$  and ${\cal Q}_{n-1} \times {\cal Q}_{n}$, respectively.  We used the pre-- and post--Legendre transformations in  (\ref{app16a})  to map these configuration spaces to the phase space ${\cal P}_n$. The following equations, therefore, hold on the image of these maps, i.e.\ on the pre-- and post--constraints surfaces. 

The equations on the left hand side of (\ref{app16b}, \ref{app16bb}) specify null vectors of the Lagrangian two--forms, for which we earlier assumed a basis. Hence, 
\ba\label{app16c}
 \frac{\p {}^-C_l^n}{\p {}^-p_j^n} \, \underset{{\cal C}^-_n}{=}\, \gamma_l{}^{l'} (x_n,p^n)\,  (L_n)_{l'}^j,\q\q\q \frac{\p {}^+C^{n}_r}{\p {}^+p_j^{n}} \,\underset{{\cal C}^+_n}{=}\, \rho_r{}^{r'} (x_{n},p^{n})\,  (R_{n})^j_{r'},\ea
where $\gamma_l{}^{l'}$, $\rho_r{}^{r'}$ are appropriately chosen matrix coefficient functions.  This associates a (non--vanishing) left null vector of $\Omega_{n+1}$ to every pre--constraint and a (non--vanishing) right null vector of $\Omega_n$ to every post constraint.

Consider the pre--constraints at $n$. Using (\ref{app16b}, \ref{app16c}), we directly compute
\ba
\{{}^-C^n_l,{}^-C^n_{l'}\}&=&\f{\p {}^-C^n_l}{\p x^i_n}\f{\p {}^-C^n_{l'}}{\p p^n_i}-\f{\p {}^-C^n_l}{\p p^n_i}\f{\p {}^-C^n_{l'}}{\p x^i_n}\nn\\
&\underset{{\cal C}^-_n}{=} &\gamma_l{}^{l_1}(L_n)^j_{l_1}\frac{\p^2 S_{n+1}}{\p x_n^j \p x_n^i}\gamma_{l'}{}^{l_2}(L_n)^i_{l_2}-\gamma_{l'}{}^{l_2}(L_n)^i_{l_2}\frac{\p^2 S_{n+1}}{\p x_n^i \p x_n^j}\gamma_l{}^{l_1}(L_n)^j_{l_1}\nn\\
&\underset{{\cal C}^-_n}{=} &0\,.\nn
\ea
Likewise, using (\ref{app16bb}, \ref{app16c}), one finds that the Poisson brackets between post--constraints at $n$ vanish on the post--constraint hypersurface $\cc^+_n$. 

Analogously, by (\ref{app16b}--\ref{app16c}), 
\ba
\{{}^-C^n_l,{}^+C^n_r\}&=&\f{\p {}^-C^n_l}{\p x^i_n}\f{\p {}^+C^n_r}{\p p^n_i}-\f{\p {}^-C^n_l}{\p p^n_i}\f{\p {}^+C^n_r}{\p x^i_n}\nn\\
&\underset{{\cal C}^-_n  \cap {\cal C}^+_n}{=} &\gamma_l{}^{l'}(L_n)^j_{l'}\left(\frac{\p^2 S_{n+1}}{\p x_n^j \p x_n^i}+ \frac{\p^2 S_{n}}{\p x_n^j \p x_n^i}\right) \rho_r{}^{r'}(R_n)^j_{r'}\,,\label{LHR}
\ea
where the term in brackets defines the Hessian of the action. This proves the statement.
\end{proof}

In consequence,
\begin{itemize}\parskip -2mm
\item[(i)] a pre--constraint ${}^-C^n_l$ is first class (in particular, its Poisson brackets with all post--constraints  vanish on the constraint hypersurface) if  $ \frac{\p {}^-C_l^n}{\p {}^-p_i^n} H^n_{ij}\simeq 0$, 
\item[(ii)] a post--constraint ${}^+C^n_r$ Poisson is first class  if $H^n_{ij}\frac{\p {}^+C^{n}_r}{\p {}^+p_j^{n+1}}  \simeq 0$. 
\end{itemize}


It is also possible that a pre--constraint is first class despite the Hessian not having null vectors. This happens if the corresponding left null vector $\gamma_l{}^{l'}(L_n)^i_{l'}$ is orthogonal to all right null vectors $(R_n)^i_{r'}$ with respect to the metric defined by the Hessian.


There is an immediate consequence of theorem \ref{thm_conalgeneral} (applying to case (a) of section \ref{sec_conmatch}):
\begin{Corollary}\label{cor_ceqc}
Let ${}^-C^n$ be a pre--constraint and ${}^+C^n$ be a post--constraint. If this pre-- and post--constraint coincide, i.e.\ $C^n:={}^-C^n={}^+C^n$, then $C^n$ is necessarily first class.
\end{Corollary}

The following theorem implies that such coinciding pre-- and post--constraints generate gauge symmetries of the  action 
(or Hamilton's principal function). 

\begin{Theorem}\label{thm_sym}
To every constraint $C^n$ which is both a (primary) pre-- and post--constraint at step $n$ there is associated 
\begin{itemize} \parskip -1mm
\item[(i)] a null vector of $H^n_{ij}$, 
\item[(ii)] a right null vector of $\Omega_n$, and 
\item[(iii)] a left null vector of $\Omega_{n+1}$. 
\end{itemize}
Furthermore, $C^n$ generates a flow tangential to $\cc_n={\cal C}^-_n  \cap {\cal C}^+_n$ 
which is a symmetry of the (effective) action.
\end{Theorem}

\begin{proof}
The equations (\ref{app16c}) applied to the case ${}^-C^n={}^+C^n$ give the associated  right and left null vector to $\Omega_n$ and $\Omega_{n+1}$ respectively. These equations also show that the right and left null vector coincide with each other and is given by $\f{\p C^n}{\p p^n_i}$. 
Subtracting the right equation in (\ref{app16b}) from the right equation in (\ref{app16bb}), both applied to this case, yields
\ba
\f{\p C^n}{\p p^n_i}H^n_{ij}=0\,\label{nullvechesse}
\ea
which shows that this vector also defines a null vector of the Hessian (evaluated on the space of solutions).


Finally, $\{x^i_n,C^n\}=\f{\p C^n}{\p p^n_i}$,  thus the flow is defined by a null vector of the Hessian and the Lagragian two--forms and hence generates transformations in flat directions at the extrema of the action.  Given that $C^n$ is first class with all other constraints, it follows that the flow  is tangential to the constraint surface $\cc_n=\cc^+_n\cap\cc^-_n$.
\end{proof}


\begin{Theorem}\label{thm_gaugemode}
To every constraint $C^n:={}^-C^n={}^+C^n$ which is both a primary pre-- and post--constraint at step $n$ there is associated a gauge mode.
\end{Theorem}

\begin{proof}
From what has been said before, evolution from $(n-1)$ to $n$ and from $n$ to $(n+1)$ leads to an  {\it a priori} and {\it a posteriori} free parameter $\lambda_n=\mu_n$, which is not fixed by any of the other pre-- or post--constraints at $n$. We have to ensure, that this does not change if we evolve either further into the future or into the past, since pre-- or post--constraints might propagate to $n$. However, all these evolutions could be described by effective actions, all of which, by the conjunction of theorems \ref{thm_presymp} and \ref{thm_momup}, would include the same pre-- and post--constraints (including the propagated ones). Thus, $\lambda_n=\mu_n$ remains a free parameter throughout evolution because its corresponding constraint always remains first class.\end{proof}

In conclusion, although the constraints in the discrete do {\it not} generate the dynamics, cases (a)--(c) of section \ref{sec_conmatch} lead to some similarities (and dissimilarities) with the continuum situation:
\begin{itemize} \parskip -1mm
\item[(a)] Coinciding pre-- and post--constraints are first class and generate symmetries of Hamilton's principal function.
\item[(b)] Pre--constraints which are independent of the post--constraints and do not fix any $\lambda_n$ (and vice versa with the $\mu_n$) are first class, but do {\it not} generate any symmetries. 
\item[(c)] Pre--constraints which are independent of the post--constraints at $n$, yet which fix {\it a priori} free $\lambda_n$ must be second class together with those post--constraints whose flows they fix. 
\end{itemize}

That is, a symmetry generating constraint is necessarily first class also in the discrete (case (a)). However, in contrast to the continuum, a first class constraint does {\it not} necessarily generate symmetries in the discrete (case (b)). We will better understand this case (b) in the context of propagating degrees of freedom in section \ref{sec_varobs} below. 
Finally, recall that in the continuum Lagrange multipliers can only get fixed if they are associated to second class constraints. Likewise, {\it a priori} free parameters $\lambda_n$ in the discrete can also only get fixed if they are associated to post--constraints which are rendered second class by pre--constraints (case (c)).

\subsection{Varying numbers of propagating degrees of freedom}\label{sec_varobs}

We shall now discuss role (iii) of the constraints, namely the classification of degrees of freedom into gauge modes and gauge invariant observables.

In the continuum the number of physical degrees of freedom is tied to classifying constraints into first and second class. The (standard) discussion of the continuum corresponds to translation invariant systems, where, in particular, the number of physical observables is constant during time evolution. In contrast, in our general set--up, the constraint surface at a given time $n$ depends on the initial and final time step, thus we also have to expect that the notion of the number of propagating degrees of freedom depends on this choice.

\subsubsection{Observables as propagating degrees of freedom}


The principal idea is to define observables by {\it propagation of data}. In the continuum (or translation invariant systems)  
the classification of the constraints at a {\it single} instant of time into first and second class yields a complete characterization of propagating observable and gauge degrees of freedom \cite{Henneaux:1992ig}.  
On the other hand, in order to specify the meaning of propagating degrees of freedom in the discrete, we need {\it two} time steps; for the notion of propagation at the canonical level requires the global Hamiltonian time evolution map $\ch_{n-1}:\cc^-_{n-1}\rightarrow\cc^+_n$. 

As discussed in section \ref{sec_singcan}, in the presence of constraints this map is a priori not well defined, rather, we need to specify {\it a priori} free parameters $\lambda_n$. Similarly, for the  inverse evolution $\cc^+_{n}\rightarrow\cc^-_{n-1}$ {\it a posteriori} free parameters $\mu_{n-1}$ must be specified. Varying these free parameters leads to orbits in the post-- and pre-- constraint hypersurfaces, respectively.
The Hamiltonian time evolution map $\ch_{n-1}$ is well defined and invertible as a map from the space of orbits at $(n-1)$  to the space of orbits at $n$. In particular, these spaces are of equal dimension (see also theorem \ref{thm_presymp}). 

Consequently, these two spaces of orbits can be identified with each other via the Hamiltonian evolution map and specify the canonical data which can be either uniquely pre-- or postdicted. These spaces, therefore, describe the propagating degrees of freedom of the evolution move $(n-1)\rightarrow n$ between the time steps $(n-1)$ and $n$. The number of propagating degrees of freedom will be defined as the dimension\footnote{This is twice the number of propagating configuration degrees of freedom, which might be used as the definition of propagating degrees of freedom in the Lagrangian picture.} of the space of orbits. 

Accordingly, we define {\it pre--observables}  $O^-_{n-1}$ at step $(n-1)$ and {\it post--observables} $O^+_n$ at step $n$ as functions which are well defined on the space of orbits, i.e.\ as $\mu_{n-1},\lambda_n$--independent functions; via $\ch_n$ the {\it pre--observables} at $(n-1)$ uniquely predict the {\it post--observables} at $n$ and, vice versa, the {\it post--observables} at $n$ uniquely postdict the {\it pre--observables} at $(n-1)$.

\begin{Theorem}\label{thm_Ok}
The \emph{pre--observables} $O^-_{n-1}$ Poisson commute (weakly) with the pre--constraints at $(n-1)$ and the \emph{post--observables} $O^+_n$ Poisson commute (weakly) with the post--constraints at $n$, i.e.
\ba
\{O^-_{n-1},{}^-C^{n-1}_l\} \,\underset{{\cal C}^-_{n-1}}{=} \, 0\,,\q\q\q\text{and}\q\q\q\{O^+_n,{}^+C^n_r\} \,\underset{{\cal C}^+_n}{=}  \, 0\,,\q\q\q\forall\,l,r\,.\nn
\ea
\end{Theorem}

\begin{proof}
The post--constraint ${}^+C^n_r$ generates a flow
\ba\label{fridayl1}
\{x_n^j, {}^+C^n_r \} &=& \frac{\partial  {}^+\!C^n_r }{\partial p^n_j} ,\nn\\
 \{p^n_j, {}^+C^n_r \} &=&-\frac{\partial  {}^+\!C^n_r }{\partial x_n^j} \,\,\underset{{\cal C}^+_n}{=}  \,\,\frac{\partial^2 S_n}{\partial x^j_n \partial x^i_n} \frac{\partial  {}^+\!C^n_r }{\partial p^n_i}\, \,=\,\, \frac{\partial {}^+\!p^n_j}{\partial x^i_n}   \frac{\partial  {}^+\!C^n_r }{\partial p^n_i}   .
\ea
In (\ref{app16c}) we showed that the gradient  $\frac{\partial  {}^+\!C^n_r }{\partial p^n_j}=:(R_n)^j_r$ is a right null vector of the Lagrangian two--form. Such a null vector leads to non--uniqueness of $x_n$ in the image of the time evolution map, see  section \ref{sec_singcan}. This is infinitesimally described  by $x_n+\epsilon\lambda^r_n (R_n)_r$ where $\lambda^r_n$ are the a priori free parameters.  Thus, the first equation  in (\ref{fridayl1}) shows that the constraint generates the corresponding flow for the configuration variables. The second equation shows that the post--constraint flow also reproduces the induced non--uniqueness for the post--momenta. Hence, the orbits discussed above are defined by the flow of the post--constraints. This flow is integrable, as the set of post--constraints is first class.

A phase space function which commutes (weakly) with the post--constraints is constant along a given orbit and will not be affected by the non--uniqueness of the time evolution map. A similar argument holds for the pre--observables. \end{proof}

%

The number of propagating degrees of freedom depends on initial and final step and the notion of observables is move dependent. This is illustrated by the following examples.

\begin{Example}
\emph{Let us reconsider examples \ref{examp1} (A) -- (C) of section \ref{sec_correctdyn}.\\
{\bf (A)} The evolution move $0\rightarrow 1$, in which three field variables and three momenta evolve into three field variables and three momenta is regular. No constraints arise and the number of propagating degrees of freedom is six. The following move $1\rightarrow 2$ evolves three phase space pairs into two pairs. There is a pre--constraint at time $1$ such that the number of propagating degrees  of freedom for this move is four. For the effective evolution $0\rightarrow 2$ a pre--constraint at $0$ arises and the number of propagating degrees of freedom coincides with the minimal number of the two moves, which is four. \\
{\bf (B)} In this case two phase space pairs at time $0$ evolve into three pairs at time $1$ which subsequently evolve back into two pairs at time $2$. One finds a post-- and pre--constraint at time $1$ and the number of propagating degrees of freedom for both moves is four. This also agrees with the number of propagating degrees of freedom for the effective evolution move $0\rightarrow 2$.\\
{\bf (C)} Three pairs at $0$ are evolving into two pairs at time $1$ and then back into three pairs at time $2$. We have a pre--constraint at $0$ and a post--constraint at $2$ such that again the number of propagating degrees of freedom for both time steps is four which coincides with the number of propagating degrees of freedom  in the effective evolution.}
\end{Example}

The number $N_{i\rightarrow f}$ of independent propagating degrees of freedom between an initial step $i$ and some final step $f$, is given by the dimension of the post--constraint hypersurface ${}^+{\cal C}_f$ modulo the flow generated by the post--constraints. This agrees with the dimension of the pre--constraint hypersurface ${}^-\!{\cal C}_i$ modulo the flow generated by the pre--constraints. Hence,
\ba
N_{i\rightarrow f}
&=&2Q-2\#(\text{pre--constraints at } i)
\;=\;2Q-2\#(\text{post--constraints at } f)\,\label{propobsno}
\ea
is the number of both the {\it pre--observables} at $i$ and {\it post--observables} at $f$. This number coincides with the rank of the symplectic form restricted to the pre-- and post--constraint surfaces discussed in theorem \ref{thm_presymp}. 
However, $N_{i\rightarrow f}$ is in general {\it not} the dimension of the reduced phase spaces at steps $i$ or $f$ because possible post--constraints at $i$ and possible pre--constraints at $f$ have thus far been ignored.

\subsubsection{The reduced phase space}\label{sec_redps1}

Let us define the reduced phase space at a given step $n$ with $i<n<f$ as the total constraint hypersurface $\cc_n:=\cc^+_n\cap\cc^-_n$ modulo any first class flows. As we shall see shortly, this reduced phase space will describe the degrees of freedom propagating from $i$ {\it through} step $n$ to $f$.  

Let us check how many of the $N_{i\rightarrow n}$ propagating degrees of freedom of the move $i\rightarrow n$ continue to propagate in the move $n\rightarrow f$. We have to impose both post-- and pre--constraints at $n$. There are three different cases that can arise, which we labelled by  (a)--(c) in section \ref{sec_conmatch}:
\begin{itemize} \parskip-1mm
\item[(a)] Part of the (possibly re--organized) post--constraints coincide with part of the (possibly re--organized) pre--constraints. According to theorem \ref{thm_sym}, this part defines constraints that are first class with all other constraints at $n$ and which, furthermore, generate gauge symmetries at $n$. Hence, for every such constraint the number of physical degrees of freedom decreases by two. However, this is automatically accounted for by the number $N_{i\rightarrow n}$ of post--observables at $n$ such that this case does not prevent any of the latter from continuing to propagate during $n\rightarrow f$.
\item[(b)] We discussed in sections \ref{sec_conmatch} and \ref{sec_symrole} also the possibility of a pre--constraint, that is independent of the post--constraints, yet is first class with all post--constraints (and automatically with all pre--constraints). Such a first class pre--constraint is related to an {\it a posteriori} free parameter $\mu_n$ on which the post--constraints do not impose any conditions. Hence, although this $\mu_n$ is not a proper gauge mode and will have propagated during the move $i\rightarrow n$ (as it is {\it not a priori} free), being {\it a posteriori} free, it will {\it not} continue to propagate during $n\rightarrow f$. That is, for any such first class pre--constraint, one canonical pair of the $N_{i\rightarrow n}$ propagating modes of the move $i\rightarrow n$ will cease to propagate for $n\rightarrow f$. 
(The analogous discussion holds for a post--constraint that is first class with all pre--constraints.)
\item[(c)] Pre-- and post constraints which are second class. For a second class post--constraint there is at least one pre--constraint with which this post--constraint has (on the constraint hypersurface) non--vanishing Poisson brackets. As the post--constraints generate the flow of the parameters $\lambda_n$ and the pre--constraints generate the flow of the parameters $\mu_n$, it follows that a corresponding pair of these variables $(\lambda_n,\mu_n)$ gets fixed by imposing the pre--constraint or the post--constraint, respectively. This can happen in two ways:\\
1. The free parameters are determined by the right $\frac{\partial {}^+C_n}{\partial p^n_j}$ or left null vectors $\frac{\partial {}^-C_n}{\partial p^n_j}$ of the Lagrangian two--forms. Thus, despite having ${}^+C_n \neq  {}^-C_n$ the null vectors may coincide and define the same free parameter $\lambda_n=\mu_n$. The latter is fixed to a certain value by the constraints, i.e.\ the {\it a priori} and the {\it a posteriori} orbits are intersected transversally by the constraint hypersurface. However, none of the $N_{i\rightarrow n}$ post--observables at $n$ is prevented from propagating further to $f$, because $\lambda_n=\mu_n$ is both {\it a priori} and {\it a posteriori} free (and thus did not propagate in any direction in the first place). \\
2. $\lambda_n \neq \mu_n$ for the parameters associated to post-- and pre--constraint, respectively. Again, both parameters are fixed by the constraints to certain values and do not decrease further the dimension of the reduced phase space. More precisely, the {\it a posteriori} free parameter $\mu_n$ is only free for the move $n\rightarrow f$, however, is not {\it a priori} free for the move $i\rightarrow n$ and thus propagates during the latter,\footnote{That is to say, $\mu_n$ actually gets predicted by the initial data at step $i$ which under $\ch_i$ propagates to the post--constraint surface $\cc^+_n$ and thus is automatically compatible with the post--constraints at $n$.} but does {\it not} continue to propagate during $n\rightarrow f$. The converse is true for the {\it a priori} free $\lambda_n$. Consequently, an observable pair, corresponding to $\mu_n$, propagates during $i\rightarrow n$ and does {\it not} propagate further. However, during the fixing via the constraints, the data of this pair is transferred to another observable pair, corresponding to $\lambda_n$ (which was not an observable pair for $i\rightarrow n$), that continues to propagate during $n\rightarrow f$. Thus, none of the $N_{i\rightarrow n}$ propagating data is `lost' in this case.

\end{itemize}

Let us  choose a basis of constraints which separates the first and second class constraints. The respective numbers are defined independently of the choice of this basis (e.g., since the number of first class constraints agrees with the number of degenerate directions of the pull back of the symplectic form to the constraint hypersurface).

Combining all of the above, the number of degrees of freedom that propagate from $i$ via $n$ to $f$ is
\ba
N_{i\rightarrow n\rightarrow f}\!\!\!\!&=&\!\!\!\!N_{i\rightarrow n}-2\times\#(\text{pre--constraints of case (b) at $n$})\nn\\
&=&\!\!\!\!2Q-2\times \#(\text{1st class constraints at $n$})-\#(\text{2nd class constraints at $n$})\,\,\label{dimRPS}
\ea
which thus coincides with the dimension of the reduced phase space at $n$. Note that this dimension depends on the choice of initial and final time steps $i$ and $f$.

Theorem \ref{thm_Ok} has an immediate consequence for a special kind\footnote{Notice that observables, in general, do not need to also Poisson commute with the second class constraints.} of observables:
\begin{Corollary}
Observables $O_n:=O^+_n=O^-_n$ which are both \emph{pre--observables} of the move $n\rightarrow f$ and \emph{post--observables} of the move $i\rightarrow n$ at time step $n$, Poisson commute (weakly) with \emph{all} constraints at time step $n$. Thus, the $O_n$ are elements of the space of functions which are well defined on the reduced phase space associated to time step $n$ of the evolution $i\rightarrow n\rightarrow f$.
\end{Corollary}


\begin{Example}
\emph{We consider the analogue of the `no boundary' proposal \cite{Hartle:1983ai} in Regge Calculus, as briefly discussed in \cite{Dittrich:2011ke}. Take an evolution move $0\rightarrow1$ from the empty triangulation at $n=0$ to the boundary of a single $D$--simplex, $D\geq3$, at $n=1$. The phase spaces $\cp_0,\cp_1$ corresponding to a single simplex are totally constrained \cite{Dittrich:2011ke} and, hence, the corresponding reduced phase spaces are zero--dimensional.
No information (i.e.\ `lattice gravitons') propagates during $0\rightarrow1$ (a single simplex is flat and does not contain any `gravitons').\footnote{This suggests to reconsider the interpretation of the `graviton propagator' from spin foam models which has been derived from a single 4--simplex (e.g., see \cite{Rovelli:2005yj}). The semiclassical limit of spin foam models is expected to be dominated by Regge geometries. However, as argued, a classical Regge 4--simplex does not contain `gravitons'.} When performing further simplicial evolution moves \cite{Dittrich:2011ke} to a larger spherical hypersurface at step $n$ (see figure \ref{fig_nb}) and solving the intermediate equations of motion, one finds that $\cp_n$ is still totally constrained (this follows from theorems \ref{thm_presymp} and \ref{thm_momup}). Thus, the reduced phase space at $n$ is also zero--dimensional. Indeed, no information can propagate from the empty triangulation at $0$ to any other spherical hypersurface $n$, or, colloquially, ``no information can propagate from `nothing' to `something'." }
\begin{figure}[htbp!]
\psfrag{n}{`Nothing'\,}
\psfrag{0}{ $0$}
\psfrag{1}{$1$}
\psfrag{k}{$n$}
{\includegraphics[scale=.7]{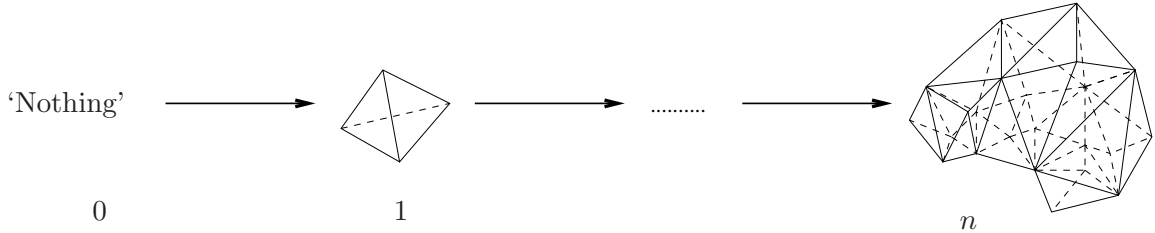}} 
\caption{\small Discrete analogue of the `no boundary' proposal \cite{Hartle:1983ai}.}\label{fig_nb}
\end{figure}

\emph{This does \emph{not} imply that all spherical 4D Regge geometries are devoid of `lattice gravitons' because the notions of observables and the reduced phase space are step dependent. Namely, 
consider, instead, a 4D evolution from a smaller (non--empty) intermediate spherical hypersurface at some $i>0$ to a larger spherical hypersurface at some $f\leq n$ (see figure \ref{fig_smallbig}). In general, the piece of 4D triangulation interpolating between these two hypersurfaces will contain curvature and will not lead to totally constrained phase spaces at $i$ and $f$, such that propagation for $i\rightarrow f$ is possible. }
\begin{SCfigure}
\psfrag{ki}{\small$i$}
\psfrag{kf}{\small$f$}
{\includegraphics[scale=.2]{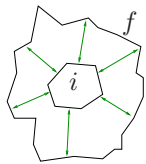}} 
\hspace*{2cm}
\caption{\small Schematic evolution from a smaller to a bigger spherical hypersurface.}\label{fig_smallbig}
\end{SCfigure}
\end{Example}

In conclusion, the pre--constraints at $i$ and the post--constraints at $f$ of a given (effective) move $i\rightarrow f$ determine what propagates from $i$ to $f$, while the conjunction of pre-- and post--constraints at the {\it same} step $n$ determines what propagates 
{\it through} $n$. The classification of the constraints into first and second class only plays a role in determining the reduced phase space at a given step and the presence of gauge symmetries. 

\subsubsection{Propagating degrees of freedom and {\it local} moves}\label{sec_obsevolslice}

So far we have only considered observables as propagating degrees of freedom under {\it global} evolution moves. Let us now also consider the {\it local} evolution of a given time slice. This is described by {\it momentum updating} as given in section \ref{sec_momup}.

By locally evolving a given time slice forward and solving any equations of motion arising on the way, one always considers the propagation of information from some initial step $k_i$ onto the evolving slice (which, in general, depends on $k_i$). That is, evolving a time slice forward is tantamount to studying the sequence of {global} moves $k_i\rightarrow k$, $k_i\rightarrow k+1$, $k_i\rightarrow k+2$, ... It is clear that the number of degrees of freedom that propagate from the fixed $k_i$ onto the evolving slice cannot increase, but at best remain constant and in general will decrease. 
This is the `physical' reason why the rank of the symplectic form restricted to the post--constraint surface of the evolving slice---which determines the number of propagating degrees of freedom---can at best remain constant, but will in general decrease 
(see theorem \ref{thm_momup}). 

On the other hand, it is also possible to {\it create} new degrees of freedom at, say, $k+1$ that only propagate from $k+1$ onwards. These new variables can be viewed as new initial data which only become relevant at step $k+1$. This highlights the fact that in general discrete systems with evolving phase spaces sufficient initial data for the entire evolution is only fully assigned in the course of evolution. 

We briefly discuss the specific role of the four types of local moves of section \ref{sec_momup}:


{\begin{description}
\item[\underline{Type I \& IV:}] These moves do {\it not} reduce the number of observables propagating from $k_i$ onto the evolving slice because they preserve the symplectic form (see theorem \ref{thm_momup}). The type I move additionally introduces new {\it a priori} free variables at step $k+1$. If these new data are not {\it a posteriori} free for some move $k+1\rightarrow k+x$, they will propagate from $k+1$ to $k+x$.
\item[\underline{Type II and III:}] These moves reduce the rank of the symplectic form (see theorem \ref{thm_momup}), and thus the number of propagating observables on the evolving slice, by two for each of their pre--constraints that is independent of the post--constraints at $k$ and is not rendered second class (case (b) of sections \ref{sec_conmatch}, \ref{sec_symrole} and \ref{sec_redps1}). 
\end{description}}

Since some new `initial data' may be {\it a priori} freely chosen at $k+1$, non--uniqueness of solutions arises. The system is non--hyperbolic, because a fixed step $k$ cannot predict/postdict the {\it entire} evolution; information propagating from it may eventually stop propagating without evolution in general breaking down.





\section{Conclusions} \label{conclusions}

The present work develops a general canonical formalism and constraint analysis for arbitrary variational discrete systems. In particular, it applies to discrete mechanical, lattice field theory and discrete gravity models with variational action principle. The formalism is equivalent to the covariant formulation, encompasses both global and local discrete time evolution moves and can handle both constant and evolving phase spaces. We have discussed the (non--) preservation of the constraints and symplectic structure in this framework and analysed the different roles constraints can assume in the discrete. In particular, constraints can `propagate' on solutions such that the number of constraints at a fixed step $n$ depends on the initial and final steps of evolution. We have shown that the sets of pre-- and post--constraints at any step each form a first class Poisson--sub--algebra, and can only become second class when both sets are considered together. Gauge symmetry generating constraints are simultaneously pre-- and post--constraints and thus necessarily first class. However, first class constraints can also arise which, in contrast to the continuum \cite{Dirac,Henneaux:1992ig}, do {\it not} generate symmetries but restrict the (physical) solution space. In analogy to the continuum, second class constraints lead to a fixing of free parameters.

We have defined observables as propagating and gauge modes as non--propagating degrees of freedom. It is established that, in general variational discrete systems where the phase spaces vary from step to step, the meaning of a propagating degree of freedom depends crucially on the initial and final step between which one considers the time evolution. Nevertheless, the definition of pre-- and post--observables is unambiguous and the general canonical formalism describing their propagation fully consistent. More precisely:
\begin{itemize}
\item The notion of an observable {\it at} a single step $n$ is in general not useful.
\item The number $N_{i\rightarrow f}$ of propagating degrees of freedom of the evolution move $i\rightarrow f$ depends on the initial and final time steps $i$ and $f$ and is determined by the pre--constraints at $i$ and the post--constraints at $f$: the $N_{i\rightarrow f}$ pre--observables at $i$ Poisson commute (weakly) with the pre--constraints at $i$ and the $N_{i\rightarrow f}$ post--observables at $f$ Poisson commute (weakly) with the post--constraints at $f$. In general, for different $i$ and $f$ one has different numbers of propagating degrees of freedom. Thus, the number of observables generally varies. 
\item The reduced phase space at a step $n$ in an evolution $i\rightarrow n\rightarrow f$ depends on $i$ and $f$ and coincides with the space of observables that propagate from $i$ {\it through} $n$ further to $f$. In particular, a totally constrained phase space at step $n$ does {\it not} imply the absence of observables, but only the absence of observables propagating {\it through} $n$.
\item A meaningful counting of degrees of freedom is provided which, e.g., entails that the number of propagating degrees of freedom from a fixed initial step onto an evolving slice can only remain constant or decrease. 
\end{itemize}
For translation invariant systems, on the other hand, the determination of the space of observables reduces to the well known procedure from the continuum.

Systems with evolving phase spaces are generally non--hyperbolic, because a fixed step $n$ cannot predict the {\it entire} future; information propagating from it (in both past and future) may eventually stop propagating without evolution in general breaking down. For instance, new `initial data' introduced during a local move $k\rightarrow k+1$ can propagate from $k+1$ onwards. Since some new `initial data' may be {\it a priori} freely chosen at $k+1$, non--uniqueness of solutions arises. 


In this article, we have so far worked with phase space extensions to suitably formulate systems with evolving phase spaces. It should be noted, however, that one may equally well define and work on a total phase space $\cp^{\rm tot}$, corresponding to all variables that ever become relevant in the {\it entire} discrete solution. For instance, in a Regge triangulation, $\cp^{\rm tot}$ is the phase space associated to the entire space--time (and not hypersurface) triangulation of a given covariant solution. The dynamics on $\cp^{\rm tot}$ proceeds then 
as follows: any variables that, at a given step $n$, are not relevant (either because they have not yet appeared in the evolution or have become `bulk') are `switched off' by the corresponding canonical constraints $p^n=0$ and only the relevant variables at a given step are `excited'. Their dynamics then simply proceeds by momentum updating. From this larger picture it is straightforward to return to the phase spaces $\cp_n$ that we used in the present article at step $n$: one only has to perform a phase space reduction on $\cp^{\rm tot}$ using the first and second class constraints that appear at $n$. 

Aspects of the present general formalism allowed us in \cite{Dittrich:2011ke} to construct the canonical formulation of Regge Calculus, based on gluing single simplices at each time step onto the spatial hypersurfaces. This provides the attractive intuitive picture of the Regge dynamics as an evolution of the triangulated hypersurfaces in a discrete `multi--fingered' time through the full Regge solution, akin to the evolution of spatial hypersurfaces in canonical General Relativity. It also constitutes an algorithm to generate solutions, given a suitable set of initial data.

As an open question for future research we would like to mention the issue of broken diffeomorphism symmetry for discrete gravity \cite{Bahr:2009ku,Dittrich:2009fb}. This implies that modes of the metric which are gauge in the continuum will actually be classified as propagating degrees of freedom in the discrete (if one considers a curved background solution). Thus, we cannot claim that all the propagating degrees of freedom (even on a regular lattice) necessarily map to the graviton (i.e.\ spin 2 ) modes of the continuum. The additional modes have been termed pseudo gauge modes \cite{Bahr:2009ku, Dittrich:2009fb}, as one expects their dynamics to differ from the `true graviton modes'. In particular, the correlations associated to these modes could be very short ranged compared to the correlations of the true graviton modes. The tools provided in this work allow us to identify triangulated hypersurfaces between which observables can propagate in principle. In the case of gravity these observables will include both gravitons and pseudo gauge modes. As discussed, there are also hypersurfaces between which propagation cannot occur. It would be of great benefit to develop a systematic criterion to distinguish pseudo gauge modes from modes that are truly propagating also in the continuum limit. This certainly would improve the comparison of lattice results for the graviton propagator \cite{Rovelli:2005yj} to the continuum.  Promisingly, the present formalism is amenable to coarse graining techniques under which pseudo gauge modes can be expected to average out.

Finally, there are many interpretational questions regarding a possible quantization of this framework. Hilbert spaces associated to phase spaces of different dimension have to be mapped to each other under time evolution. In this context, the usual unitarity requirement for time evolution has to be replaced. In particular, the concept of dynamical cylindrical consistency \cite{Dittrich:2012jq} becomes essential, as it connects measures on Hilbert spaces based on coarser and finer discretizations. The `general boundary formalism' \cite{Oeckl:2003vu,Oeckl:2005bv}, on the other hand, offers useful tools for devising a meaningful probability interpretation for such evolving Hilbert spaces. Research in this direction is currently under way \cite{phevolH}.





\appendix

\section{Relation between effective canonical and Lagrangian two--forms}\label{app_bulk}

In this appendix we shall prove equation (\ref{omeff}) of section \ref{sec_bulk}. For simplicity and brevity, we shall assume that boundary data constraints are absent. We will only show this for $\tilde{\Omega}_n$ under the post--Legendre transform as the proof for $\tilde{\Omega}_{n+1}$ under the pre--Legendre transform proceeds identically.

Prior to imposing the equations of motion the Lagrangian two--form (\ref{app6}) on ${\cq}_{n-1}\times{\cq}_n$ reads
\ba\label{lag2form}
\Omega_n=-\f{\p^2S_n}{\p x^e_{n-1}\p x^{e'}_n}dx^e_{n-1}\,\w\,dx^{e'}_n-\f{\p^2S_n}{\p x^e_{n-1}\p x^t_n}dx^e_{n-1}\,\w\,dx^t_n.
\ea
Differentiation of 
\ba
\frac{\partial S_n(x^e_{n-1},x^t_n,x^{e'}_n)}{\partial x^t_n}\left|_{\chi^t_n(x^e_{n-1},x^{e'}_n,\kappa^m)}\right.=0.\nn
\ea
yields
\ba
\frac{\partial^2\, S_n}{\partial x^t_n\partial x^e_n}+\frac{\partial^2\, S_n}{\partial x^t_n\partial x^{t'}_n}\frac{\partial \chi^{t'}_n}{\partial x^e_n}=0,\q\q \q \frac{\partial^2\, S_n}{\partial x^t_n\partial x^e_{n-1}}+\frac{\partial^2\, S_n}{\partial x^t_n\partial x^{t'}_n}\frac{\partial \chi^{t'}_n}{\partial x^e_{n-1}}=0,
\ea
which can be employed in
\ba
dx^t_n=\frac{\partial \chi^t_n}{\partial x^{e'}_n}dx^{e'}_n+\frac{\partial \chi^t_n}{\partial x^e_{n-1}}dx^e_{n-1}.\nn
\ea
By the conjunction of the above and assuming $\frac{\partial^2\, S_n}{\partial x^t_n\partial x^{t'}_n}$ is invertible,\footnote{If this matrix is degenerate, one can factor out the degenerate directions and invert the resulting matrix, which, for simplicity, we shall not worry about here.} one obtains from (\ref{lag2form}) the effective Lagrangian two--form on $\tilde{\cq}_{n-1}\times\tilde{\cq}_n$
\ba\label{effomega}
\tilde{\Omega}_n&=&-\left(\frac{\partial^2\,S_n}{\partial x^{e'}_n\partial x^e_{n-1}}-\frac{\partial^2\,S_n}{\partial x^{e'}_n\partial x^t_n}\left(\frac{\partial^2\,S_n}{\partial x^t_n\partial x^{t'}_n}\right)^{-1}\frac{\partial^2\,S_n}{\partial x^{t'}_n\partial x^e_{n-1}}\right)\,dx^e_{n-1}\,\wedge\, dx^{e'}_n\nn\\
&=&-\frac{\partial^2\,\tilde{S}_n}{\partial x^e_{n-1}\partial x^{e'}_n}dx^e_{n-1}\,\wedge\, dx^{e'}_n.
\ea

On the other hand, the symplectic form on $T^*\cq_n$ reads
\ba\label{wn}
\omega_n=dx^e_n\,\w\,dp^n_e+dx^t_n\,\w\,dp^n_t.
\ea

Using the embedding $\iota_n:(x^e_n,\tilde{p}^n_e)\mapsto(x^e_n,p^n_e=\tilde{p}^n_e,x^t_n=\chi^t_n,p^n_t=0)$ of $T^*\tilde{\cq}_n$ into $T^*\cq_n$,\footnote{Viewing the $x^e_{n-1}$ in this case as `external' parameters.} we can pull back the symplectic form (\ref{wn}) to obtain, 
\ba\label{presymp}
\tilde{\omega}_n:=\iota_n^*\omega_n=dx^e_n\,\w\,d\tilde{p}^n_e\,,
\ea
the symplectic form on the `reduced phase space' $\tilde{\cp}_n:=T^*\tilde{\cq}_n$. 
(Note that additional constraints $C(x^e_n,\tilde{p}^n_{e'})=0$ may occur on $\tilde{\cp}_n$.) 
Direct computation shows that on--shell one, indeed, finds
\ba
\tilde{\Omega}_n=(\tilde{\mathbb{F}}^+\tilde{S}_n)^*\tilde{\omega}_n.
\ea
Analogously, one shows that $\tilde{\Omega}_{n+1}=(\tilde{\mathbb{F}}^-\tilde{S}_n)^*\tilde{\omega}_n$.

\section{Momentum updating as a canonical transformation}\label{sec_momupcan}

In section \ref{sec_momup} we mentioned that it is also possible to formulate {\it momentum updating} alternatively as a canonical transformation. Let us briefly elaborate on this. Since {\it momentum updating} applies to {\it all} canonical pairs on the extended phase space equally, we can extend the corresponding time evolution map beyond the pre-- or post--constraint surfaces defined by $p^k_n=0$ or $p^{k+1}_o=0$ to the full extended phase space without specifying the momenta at $k$ beforehand. That is, the extended time evolution map $\bar{\ch}_{k}$ defined on the full phase space at $k$ reads {\it for all} variable pairs $i=b,e,n,o$
\ba\label{anh23}
x^i_{k+1}=x^i_{k} \, ,\q\q\q p^{k+1}_i= p^k_i+ \frac{\partial S_{k+1}}{\partial x^i_{k+1}} \,.
\ea
The gauge fixing conditions $x^n_k=x^n_{k+1}$ and $x^o_k=x^o_{k+1}$ (which were a free choice in the previous formulation of momentum updating) are thus built in from the outset. The image of (\ref{anh23}) is given by the full extended phase space at time $k+1$. The {\it extended momentum updating} (\ref{anh23}), which involves a generating function of the second kind
\ba
G(x_k,p^{k+1})=\sum_i x^i_kp_i^{k+1}-S_{k+1}\,,\nn
\ea
preserves the canonical two--form (of the extended phase space). Namely, by symmetry and antisymmetry and using $\omega_{k+1}=dx^i_{k+1}\wedge dp^{k+1}_i$, one immediately verifies 
\ba
(\bar{\ch}_{k})^*\omega_{k+1}=dx^i_k\wedge d\left(p^k_i+\frac{\partial S_{k+1}}{\partial x^i_k}\right)=dx^i_k\wedge dp^k_i+\frac{\partial^2S_{k+1}}{\partial\, x^i_k\partial\,x^{i'}_k}dx^i_k\wedge dx^{i'}_k=dx^i_k\wedge dp^k_i=\omega_k.\nn
\ea

Subsequently, one must impose the pre-- and post--constraints $p_n^k=0$ and $p^{k+1}_o=0$, respectively, by hand. Note that, as a consequence of $x^i_k=x^i_{k+1}$, this leads to the post--constraints ${}^+C^{k+1}_n=p^{k+1}_n-\frac{\partial S_{k+1}}{\partial x^n_{k+1}}$ (which are the image of the pre--constraints $p^k_n=0$ under (\ref{anh23})) and the pre--constraints ${}^-C^k_o=p^k_o+\frac{\partial S_{k+1}}{\partial x^o_{k}}$ (which are the pre--image of $p^{k+1}_o=0$) on the extended phase space. However, the latter pre-- and post--constraints need not in general be constraints on the unextended phase spaces due to the general dependence of the action on both $x^n_{k+1}$ and $x^o_k$. In this case, the pre--constraints $p_n^k=0$ and ${}^-C^k_o$ and, likewise, the post--constraints $p^{k+1}_o=0$ and ${}^+C^{k+1}_n$, in fact, are generally second class. Those ${}^-C^k_o$ and ${}^+C^{k+1}_n$ which are only constraints on the extended, but not on the unextended phase space (originating, after all, in the general condition $x^i_k=x^i_{k+1}$) can be viewed as gauge fixing conditions such that not only $x^n_k=x^n_{k+1}$ and $x^o_k=x^o_{k+1}$, but also $x^n_{k+1}$ and $x^o_k$ are no longer free parameters. 

This implies preservation under time evolution of the induced symplectic forms on the combined gauge fixing and constraint hypersurfaces. On these gauge fixing and constraint surfaces, the extended momentum updating is equivalent to the previous version of momentum updating which was only formulated as a pre--symplectic transformation.

\section{Proof of theorem \ref{thm_momup} of section \ref{sec_momup}}\label{app_momup}

\begin{proof}
Consider \underline{\bf type I}. Take the symplectic forms at $k$ and $k+1$ on $\bar{\cp}_k,\bar{\cp}_{k+1}$
\be
\omega_k=dx^b_k \wedge dp_b^k + dx^e_k \wedge dp_e^k +dx^n_k \wedge dp_n^k  \, ,  \,\,\, \omega_{k+1}=dx^b_{k+1} \wedge dp_b^{k+1} + dx^e_{k+1} \wedge dp_e^{k+1} +dx^n_{k+1} \wedge dp_n^{k+1}.\nn
\ee
Notice that $p^k_n=0$ is both a pre-- and post--constraint at $k$. Restrict the symplectic forms to the partial post--constraint surfaces defined only by the $K$ $p^k_n=0$ at $k$ and only the $K$ post--constraints ${}^+C^{k+1}_n$ on the right in (\ref{anh1c}) at $k+1$ and denote them by $\ck^+_k,\ck^+_{k+1}$. {$\ck^+_{k+1}$ is the image of $\ck^+_k$ under $\fh_k$ and vice versa.}
At time $k$ employ the canonical embedding $\jmath_k:\ck^+_k\hookrightarrow\bar{\cp}_k$ 
\ba
\jmath_k:(x^b_k,p_b^k,x^e_k,p_e^k,x^n_k)\mapsto (x^b_k,p_b^k,x^e_k,p_e^k,x^n_k,p_n^k=0).\nn
\ea
The pull--back of the symplectic form to $\ck^+_k$ coincides with the symplectic form of the unextended phase space $\cp_k$,
\ba\label{type1sym1}
(\jmath_k)^*\omega_k&=& dx^b_k \wedge dp_b^k + dx^e_k \wedge dp_e^k  \, .
\ea
For time step $k+1$, consider the embedding $\jmath_{k+1}:\ck^+_{k+1}\hookrightarrow\cp_{k+1}$ given by
\be
\jmath_{k+1}:(x^b_{k+1},p_b^{k+1},x^e_{k+1},p_e^{k+1},x^n_{k+1})\mapsto (x^b_{k+1},p_b^{k+1},x^e_{k+1},p_e^{k+1},x^n_{k+1}, p^n_{k+1}=\frac{\partial S_{k+1}}{\partial x^n_{k+1}}) \, .\nn
\ee
The corresponding pull--back of the symplectic form to $\ck^+_{k+1}$ reads
\ba\label{sy01}
(\jmath_{k+1})^*\omega_{k+1} &=&  dx^b_{k+1} \wedge dp_b^{k+1} + dx^e_{k+1} \wedge dp_e^{k+1} +dx^n_{k+1} \wedge \frac{\partial^2 S_{k+1}}{\partial x^n_{k+1} \partial x^e_{k+1}}dx^e_{k+1}\nn\\
&=&dx^b_{k+1} \wedge dp_b^{k+1} + dx^e_{k+1} \wedge d\left(p_e^{k+1} - \frac{\partial S_{k+1}}{ \partial x^e_{k+1}}\right)  \, .
\ea
Using the coordinate form (\ref{anh1}--\ref{anh1c}) of $\fh_k$ to pull back the pre--symplectic form (\ref{sy01}) at time $k+1$, we find
\ba\label{sy01b}
{\fh}_k^*(\jmath_{k+1})^*\omega_{k+1} &=&dx^b_{k} \wedge dp_b^{k} + dx^e_{k} \wedge dp_e^{k} =(\jmath_k)^*\omega_k\,. \nn
\ea
{Any post--constraint that further reduces the rank of (\ref{sy01}) must be of the form\footnote{Here and in the sequel an equation such as $p^k={}^+p^k$ represents a shorthand notation for equation (\ref{b3}) $p^k={}^+p^k=+\f{\p \bar{S}_k(x_k,x_0)}{\p x_k}$ for the post--momenta, where $\bar{S}_k$ is the (possibly effective) action for the (global) move $0\rightarrow k$ for some (here irrelevant) initial step $0$. The analogous shorthand notation is used for pre--momenta.}
\ba
{}^+C^{k+1}\left(x^e_{k+1},x^b_{k+1},p_e^{k+1} - \frac{\partial S_{k+1}}{ \partial x^e_{k+1}},p^{k+1}_b\right)\Big|_{p^{k+1}={}^+p^{k+1}}=0.\nn
\ea
By (\ref{anh1}--\ref{anh1c}), each of these pulls back one-to-one to a post--constraint at $k$, 
\ba
\fh_k^*\,{}^+C^{k+1}={}^+C^k(x^e_k,x^b_k,p^k_e,p^k_b)\Big|_{p^k={}^+p^k}=0.\nn 
\ea
Hence, $\fh_k$ preserves the total post--constraint surfaces, $\fh_k:\cc^+_k\rightarrow\cc^+_{k+1}$ and $(\iota_k)^*\omega_k=\fh_k^*(\iota_{k+1})^*\omega_{k+1}$, where $\iota_{k}:\cc^+_k\hookrightarrow\bar{\cp}_k$ and $\iota_{k+1}:\cc^+_{k+1}\hookrightarrow\bar{\cp}_{k+1}$ are the embeddings of the post--constraint surfaces at $k$ and $k+1$, respectively, into the extended phase space.}

Next, consider \underline{\bf type II}. The proof proceeds analogously by replacing the label $n$ by $o$, but taking into account the presence of non--trivial pre--constraints ${}^-C^k_o$. Denote by $\ck^-_k$ the partial pre--constraint surface at step $k$ defined only by the $K$ pre--constraints ${}^-C^k_o$ in (\ref{case2c}). {No other pre--constraints are created in the move.} Denote by $\ck^+_{k+1}$ the partial (post--)constraint surface at step $k$ defined only by the $K$ (post--)constraints $p^{k+1}_o=0$ which are the image of the ${}^-C^k_o$ under $\fh_k$. Let the corresponding embeddings be $\jmath_k:\ck^-_k\hookrightarrow\bar{\cp}_k$ and $\jmath_{k+1}:\ck^+_{k+1}\hookrightarrow\bar{\cp}_{k+1}$. In analogy to (\ref{type1sym1}, \ref{sy01}), one obtains
\ba
(\jmath_k)^*\omega_k&=&dx^b_k \wedge dp_b^k + dx^e_k \wedge d\left(p_e^k +\frac{\p S_{k+1}(x^{e'}_k,x^o_k)}{\p x^e_k}\right) \,,\label{type2k}\\
 (\jmath_{k+1})^*\omega_{k+1}&=&dx^b_{k+1} \wedge dp_b^{k+1} + dx^e_{k+1} \wedge dp_e^{k+1}\,,\label{type2k1}
 \ea
and, consequently, by (\ref{case2}--\ref{case2c}), $\fh_k^* (\jmath_{k+1})^*\omega_{k+1}=(\jmath_k)^*\omega_k$. 

{Any additional post--constraint at $k+1$ that further reduces the rank of $(\jmath_{k+1})^*\omega_{k+1}$ must be of the form ${}^+C^{k+1}(x^b_{k+1},x^e_{k+1},p^{k+1}_b,p_e^{k+1})\Big|_{p^{k+1}={}^+p^{k+1}}=0$. Using (\ref{case2}--\ref{case2c}), it is pulled back one-to-one to a post--constraint at $k$
\ba
{}^+C^k\left(x^b_k,x^e_k,p^k_b,p_e^k +\frac{\p S_{k+1}(x^{e'}_k,x^o_k)}{\p x^e_k}\right)\Big|_{p^k={}^+p^k}=\fh^*_k\,{}^+C^{k+1}(x^b_{k+1},x^e_{k+1},p^{k+1}_b,p_e^{k+1})\Big|_{p^{k+1}={}^+p^{k+1}}=0.\nn
\ea
which likewise reduces the rank of $(\jmath_{k})^*\omega_{k}$.} Hence, $\fh_k$ preserves all independent rank reducing constraints. Thus, we must have $\fh_k:\cc^+_k\cap\ck^-_k\rightarrow\cc^+_{k+1}$ and $(\iota_{k})^*\omega_{k}={\fh}_k^*(\iota_{k+1})^*\omega_{k+1}$ with the corresponding embeddings $\iota_k:\cc^+_k\cap\ck^-_k\hookrightarrow\bar{\cp}_k$ and $\iota_{k+1}:\cc^+_{k+1}\hookrightarrow\bar{\cp}_{k+1}$, where $\cc^+_k,\cc^+_{k+1}$ are the total post--constraint surfaces at $k$ and $k+1$, respectively.

{Consider \underline{\bf type III}. Again, $p^k_n=0$ constitute both pre-- and post--constraints at $k$. Denote by $\ck^+_k$ the partial post--constraint surface at $k$ defined by the $K$ $p^k_n=0$ only and by $\jmath_k$ its embedding in $\bar{\cp}_k$. Similarly, denote by $\ck^+_{k+1}$ the partial post--constraint surface at $k+1$ defined by the $K$ $p^{k+1}_o=0$ only and by $\jmath_{k+1}$ its embedding in $\bar{\cp}_{k+1}$. This yields
\ba
(\jmath_k)^*\omega_k&=&dx^b_k\w dp^k_b+dx^e_k\w dp^k_e+dx^o_k\w dp^k_o\,,\nn\\
(\jmath_{k+1})^*\omega_{k+1}&=&dx^b_{k+1}\w dp^{k+1}_b+dx^e_{k+1}\w dp^{k+1}_e+dx^n_{k+1}\w dp^{k+1}_n\,,\nn
\ea
which coincides with the canonical forms of the unextended phase spaces $\cp_k,\cp_{k+1}$. Furthermore, using (\ref{anh21}--\ref{anh21d}), one finds 
\ba
\fh_k^* (\jmath_{k+1})^*\omega_{k+1}&=&dx^b_k\w dp^k_b+dx^e_k\w d\left(p^k_e+\f{\p S_{k+1}}{\p x^e_k}\right)+dx^n_{k+1}\w d\left(\f{\p S_{k+1}}{\p x^n_{k+1}}\right)\,,\label{tiii}\\
&=&dx^b_k\w dp^k_b+dx^e_k\w dp^k_e+dx^o_k\w d\left(-\f{\p S_{k+1}}{\p x^o_k}\right)=(\jmath_k)^*\omega_k.\nn
\ea 
If the rank of $\f{\p^2S_{k+1}}{\p x^o_k\p x^n_{k+1}}$ is $K-\kappa$, there will arise $\kappa$ post--constraints 
\ba
{}^+C^{k+1}(x^e_{k+1},x^n_{k+1},p^{k+1}_n)\Big|_{p^{k+1}_n={}^+p^{k+1}_n=\f{\p S_{k+1}}{\p x^n_{k+1}}}=0\label{tiii2}
\ea
from the second equation in (\ref{anh21c}) which further reduce the rank of $(\jmath_{k+1})^*\omega_{k+1}$ by $2\kappa$. The pullback of (\ref{tiii2}) under (\ref{anh21}--\ref{anh21d}) yields the identities $\fh_k^*{}^+C^{k+1}={}^+C^{k+1}(x^e_k,x^n_{k+1},\f{\p S_{k+1}}{\p x^n_{k+1}})=0$, i.e.\ trivializes the constraints at $k$ and thus does {\it not} reduce the rank of $(\jmath_k)^*\omega_k$. However, in complete analogy the $\kappa$ pre--constraints
\ba
{}^-C^k(x^e_k,x^o_k, p^k_o)\Big|_{p^k_o={}^-p^k_o=-\f{\p S_{k+1}}{\p x^o_k}}=0\nn
\ea
which arise from the second equation in (\ref{anh21d}) and define $\ck^-_k$ reduce the rank of $(\jmath_k)^*\omega_k$ by $2\kappa$ and are promoted to trivial identities at $k+1$.

Any additional post--constraint ${}^+C^{k+1}(x^b_{k+1},x^e_{k+1},x^n_{k+1},p^{k+1}_b,p^{k+1}_e,p^{k+1}_n)\Big|_{p^{k+1}={}^+p^{k+1}}=0$ at $k+1$
that further reduces the rank of $(\jmath_{k+1})^*\omega_{k+1}$ is pulled back, employing (\ref{anh21}--\ref{anh21d}), to a non--trivial post--constraint 
$
{}^+C^k(x^b_{k},x^e_{k},x^n_{k+1},p^{k}_b,p^{k}_e+\f{\p S_{k+1}}{\p x^e_k},\f{\p S_{k+1}}{\p x^n_{k+1}})\Big|_{p^{k}={}^+p^{k}}=0$
at $k$ which, by virtue of (\ref{tiii}), likewise reduces the rank of $(\jmath_k)^*\omega_k$. In complete analogy, one verifies that also the reverse holds true such that $\fh_k$ defines a one-to-one mapping between the additional post--constraints at $k$ and $k+1$. The conjunction of the above implies $\fh_k^*(\iota_{k+1})^*\omega_{k+1}=(\iota_k)^*\omega_k$, where $\iota_{k}:\cc^+_k\cap\ck^-_k\hookrightarrow\bar{\cp}_k$ and $\iota_{k+1}:\cc^+_{k+1}\hookrightarrow\bar{\cp}_{k+1}$ are the corresponding embeddings.}


%

{Finally, the momentum updating map $\fh_k$ of \underline{\bf type IV}, given by (\ref{typeiv1}, \ref{typeiv2}), preserves the symplectic structure because no new pre-- or post--constraints arise in the move and any already existing post--constraints at $k$ are trivially mapped one-to-one to post--constraints at $k+1$.}
\end{proof}

\section*{\small Acknowledgements}
Research at Perimeter Institute is supported by the Government of Canada through Industry Canada and by the Province of Ontario through the Ministry of Research and Innovation.

\providecommand{\href}[2]{#2}\begingroup\raggedright\endgroup

\end{document}